%% file: 00-main.tex
\documentclass[a4paper,UKenglish,cleveref,autoref,thm-restate]{lipics-v2021}

\graphicspath{{./figures/}}

\input{preamble}
\title{Sliding Cubes in Parallel}
\titlerunning{Sliding Cubes in Parallel}

\author{Hugo A. Akitaya}{Miner School of Computer and Information Sciences, University of Massachusetts Lowell, USA}{hugo_akitaya@uml.edu}{https://orcid.org/0000-0002-6827-2200}{}
\author{Joseph Dorfer}{Institute of Algorithms and Theory, Graz University of Technology, Austria}{joseph.dorfer@tugraz.at}{https://orcid.org/0009-0004-9276-7870}{Austrian Science Fund (FWF) 10.55776/DOC183.}
\author{Peter Kramer}{Department of Computer Science, TU Braunschweig, Germany}{kramer@ibr.cs.tu-bs.de}{https://orcid.org/0000-0001-9635-5890}{Fellowship of the German Academic Exchange Service~(DAAD).}
\author{Christian Rieck}{Institute of Mathematics, University of Kassel, Germany}{christian.rieck@mathematik.uni-kassel.de}{https://orcid.org/0000-0003-0846-5163}{Funded by the Deutsche Forschungsgemeinschaft (DFG) -- 522790373.}
\author{Gabriel Shahrouzi}{Miner School of Computer and Information Sciences, University of Massachusetts Lowell, USA}{gabriel_shahrouzi@student.uml.edu}{https://orcid.org/0009-0004-2858-7322}{}
\author{Frederick Stock}{Miner School of Computer and Information Sciences, University of Massachusetts Lowell, USA}{frederick_stock@student.uml.edu}{https://orcid.org/0009-0008-9005-6855}{}

\authorrunning{H.\ A.\ Akitaya, J.\ Dorfer, P.\ Kramer, C.\ Rieck, G.\ Shahrouzi, and F.\ Stock}

\Copyright{Hugo A.\ Akitaya, Joseph Dorfer, Peter Kramer, Christian Rieck, Gabriel Shahrouzi, and Frederick Stock}

\ccsdesc{Theory of computation~Computational geometry}

\keywords{Sliding squares, parallel motion, reconfigurability, three dimensions, constant makespan, log-APX hardness, NP-hardness, worst-case optimality}

\funding{%
    \emph{Hugo A. Akitaya, Gabriel Shahrouzi, and Frederick Stock}: Supported by the National Science Foundation (NSF), grant CCF-2348067.
}

\nolinenumbers

\hideLIPIcs

\begin{document}

    \maketitle

    \begin{abstract}
        We study the classic \newterm{sliding cube model} for programmable matter under \newterm{parallel} reconfiguration in three dimensions, providing novel algorithmic and surprising complexity results in addition to generalizing the best known bounds from two to three dimensions.
        In general, the problem asks for reconfiguration sequences between two connected configurations of $n$ indistinguishable unit cube \newterm{modules} under connectivity constraints; a connected \newterm{backbone} must exist at all times.
		The \newterm{makespan} of a reconfiguration sequence is the number of parallel moves performed.

        We show that deciding the existence of such a sequence is \NP-hard, even for constant makespan and if the two input configurations have constant-size symmetric difference, solving an open question in~\cite{a.akitaya_et_al:LIPIcs.ESA.2025.28}.
        In particular, deciding whether the optimal makespan is 1 or 2 is \mbox{\NP-hard.}
        We also show \logAPX-hardness of the problem in sequential and parallel models, strengthening the \APX-hardness claim in~\cite{akitaya.demaine.korman.ea2022compacting-squares}.
        Finally, we outline an asymptotically worst-case optimal input-sensitive algorithm for reconfiguration.
        The produced sequence has length that depends on the bounding box of the input configurations which, in the worst case, results in a $\mathcal{O}(n)$~makespan.
    \end{abstract}

    \input{01-introduction}    \input{02-preliminaries}%
    \input{03-complexity}%
    \input{04-algorithm}    \input{05-conclusion}
    \bibliography{references}
    \newpage
    \appendix
    \input{A03-complexity}
    \input{A04-teleport}
    \input{A05-algorithm}
    \input{A05_1-gather}
    \input{A05_4-scaled-compact}
\end{document}

%% file: preamble.tex
\DeclareGraphicsExtensions{.jpg,.JPG,.jpeg,.pdf,.PDF,.JPEG,.png,.PNG}
\usepackage{amsfonts}
\usepackage{mathtools}
\usepackage{nicefrac}
\usepackage{complexity}
\usepackage{xspace}
\usepackage[ruled]{algorithm}
\usepackage{algpseudocode}

\DeclarePairedDelimiter\floor{\lfloor}{\rfloor}
\DeclarePairedDelimiter\norm{\lVert}{\rVert}
\DeclarePairedDelimiter\abs{\lvert}{\rvert}

\let\xPTAS\PTAS
\renewcommand{\PTAS}{\xPTAS\xspace}
\let\xFPT\FPT
\renewcommand{\FPT}{\xFPT\xspace}
\let\xNP\NP
\renewcommand{\NP}{\xNP\xspace}
\let\xPSPACE\PSPACE
\renewcommand{\PSPACE}{\xPSPACE\xspace}
\let\xXP\XP
\renewcommand{\XP}{\xXP\xspace}
\let\xFPTAS\FPTAS
\renewcommand{\FPTAS}{\xFPTAS\xspace}

\newcommand{\logAPX}{{\textup{\textsf{log-APX}}}\xspace}
\newcommand{\newterm}[1]{\textbf{\textit{#1}}}
\newcommand{\module}{\textsf}
\newcommand{\selector}{\textup{\textsf{sel}}\xspace}
\newcommand{\area}{\mathbb{A}\xspace}

\newcommand{\parallelsquares}{\textsc{Parallel Sliding Squares}\xspace}
\newcommand{\sequentialsquares}{\textsc{Sequential Sliding Squares}\xspace}
\newcommand{\slidingcubes}{\textsc{Sliding Cubes}\xspace}
\newcommand{\parallelcubes}{\textsc{Parallel Sliding Cubes}\xspace}
\newcommand{\sequentialcubes}{\textsc{Sequential Sliding Cubes}\xspace}

\crefname{claim}{Claim}{Claims}

%% file: 01-introduction.tex
\section{Introduction}
\label{sec:introduction}

Programmable matter systems, composed of large numbers of simple modules capable of self-reconfiguration, have attracted significant attention for applications in robotics, manufacturing, and distributed computing. 
From a Computational Geometry perspective, a fundamental algorithmic challenge is to transform one connected configuration of modules into another while preserving connectivity and avoiding collisions. 
This problem has been studied extensively~\cite{
	abel.akitaya.kominers.ea2024universal-in-place,
	akitaya.demaine.korman.ea2022compacting-squares,
	a.akitaya_et_al:LIPIcs.ESA.2025.28,
	AkitayaCCCG2025,
	dumitrescu.pach2006pushing-squares,
	dumitrescu.suzuki.yamashita2004motion-planning,
	hurtado.molina.ramaswami.ea2015distributed-reconfiguration,
	kostitsyna.ophelders.parada.ea2024optimal-in-place,
	michail.skretas.spirakis2019on-transformation,
	moreno.sacristan2020reconfiguring-sliding}, introduced by Fitch, Butler, and~Rus~\cite{fitch.butler.rus2003reconfiguration-planning}, where modules move via discrete slides or convex transitions on a square (or hyper-cube) lattice, as visualized in~\cref{fig:intro-moves}.

\begin{figure}[htb]
    \hfil%
    \includegraphics[page=1]{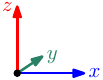}%
    \hfil%
    \includegraphics[page=1]{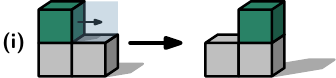}%
    \hfil%
    \includegraphics[page=2]{intro-moves}%
    \caption{Two types of legal move, (i) slide and (ii) convex transition.}
    \label{fig:intro-moves}
\end{figure}%

Early work focused on sequential models, where only one module moves at a time. 
When instances are two or three-dimensional, we refer to this model as the \sequentialsquares or \sequentialcubes model, respectively.
Akitaya et al.~\cite{akitaya.demaine.korman.ea2022compacting-squares} studied input-sensitive algorithms for \sequentialsquares, providing an $\mathcal{O}(nP)$ bound for compacting a configuration of $n$ modules, where $P$ is the bounding box perimeter. 
Their algorithm is also in-place, meaning that all but one module are constrained to be within the bounding boxes of the initial configurations at all times.
They also show that the decision version of the makespan minimization problem is \NP-complete via a reduction from \textsc{Planar Monotone 3Sat}.
The produced instance is a pair of configurations whose symmetric difference is linear in the number of variables.
The authors also claim that the same reduction can be adapted to show \APX-hardness to approximate the minimum makespan for instances in three dimensions, but they do not give a proof.

More recently, Akitaya~et~al.~\cite{a.akitaya_et_al:LIPIcs.ESA.2025.28} introduced a \parallelsquares model in two dimensions, showing that parallel motion can dramatically reduce the makespan compared to sequential schedules.
Their algorithm achieves worst-case optimal makespan~$\mathcal{O}(P)$, where~$P$ is the bounding box perimeter.
They proved \NP-completeness for deciding whether a reconfiguration can be completed in makespan~$1$ (unlabeled case) and makespan~$2$ (labeled~case). 
This rules out an \FPT algorithm parameterized by the makespan.
However, as in~\cite{akitaya.demaine.korman.ea2022compacting-squares}, their reduction produces instances having linear-size symmetric difference, which led the authors to pose the question whether the makespan minimization problem remains hard if the input has constant-size symmetric difference.

For three dimensional instances, 
universal reconfiguration was first shown by Abel and Kominers, who later, with additional coauthors, improved the bounds on the number of sequential moves to~$\mathcal{O}(n^2)$~\cite{abel.akitaya.kominers.ea2024universal-in-place} and made their algorithm in-place. 
Kostitsyna et al.~\cite{kostitsyna.ophelders.parada.ea2024optimal-in-place} showed that compacting any given configuration can be done optimally in a number of moves proportional to the sum of coordinates of the cubes in the input.
Until now, no better bounds are known for universal reconfiguration when modules are allowed to move in parallel.

\subsection{Our contributions}
\label{subsec:our-contributions}

We generalize the \parallelsquares model introduced by~\cite{a.akitaya_et_al:LIPIcs.ESA.2025.28} to three dimensions (which we call {\parallelcubes}), and establish both positive and negative results.

\begin{itemize}
    \item The decision variant of makespan minimization in {\parallelcubes} is \NP-hard even to distinguish between makespans $1$ and $2$.
    Our reduction has symmetric difference between initial and final configuration of $1$, implying that no \FPT algorithm exists.
    \item The makespan minimization problem in both {\textsc{Parallel}} and \sequentialcubes is \logAPX-hard and cannot be approximated by a factor better than $\Theta(\log n)$.
    \item Schedules for universal reconfiguration of makespan $\mathcal{O}(A+h)$ can be computed in polynomial time, where $A$ is the area of the projection of the configuration onto the base of its bounding box and $h$ is the height of the bounding box.
    In the worst case this input-sensitive bound becomes $\mathcal{O}(n)$ which is worst-case optimal.
\end{itemize}

\subsection{Related Work}
\label{subsec:related-work}

The {sliding cube model}, introduced by Fitch, Butler, and Rus~\cite{fitch.butler.rus2003reconfiguration-planning}, became a central abstraction for algorithmic studies of modular reconfiguration. 
In the sequential model, universal algorithms with quadratic worst-case complexity on the number of moves are known (for squares, cubes and hyper-cubes)~\cite{abel.akitaya.kominers.ea2024universal-in-place,akitaya.demaine.korman.ea2022compacting-squares,dumitrescu.pach2006pushing-squares,kostitsyna.ophelders.parada.ea2024optimal-in-place}.
The in-place algorithms in ~\cite{abel.akitaya.kominers.ea2024universal-in-place,AkitayaCCCG2025} use scaffolding techniques (a regular substructure that provides the connectivity between the steps of the algorithm, proposed independently by~\cite{kotay.rus2000algorithms-for} and~\cite{nguyen.guibas.yim2001controlled-module}) to produce a sweep-plane that compacts the input configuration, similar to our algorithm. 

Parallel reconfiguration has received increasing attention.
In the context of distributed algorithms for two dimensions, Dumitrescu et al.~\cite{dumitrescu.suzuki.yamashita2004motion-planning} first studied the problem providing linear makespan between two $y$-monotone shapes.
Michail et al.~\cite{michail.skretas.spirakis2019on-transformation} and Wolters~\cite{wolters2024parallel-algorithms} also provided algorithms that work on special input.
Hurtado et al.~\cite{hurtado.molina.ramaswami.ea2015distributed-reconfiguration} showed that a less restrictive model achieves universality within linear makespan.
They also show how to use meta-modules, small groups of cooperating modules that behave like a single entity, to reproduce the same result in the sliding square model for configurations that are already made of meta-modules. 
Akitaya et al.~\cite{a.akitaya_et_al:LIPIcs.ESA.2025.28} used meta-modules and scaffolding techniques (similar to~\cite{abel.akitaya.kominers.ea2024universal-in-place,AkitayaCCCG2025}) to obtain a worst-case optimal algorithm for \emph{general} input.
While our algorithm roughly follows the same principles, generalizing such ideas to three dimensions is highly nontrivial and required several new structural proofs.

\subparagraph*{Approximation.}
The majority of the literature on algorithmic reconfiguration focuses on worst-case bounds since it is typically much more challenging (and more often computationally intractable) to obtain optimal solutions.
A prominent example is the flip distance between triangulations of convex point sets.
While optimal worst-case bounds are known, the complexity of the minimization problem remains a famous open problem~\cite{sleator1986rotation}.
More related to our problem, Demaine et al.~\cite{demaine.fekete.keldenich.ea2019coordinated-motion}
obtained a constant factor approximation for parallel motion planning of a \emph{swarm} of robots, where grid aligned squares can move to an adjacent grid cell (similar to our slide move), but with no connectivity constraint.
They also show that the problem does not admit an \FPTAS unless \P=\NP, but \APX-hardness is not known.
This problem has been shown \APX-hard in the sequential setting~\cite{calinescu.dumitrescu.pach2008reconfigurations-in}.

Using similar techniques to~\cite{demaine.fekete.keldenich.ea2019coordinated-motion}, Fekete et al.~\cite{fekete.keldenich.kosfeld.ea2023connected-coordinated,fekete.kramer.rieck.ea2024efficiently-reconfiguring} achieve constant
stretch, i.e., a makespan within a constant factor of the maximum minimum distance that a 
single robot would have to move between the input configurations, when connectivity is required, but for a different (more permissive) model of parallel moves than ours.
However their algorithm can only be applied to scaled configurations and thus does not work for general inputs. 

As stated before, Akitaya et al.~\cite{akitaya.demaine.korman.ea2022compacting-squares} claimed \APX-hardness for minimizing the makespan in the \sequentialcubes~model.
Our hardness of approximation results are the first for the $\parallelcubes$~model, and, to the best of our knowledge,  the first \logAPX-hardness for this class of motion planning reconfiguration problems. 

%% file: 02-preliminaries.tex
\section{Preliminaries}
\label{sec:preliminaries}

Our work considers the coordinated motion of unit cubes in the three-dimensional \newterm{integer grid}, following and extending the notation of~\cite{a.akitaya_et_al:LIPIcs.ESA.2025.28}.
For $v\in\mathbb{Z}^3$, let $x(v),y(v),z(v)\in \mathbb{Z}$ refer to its coordinates along each of the three axes.
The three axes point \emph{east}, \emph{north}, and \emph{up} in that order.
To measure distances in the integer grid, we use the $L_1$ and $L_\infty$ norms defined as
\[
    \norm{u}_1 = \abs{x(u)} + \abs{y(u)} + \abs{z(u)}\\
    \text{and}\\
    \norm{u}_\infty = \max\{\abs{x(u)}, \abs{y(u)}, \abs{z(u)}\}.
\]

A \newterm{cell} $c$ in our setting is a cube~$\mathcal{Q}_3$, anchored at its minimal coordinates, we write $c\in\mathbb{Z}^3$.
Each cell thus has eight vertices and six {faces}, each bounded by a four-cycle.
Two cells are then \newterm{face-}, \newterm{edge-}, or \newterm{vertex-adjacent} if they share a common face, edge, or vertex.

\subparagraph*{Configurations.}
A \newterm{configuration} $C$ of $n$ cube \newterm{modules} is defined as a set of $n$ cells that are occupied by one module each, and is valid exactly if the face-adjacency relation (the \newterm{dual graph of~$C$}) of occupied cells induces a connected graph.
For a given $C$, we denote by $B$ the minimal axis-aligned \newterm{bounding box} containing all modules.
The \newterm{extent} of a bounding box along the $x$, $y$, and $z$ axes defines its \newterm{width}, \newterm{depth} and \newterm{height} (recall that $z$ is \emph{up}).

\subparagraph*{Moves.}
Modules perform two types of \newterm{move}: a \newterm{slide} and a \newterm{convex transition}. 
A slide translates a module by one unit along the surface of two adjacent modules, whereas a convex transition moves it along the surface of a single module across two axes. 
Following the notation and conventions of~\cite{a.akitaya_et_al:LIPIcs.ESA.2025.28}, both moves take unit time and are executed at constant~speed.

\subparagraph*{Transformations.}
Parallel moves are performed in discrete \newterm{transformations} with unit-time duration.
These must preserve a connected \newterm{backbone} which requires moving modules to be \newterm{free}: moving modules do not provide connectivity.
A set of modules $M$ is free in a configuration $C$ if the removal of any subset of $M$ from $C$ produces a connected configuration.
An individual module $\module{m}\in C$ is then free if the set $\{\module{m}\}$ is.
In addition, transformations must not induce \newterm{collisions}, in which the volume of two modules overlap.
For instance, this forbids modules from trading places during a transformation, entering the same cell, or performing a convex transition through the same cell.
We refer to~\cite{a.akitaya_et_al:LIPIcs.ESA.2025.28} and~\cref{fig:collisions} for details.

\begin{figure}
    \includegraphics{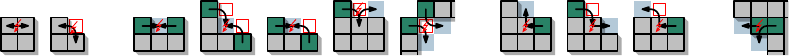}%
    \caption{Examples of collisions. Figure from~\cite{a.akitaya_et_al:LIPIcs.ESA.2025.28}, used with the authors' permission.}
    \label{fig:collisions}
\end{figure}

\subparagraph*{Problem statement.}
The \parallelcubes problem gives two configurations of~$n$ modules each and an integer $k\in\mathbb{N}^+$ and asks for a sequence of at most $k$ transformations that obtains one from the other.
Such a sequence is then called a \newterm{schedule} of~\newterm{makespan}~$k$.

\medskip%
This completes the description of our model; the following are auxiliary terms~and~notation.

\subparagraph*{Projection.}
We define the \newterm{projection} of a configuration $C\in\mathbb{Z}^3$ along an axis $\lambda\in\{x,y,z\}$ as the two-dimensional configuration $\lambda[C]=\{(\rho(\module{m}),\phi(\module{m}))\mid \module{m}\in C, \{\rho,\phi\} = \{x,y,z\}\setminus\{\lambda\}\}$.
Furthermore, we define the \newterm{area} of a two-dimensional projection as $\area(\lambda[C]) =\abs{\lambda[C]}$.

\subparagraph*{Neighborhoods.}
We use notions of \emph{neighborhoods} of cells and modules throughout our paper.
Let $c\in\mathbb{Z}^3$.
We define the open $L_1$ and $L_\infty$ neighborhoods of $c$ as
\[
    N(c)\coloneqq\{c'\in\mathbb{Z}^3\mid \norm{c-c'}_1=1\}\\
    \text{and}\\
    N[c]\coloneqq\{c'\in\mathbb{Z}^3\mid \norm{c-c'}_\infty=1\}.
\]
By an additional asterisk, e.g., $N^*[c]$, we indicate the neighborhoods which also include $c$.
With a slight abuse of notation, these definitions extend to sets of cells $C\subset\mathbb{Z}^3$, e.g.,
\[
    N^*(C)\coloneqq\bigcup_{c\in C}N^*(c)\\
    \text{and}\\
    N(C)\coloneqq N^*(C)\setminus C.
\]
The (closed) $r$-neighborhood is the $r$th iterate, e.g., ${N_r^*(c)=N(N_{r-1}^*(c))}$ with~${N_1^*=N^*}$.

%% file: 03-complexity.tex

\section{Computational complexity}
\label{sec:complexity}

In this section, we provide several complexity results for \parallelcubes.
We~first give a negative answer to a question posed in~\cite{a.akitaya_et_al:LIPIcs.ESA.2025.28} by showing \NP-hardness for makespan $1$ and symmetric difference $1$ for the three-dimensional setting.
This shows that \parallelcubes is not \FPT parameterized by the symmetric difference.
Subsequently, we establish \logAPX-hardness by reducing from \textsc{Set Cover}; this reduction inherently applies to the sequential setting.
We start with \NP-hardness for constant-size symmetric difference.

\begin{restatable}[$\star$]{theorem}{theoremNPHardness}
    \label{thm:unlabeled-sym-diff-hard}
    {\parallelcubes} is \NP-complete for makespan $1$ and symmetric difference size $1$.
\end{restatable}

\begin{proof}[Proof sketch]
    We reduce from \textsc{Planar Monotone 3Sat}, which asks whether a given Boolean formula is satisfiable~\cite{berg.khosravi2012optimal-binary}.
    Each clause consists of at most $3$ literals, all either positive or negative, and the clause-variable incidence graph must admit a plane drawing where variables are mapped to the $x$-axis, positive (resp., negative) clauses are mapped to the upper (resp., lower) half-plane, and edges do not cross the $x$-axis.
    For the construction, we place multiple instances of three gadget types based on a formula $\varphi$ as depicted in~\cref{fig:not-fpt-hardness-overview}.

    \begin{figure}[htb]
        \captionsetup[subfigure]{justification=centering}%
        \begin{subfigure}[b]{\columnwidth}%
            \centering%
            \includegraphics[width=0.75\columnwidth]{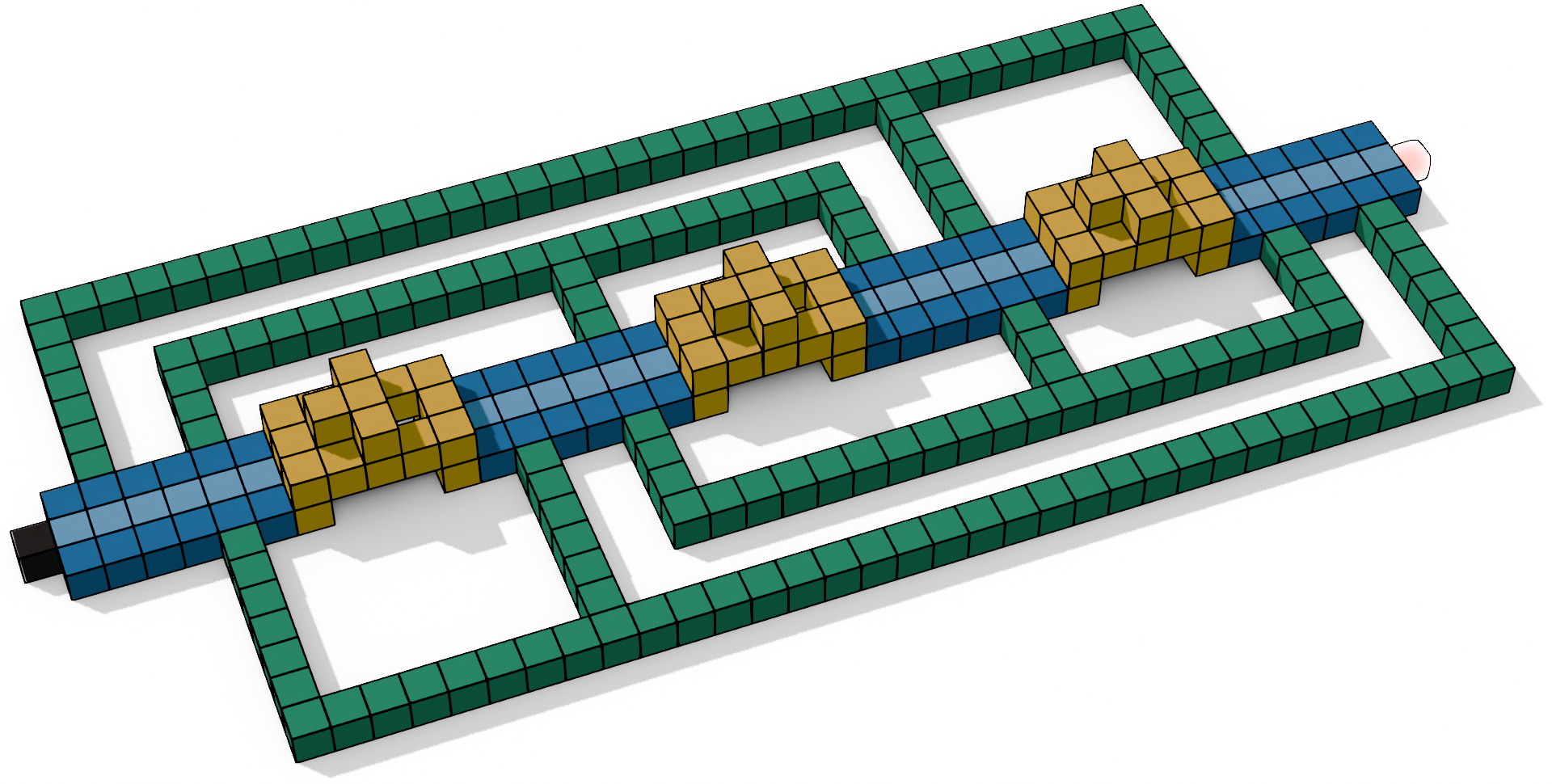}
            \subcaption{}
            \label{fig:not-fpt-hardness-overview}%
        \end{subfigure}%
        \par\bigskip%
        \begin{subfigure}[b]{\columnwidth/3}%
            \centering%
            \includegraphics[page=3]{figures/sat-reduction/variable-connector-diagram}\par%
            \includegraphics[page=2]{figures/sat-reduction/variable-connector-diagram}\par%
            \includegraphics[page=1]{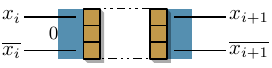}\par%
            \subcaption{}
            \label{fig:hardness-connector-blowup}%
        \end{subfigure}%
        \begin{subfigure}[b]{\columnwidth/3}%
            \centering%
            \includegraphics[width=\columnwidth]{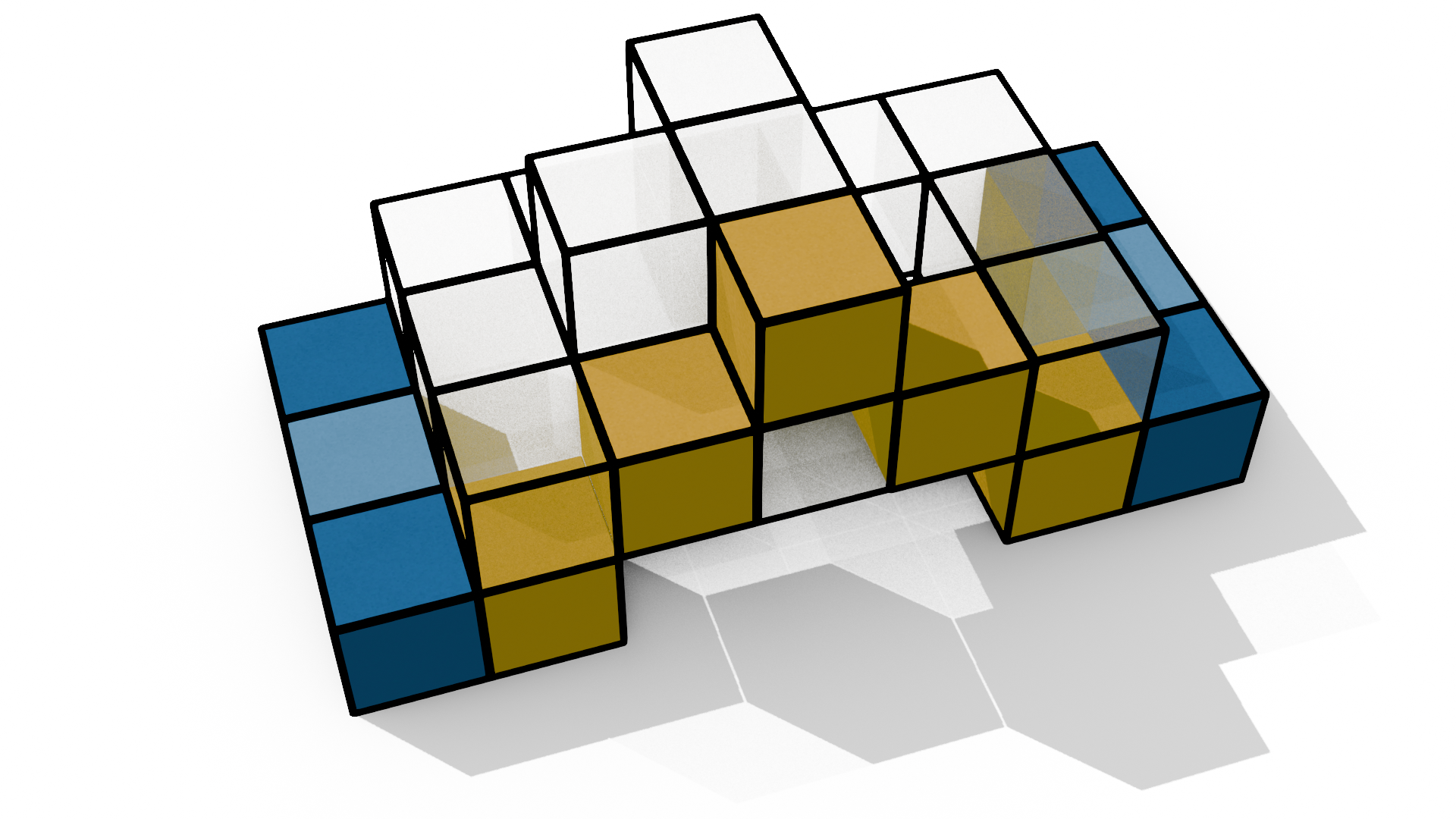}%
            \llap{\includegraphics[width=\columnwidth]{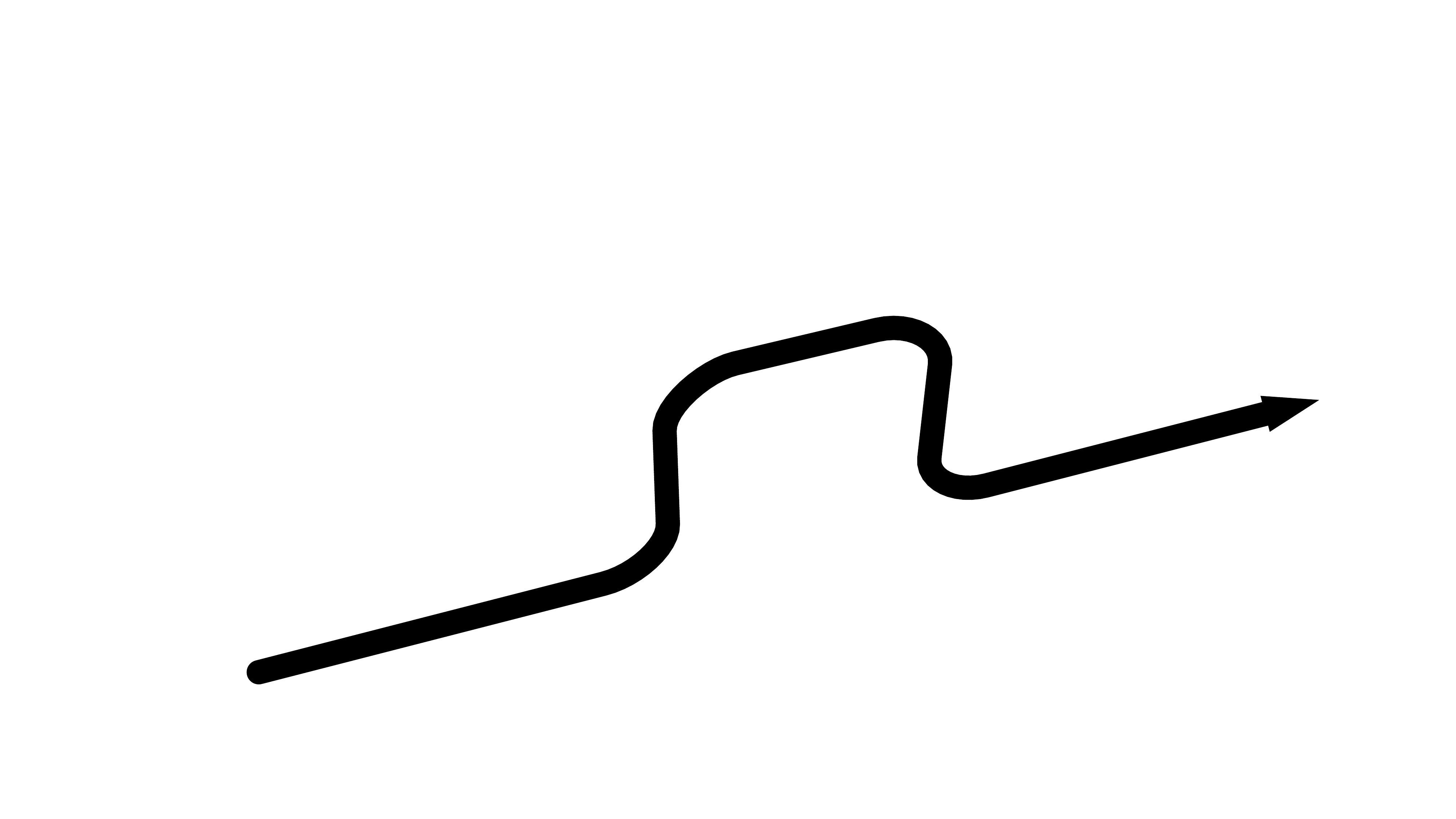}}%
            \subcaption{}
            \label{fig:not-fpt-hardness-connector-closeup-a}%
        \end{subfigure}%
        \begin{subfigure}[b]{\columnwidth/3}%
            \centering%
            \includegraphics[width=\columnwidth]{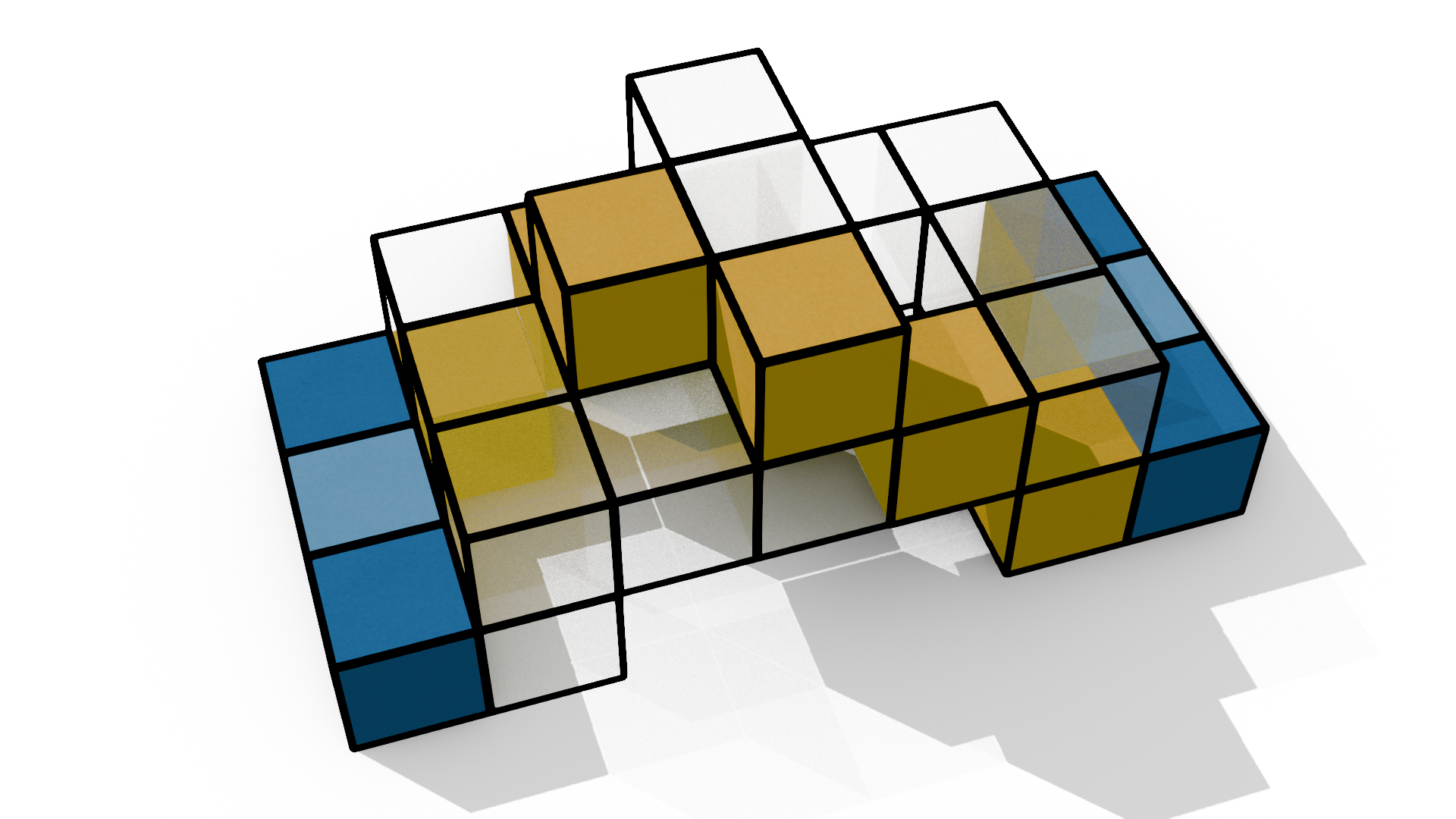}%
            \llap{\includegraphics[width=\columnwidth]{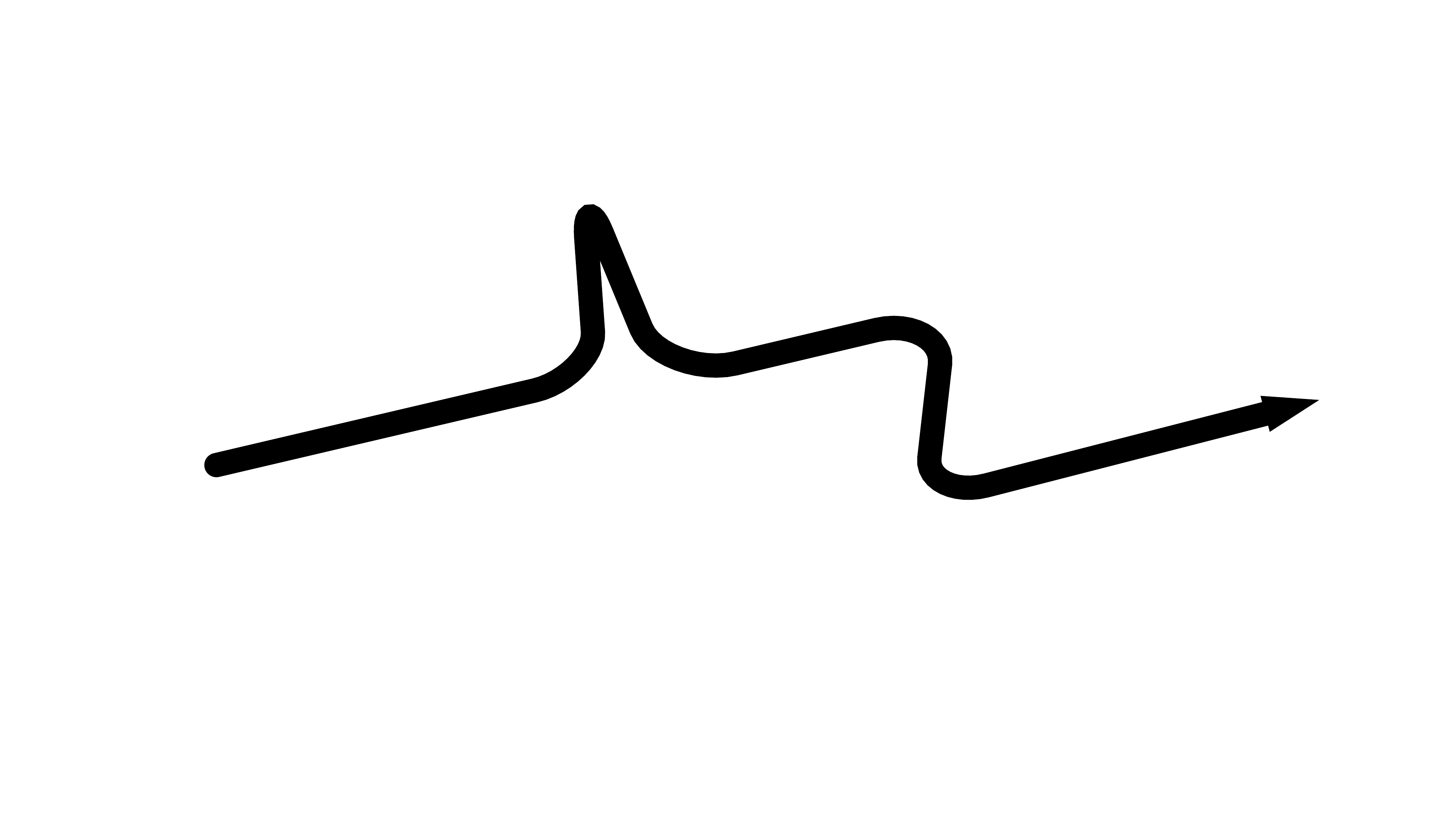}}%
            \subcaption{}
            \label{fig:not-fpt-hardness-connector-closeup-b}%
        \end{subfigure}%
        \caption{In (a) we show our construction, with colours indicating gadgets as outlined in the text, for $\varphi=(x_1\lor x_3\lor x_4)\land(x_1\lor x_2\lor x_3)\land(\overline{x_1}\lor\overline{x_2}\lor\overline{x_4})\land(\overline{x_2}\lor\overline{x_3}\lor\overline{x_4})$.
        Variable gadgets are placed in ascending order left to right.
        In (b)--(d) we show the connector gadget and possible paths through.}
        \label{fig:not-fpt-hardness}
    \end{figure}%

\pagebreak
    In particular, we place \newterm{variable gadgets} (blue) in the $xy$-plane along the $x$-axis, and differentiate between \newterm{positive} and \newterm{negative} modules in the variable gadgets based on their $y$-coordinate.
    Between any two successive variables, we place a \newterm{connector gadget} (yellow).
    A detailed breakdown can be seen in~\cref{fig:hardness-connector-blowup}.
    Also attached to each variable gadget is at least one \newterm{clause gadget}~(green), effectively a thin comb structure that attaches to the positive or negative side of each contained variable's gadget.
    Finally, a single \newterm{difference module} (black/transparent) is located at one end of the construction in the start configuration and at the other end in the target configuration.
    It comprises the entire symmetric difference.
    To empty the difference module's initial cell and fill its target cell in the same step, we must find a path from one position to the other that does not cut the configuration apart.
    In particular, this path passes through each variable gadget on either the positive or negative side, separating them from clauses and thus encoding an assignment of Boolean values.
\end{proof}

\vspace*{-0.5em}%
\begin{restatable}[$\star$]{corollary}{corollaryInapprox}
    Unless \P = \NP, \parallelcubes cannot be approximated within a factor better than $2$ in polynomial time, even for constant-size symmetric difference.
\end{restatable}

We emphasize that we reduce from \textsc{Planar Monotone 3Sat} for clarity only.
Nonetheless, the same construction remains valid when reducing from \textsc{Satisfiability}; when the clause-variable incidence graph is non-planar, the clauses can be embedded in the third dimension.
The following is an immediate consequence of~\cref{thm:unlabeled-sym-diff-hard}.
\begin{corollary}
    \label{cor:not-fpt-in-sym-diff}
    {\parallelcubes} is not \FPT\ parameterized by the size of the symmetric difference.
\end{corollary}

We proceed by showing that both the parallel and sequential variants of \slidingcubes are \logAPX-hard.
We provide a reduction from \textsc{Set Cover}.
Our entire construction fits inside a bounding cube of size $\mathcal{O}(n^4 k^7)$ where $n,k$ refer to the set cover instance having $n$~many elements, and $k$ many sets.
Its reliance on large three-dimensional structure makes it unsuitable for the two-dimensional version of the problem.

\begin{restatable}[$\star$]{theorem}{theoremLogAPX}
    Both \parallelcubes and \sequentialcubes are \logAPX-hard; thus, cannot be approximated with a factor better than $\Theta(\log(n))$, unless~$\P=\NP$.
\end{restatable}

\begin{proof}[Proof sketch.]
	We break the \logAPX-hardness proof of \parallelcubes into two parts.
	In the first part, we show \logAPX-hardness of a variant of the problem that is easier to analyze, which we call \textsc{Immobile}.
	In this variant, there exist two types of~modules: (1)~\newterm{immobile modules} that cannot perform any moves and, thus, will always be in the same position, and (2) \newterm{mobile modules} that actually can perform the introduced moves.
	Note that immobile modules may still induce collisions as well as provide a backbone.
	Moreover, only a single module moves at a time, as we are used to from the sequential setting.

	In the second part, we argue that the assumptions made in the immobile variant remain valid once we augment every immobile module with a \newterm{spike gadget}.
	A spike gadget is a long, width-1 chain of modules extending in a fixed direction, whose sole purpose is to block the motion of all modules that were immobile in the \textsc{Immobile} variant.
	Any attempt to slide these modules disconnects the spike gadget from the remainder of the construction, thereby violating the backbone constraint. 

	\medskip
	Our reduction is from the \logAPX-hard problem \textsc{Set Cover}~\cite{DinurS14}. 
	An instance $\Sigma_{n,k}$ consists of natural numbers $1,\ldots,n$ and a set $S=\{s_1,\ldots,s_k\}$ of subsets of $\{1,\ldots,n\}$ such that $\bigcup_{i=1}^k s_i = \{1,\ldots ,n\}$. 
	The task is then to find a set $S^\star$ of minimum cardinality such that $\bigcup_{s_i \in S^\star} s_i = \{1,\ldots, n\}$. 
	A top-down view of the construction’s blueprint is shown in~\cref{fig:sc-reduction-diagram}.
	Black lines describe thin paths of immobile modules, each of width and height $1$, that connect all gadgets to one another.
	\begin{figure}[htb]%
		\centering%
		\includegraphics{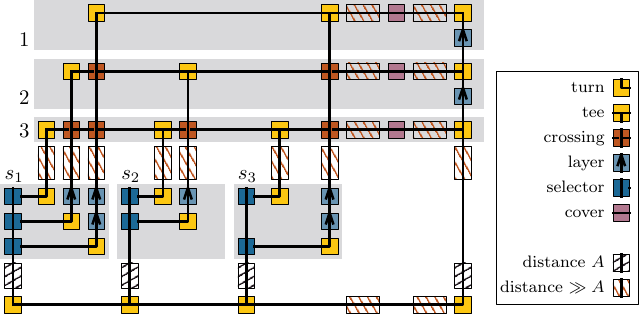}%
		\caption{A diagram showing our reduction from the \textsc{Set Cover} instance with $n=3$ over the sets $s_1=\{1,2,3\}$, $s_2=\{2,3\}$, and $s_3=\{1,3\}$.}
		\label{fig:sc-reduction-diagram}%
	\end{figure}%
	
	We introduce a couple of different gadgets.
	The gadgets that directly relate to the set cover instance are the \newterm{cover gadget} and the \newterm{selector gadget}.
	In particular, for every $j = 1,\ldots,n \in \mathbb{N}$, we introduce a cover gadget, depicted as a violet square in~\cref{fig:sc-reduction-diagram}, and in detail in~\cref{fig:cover-gadget-main}.
	In a cover gadget, the path of immobile modules is interrupted by a single mobile module $\module m_j$, highlighted in orange in~\cref{fig:cover-gadget-main}.
		\begin{figure}[p]%
		\captionsetup[subfigure]{justification=centering}%
		\begin{subfigure}[b]{0.5\columnwidth-0.5em}%
			\centering%
			\includegraphics[width=0.7\columnwidth]{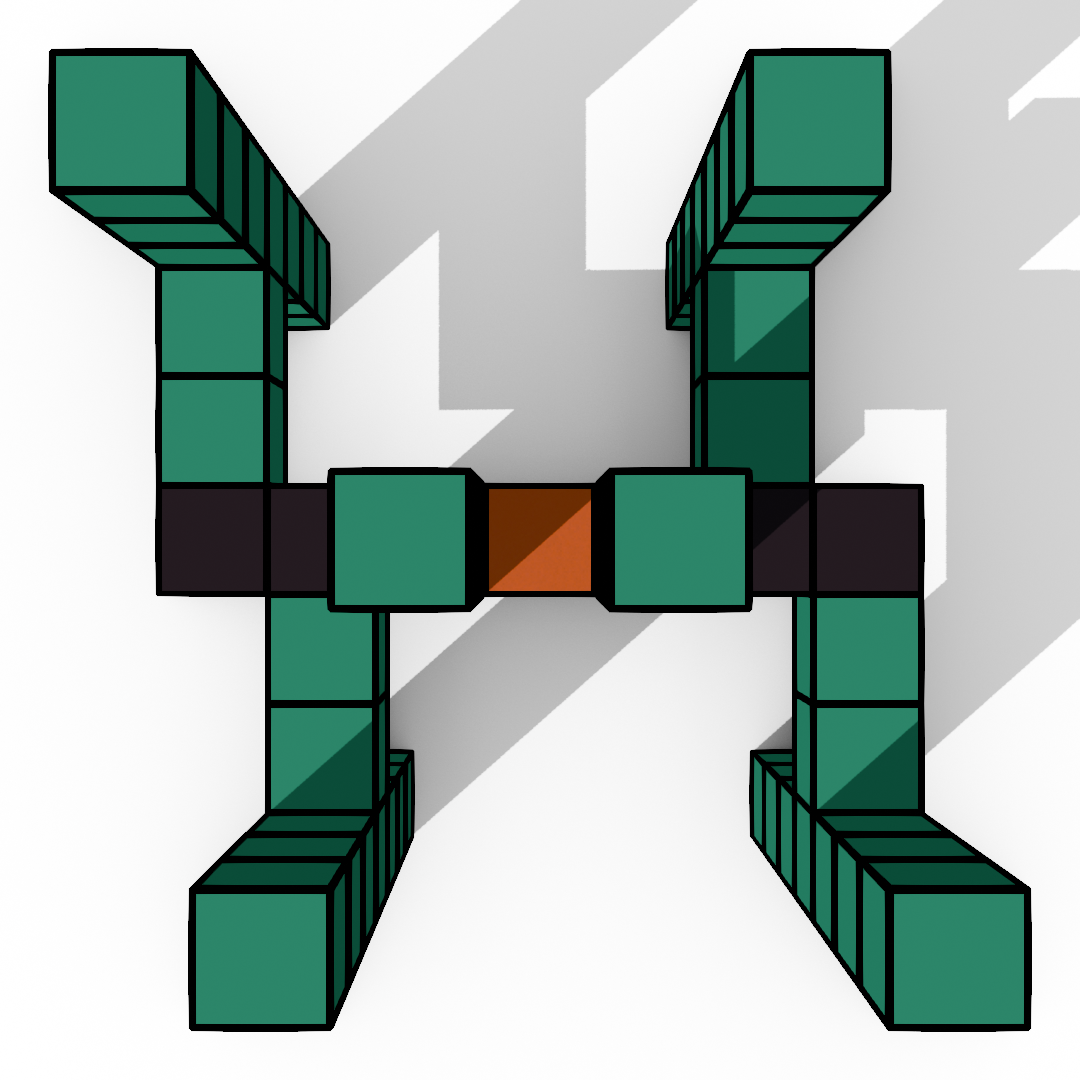}%
			\subcaption{}
			\label{fig:covera-main}
		\end{subfigure}
		\hfill%
		\begin{subfigure}[b]{0.5\columnwidth-0.5em}%
			\centering%
			\includegraphics[width=0.7\columnwidth]{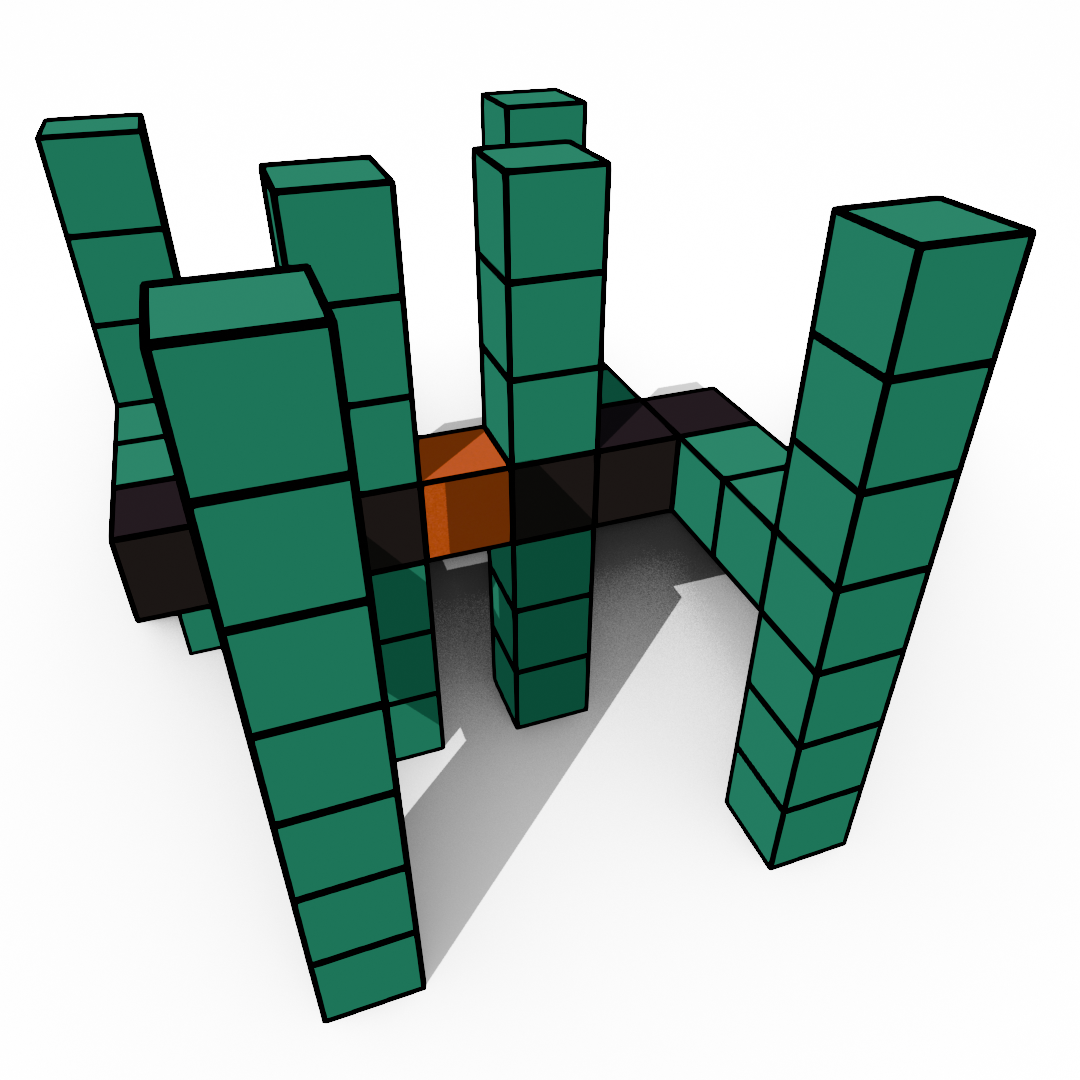}%
			\subcaption{}
			\label{fig:coverb-main}
		\end{subfigure}%
		\par%
		\vspace{0.75cm}%
		\begin{subfigure}[b]{\columnwidth}%
			\hfil%
			\includegraphics[page=1]{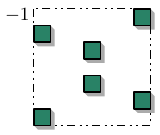}\hspace{2em}%
			\includegraphics[page=2]{figures/sc-reduction/cover-gadget-diagram}\hspace{2em}%
			\includegraphics[page=3]{figures/sc-reduction/cover-gadget-diagram}%
			\subcaption{}
		\end{subfigure}
		\caption{The cover gadget. The orange cube describe the module in the symmetric difference, with starting position at layer $0$, and target position in layer $+1$. The dark cubes describe the immobile modules. Green cubes correspond to spikes.}
		\label{fig:cover-gadget-main}
	\end{figure}	
	Cover gadgets contain the modules of the symmetric difference between the initial and the final configuration: more precisely, in $C_I$, the orange module is at level $0$, whereas in $C_F$ it has been raised to level~$+1$, directly above its original position.
	To enable motion of the orange modules inside the cover gadgets, we must ensure that the backbone remains connected.
	This is realized by the selector gadgets, depicted as a dark blue square in~\cref{fig:sc-reduction-diagram}, and in detail in~\cref{fig:selector-gadget-main}.
	\begin{figure}[p]%
		\captionsetup[subfigure]{justification=centering}%
		\begin{subfigure}[b]{0.5\columnwidth-0.5em}%
			\centering%
			\includegraphics[width=0.7\columnwidth]{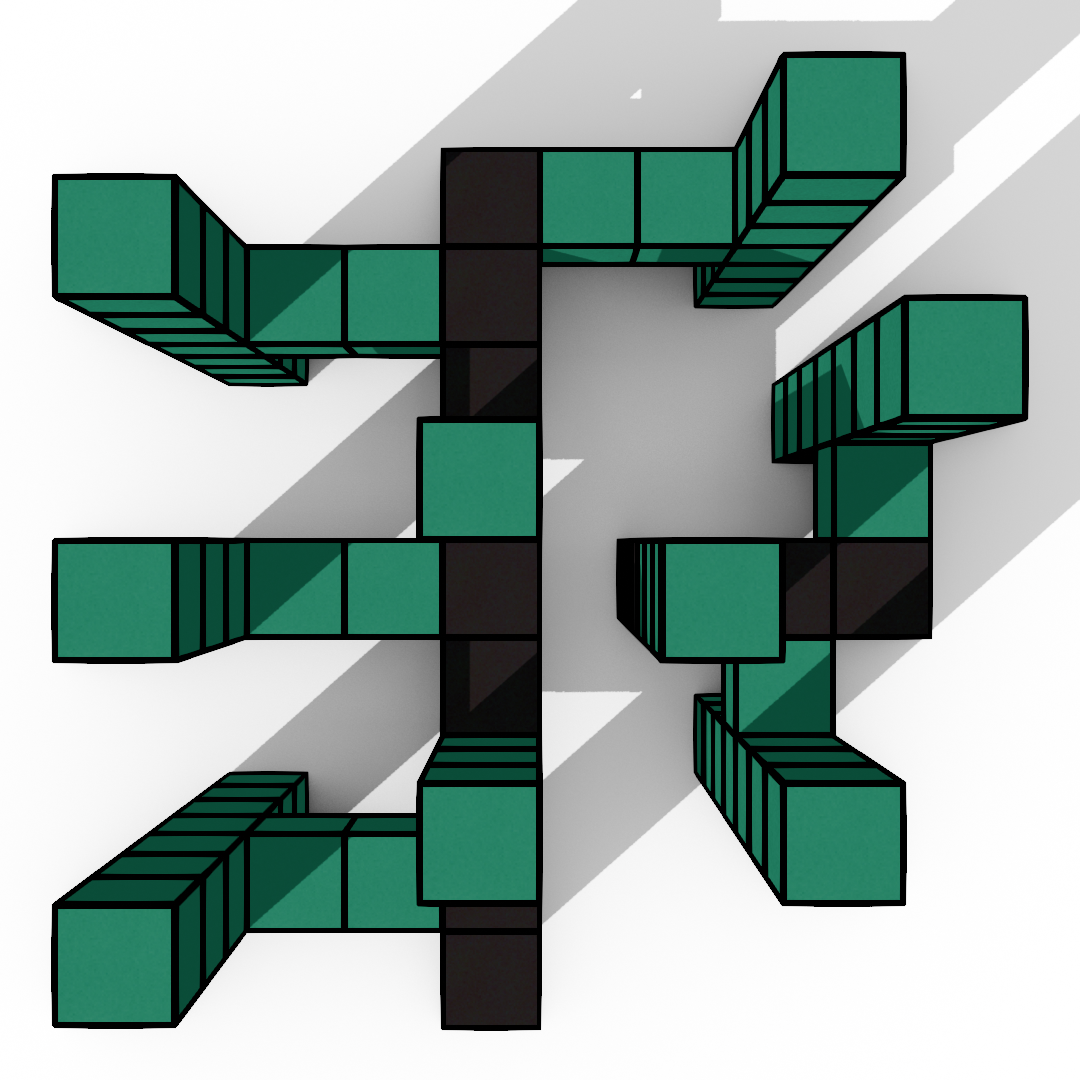}%
			\subcaption{}
			\label{fig:selectora-main}
		\end{subfigure}
		\hfill%
		\begin{subfigure}[b]{0.5\columnwidth-0.5em}%
			\centering%
			\includegraphics[width=0.7\columnwidth]{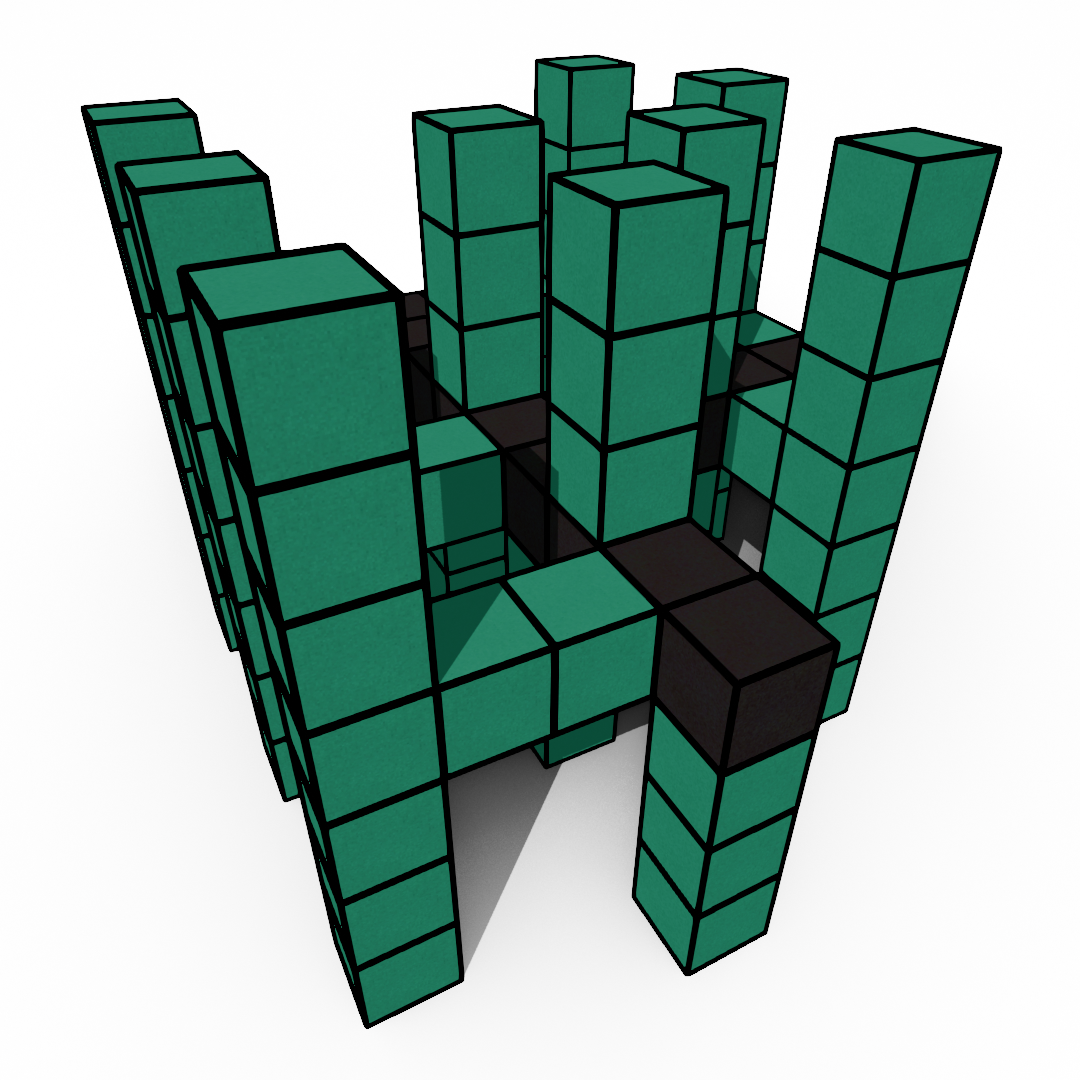}%
			\subcaption{}
			\label{fig:selectorb-main}
		\end{subfigure}%
		\par%
		\vspace{0.75cm}%
		\begin{subfigure}[b]{\columnwidth}%
			\hfil%
			\includegraphics[page=1]{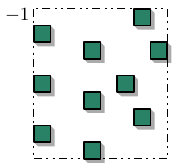}\hspace{2em}%
			\includegraphics[page=2]{figures/sc-reduction/selector-gadget-diagram}\hspace{2em}%
			\includegraphics[page=3]{figures/sc-reduction/selector-gadget-diagram}%
			\subcaption{}
		\end{subfigure}
			\caption{Selector gadgets enable modules within cover gadgets to perform their moves.}
			\label{fig:selector-gadget-main}
		\end{figure}
	The gadget is used to represent the containment of a number $j$ in a set $s_i$.
	Within a selector gadget, one path of immobile modules is oriented orthogonally to a second path, and each path connects to the cover gadget associated with the number $j$, but from different directions.
	Moving a mobile module to connect the two paths establishes the backbone constraint that in turn enables the orange module of the cover gadget to move.
	The remaining gadgets handle connectivity, crossings, and the lifting of paths needed to realize some of those crossings.
	
	The reduction then works as follows.
	Because the dual graph of $C_I$ forms a tree, no orange module in the symmetric difference is able to move without violating the backbone constraint.
	We begin with a single mobile module $\module m$ in $C_I$, placed at the bottom-left corner, and move it toward a selector $s_i$ to fill the gap inside it. 
	By doing so, we close a cycle in the dual graph, and therefore enable those orange modules inside of the cover gadgets that belong to the set $s_i$ to move to their respective final positions.
	Since the distance between the bottom row (which forms a comb with the selector gadgets, forcing us to route through~it) and any selector is large, we aim to traverse these distances as rarely as possible.
	
	What remains to show is the equivalence between an optimal solution to the set cover instance and a schedule of minimum makespan.
	This follows from appropriate adjustments to the distances between gadgets, with a thorough analysis of all admissible movements.
	\end{proof}

%% file: 04-algorithm.tex

\pagebreak
\section{Teleport}
\label{sec:teleport}

Given a subset $Q$ of the configuration $C$ we say that a schedule \emph{teleports} $Q$ if it transforms~$C$ into a new configuration $C'$ with $(C\setminus Q)\subset C'$.
We show that given two configurations $C$ and $C'$ so $C' = (C\setminus S)\cup S'$ where $C\setminus S$, $S$ and $S'$ are all connected sub-configurations, there is a schedule that transforms $C$ into $C'$, i.e., $S$ teleports into $S'$.
This result is of independent interest, but in this paper we do not try to minimize the makespans of such transformations.
This was done for simplicity as our main result only uses these schedules for constant sized sub-configurations, so the makespan is always constant.

The main technical tool of this section is a generalization of \textsc{LocateAndFree} from~\cite{abel.akitaya.kominers.ea2024universal-in-place}.
\textsc{LocateAndFree} locates a module in the outer face of the configuration and frees this module without changing the outer face.
We are more strict in the sense that we wish to teleport while maintaining the rest of the configuration (internal and external).

A \emph{face} of the configuration $C$ is the set of module faces in the boundary of a component of the complement $\overline{C}$ (the set of cells in $\mathbb{Z}^3$ not in $C$).
The \emph{outer face} of $C$ is the face in the boundary of the unbounded component of $\overline{C}$.
Two adjacent module faces $f_1$ and $f_2$ in a face of $C$ are \emph{slide-adjacent} if either they bound the same same cell $c\in\overline{C}$, or if a module in $c_1\in\overline{C}$ containing $f_1$ can reach a cell $c_2\in\overline{C}$ containing $f_2$ after a single move in the configuration $C\cup\{c_1\}$.
Each face of $C$ can be seen as a 2-manifold obtained by joining slide-adjacent module faces at their common edge. (See \cite{abel.akitaya.kominers.ea2024universal-in-place} for more details. Each face of $C$ corresponds to a component of the ``slide-adjacency graph'' in \cite{abel.akitaya.kominers.ea2024universal-in-place}.)
An edge $e$ of a configuration face is \emph{pinched} if it is contained in more than two module faces in a face of $C$. In particular, in three dimensions a pinched edge is contained in exactly four module faces and up to two configuration faces. This happens when two modules that are not face-adjacent share an edge.
We define an \emph{extended face} of $C$, the manifold obtained by performing a topological surgery at each pinched face joining pairs of module faces $f_1$ and $f_2$ at the common pinched edge so that $f_1$ and $f_2$ are contained in the boundary of the same module in $C$.
Note that an extended face of $C$ might contain many faces of $C$ since the topological surgery might merge different faces of $C$.
The \emph{outer extended face} of $C$ is the extended face that contains the outer face of $C$.
We say that a module is \emph{in a face} $F$ of $C$ if one of its faces is contained in $F$.
Intuitively, $C$ has a single extended face if and only if its interior is simply connected: the interior of the polycube equivalent to $C$ is homeomorphic to the open ball.

\begin{restatable}[$\star$]{lemma}{lemmafreeexists}
\label{lem:free-existance}
There is a free module in the outer extended face of any connected configuration.
\end{restatable}

\begin{restatable}[$\star$]{lemma}{lemmaTeleport}
    \label{lem:teleport}
    Given a configuration $C$, a connected subset $S\subset C$ with $|S|\ge2$ such that $C \setminus S$ is connected, an empty cell $e$ adjacent to $S$, a free module $\module{m}\in S$ so that $e$ and $\module m$ are in the same extended face of $S\setminus\{\module m\}$, and so that both  $(S\setminus\{\module m\})\cup\{e\}$ and $(C\setminus\{\module m\})\cup\{e\}$ are connected. There exists a schedule with makespan $2^{\mathcal{O}(|S|)}$ teleporting $\module m$ to $e$.
\end{restatable}

We require that $|S|\ge 2$ and $(S\setminus\{\module m\})\cup\{e\}$ is connected so we can bound the makespan as a function of $|S|$.
If either of these conditions fail there are examples that require $\module m$ to walk on modules not in $S$, giving a runtime dependency on $|C|$.
Those conditions can be dropped if we are only concerned with the existence of such schedules.
Given two configurations $C\cup S$ and $C\cup S'$ we can apply successively \Cref{lem:teleport} to grow a component of  $S\cap S'$ by filling an empty cell of $S'\setminus S$ that is adjacent to this component. We can then obtain the following:

\begin{corollary}
    \label{cor:teleport}
    Given two configurations $C\cup S$ and $C\cup S'$ with connected $C$, $S$ and $S'$, there exist a schedule that teleports $S$ into $S'$.
\end{corollary}

\section{A worst-case optimal algorithm}
\label{sec:alg}

We leverage the results from~\cref{sec:teleport} to efficiently reconfigure in the \parallelcubes model.
On a high level, our methods form a generalization of~\cite{a.akitaya_et_al:LIPIcs.ESA.2025.28} to three dimensions.
The overall strategy can be divided into three phases: \textsc{Gather}, \textsc{Compress}, and \textsc{Compact Reconfiguration}.
The goal of \textsc{Gather} is to obtain a structure called~a~\newterm{snake} which is informally a set of connected $5\times 5\times 5$ metamodules.
Intuitively, the snake can move freely and quickly through the configuration without affecting the connectivity of the remaining configuration.
In the \textsc{Compress} phase, we use the snake to create a coarse projection of the configuration onto one of the faces of the configuration's bounding box.
Without loss of generality, assume this is the top face.
We then use this structure to sweep through the bounding box compressing the rest of the configuration into a \emph{compact} configuration of $5\times 5\times 5$ metamodules.
We can apply the previous phases to both $C_1$ and $C_2$ and the remaining problem reduces to reconfiguring between two compact configurations, which is handled by our \textsc{Compact Reconfiguration} phase.
A configuration $C$ is compact as defined in the literature if for every module $\module{u}\in C$, the cuboid defined by the position of $u$ and the origin is fully contained in the configuration:
\[
    \forall \module{u}\in C:\quad \forall a\in [0,x(u)]:\forall  b\in [0,y(u)]:\forall  c\in [0,z(u)]:\quad(a,b,c)\in C.
\]

\begin{figure}[htb]
    \captionsetup[subfigure]{justification=centering}%
    \begin{subfigure}[t]{0.25\columnwidth -0.75em}%
        \includegraphics[page=1,width=\columnwidth]{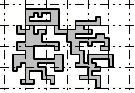}%
        \subcaption{}
    \end{subfigure}%
    \hfill%
    \begin{subfigure}[t]{0.25\columnwidth -0.75em}%
        \includegraphics[page=2,width=\columnwidth]{overview/overview}%
        \subcaption{}
    \end{subfigure}%
    \hfill%
    \begin{subfigure}[t]{0.25\columnwidth -0.75em}%
        \includegraphics[page=3,width=\columnwidth]{overview/overview}%
        \subcaption{}
    \end{subfigure}%
    \hfill%
    \begin{subfigure}[t]{0.25\columnwidth -0.75em}%
        \includegraphics[page=4,width=\columnwidth]{overview/overview}%
        \subcaption{}
    \end{subfigure}%
    \par\hfil%
    \begin{subfigure}[t]{0.25\columnwidth -0.75em}%
        \includegraphics[page=5,width=\columnwidth]{overview/overview}%
        \subcaption{}
    \end{subfigure}%
    \hfil%
    \begin{subfigure}[t]{0.25\columnwidth -0.75em}%
        \includegraphics[page=6,width=\columnwidth]{overview/overview}%
        \subcaption{}
    \end{subfigure}%
    \hfil%
    \begin{subfigure}[t]{0.25\columnwidth -0.75em}
        \includegraphics[page=7,width=\columnwidth]{overview/overview}%
        \subcaption{}
    \end{subfigure}%
    \caption{
        Overview of our algorithm (not to scale).
        In (a)--(d), we construct a snake, which we grow to $\Theta(A+h)$ modules in (e) and use to compact the remaining configuration in (f)--(g).
    }
    \label{fig:overview}
\end{figure}

In order to guide the creation and manipulation of the $5\times 5\times 5$ metamodules we subdivide the space into $5\times 5\times 5$ \emph{meta-cells} made of $5^3 = 125$ unit grid cells.
Without loss of generality, we can assume that the initial bounding box $B_i$ of the input configuration~$C_i$, $i\in\{1,2\}$, has dimensions that are multiple of 5, else we assume $B_i$ to be a smallest such a box containing~$C_i$ (adding at most $4$ units to each dimension).

Since the first two phases deal with a single input configuration, we now assume we are given a configuration $C$ with bounding box $B$.
We use the meta-cell decomposition to define a spanning tree $T$ of the configuration that will guide the \textsc{Gather} phase.
Let $\mathcal{C}$ be the set of components of $C$ induced by each meta-cell: $C$ restricted to a single meta-cell may be disconnected and each of these components is an element of $\mathcal{C}$.
We build ${T}$ as a spanning tree of $C$ in the following way.
Compute a spanning tree $\mathcal{T}$ of the adjacency graph of $\mathcal{C}$ and root it arbitrarily.
Using $\mathcal{T}$ as a guide, we build ${T}$ as follows.
The root of a component $S\in\mathcal{C}$ is either an arbitrary module in $S$ if $S$ is the root of $\mathcal{T}$, or a module in $S$ that is adjacent to its parent component in $\mathcal{T}$.
Compute a DFS spanning tree on each component in $\mathcal{C}$ from its root, making ${T}$ a spanning forest.
Now, connect adjacent trees until we have a single spanning tree of $C$, connecting the root of a component to its parent component in $\mathcal{T}$.
The root of $T$ is the root of its root component in $\mathcal{C}$.

We need to gather at least $(A+h)125$ modules in the snake for the \textsc{Compress} phase.
To achieve the makespan of $\mathcal{O}(A+h)$ we can only afford to gather at most a constant times~$(A+h)$ modules in the \textsc{Gather} phase.
For this reason, we find an appropriate subtree of ${T}_\module{s}$ of size $\Theta(A+h)$ to convert into the snake. The details are given in \cref{app:sec:alg:Ts}.

\subsection{Gathering phase}

\subsubsection{Initialize}

The first step to creating a snake is to obtain a meta-cell completely full of modules.
This section is dedicated to that task.
We first use \Cref{lem:subtree} to obtain a subtree ${T}_\module{m}$ with size $|{T}_\module{m}|\ge 2625$ and $|{T}_\module{m}|\in\mathcal{O}(1)$.
The number $2625$ is due to the fact that that is the number of cells contained in a $L_1$ ball of radius $12$.
Informally, the strategy is to use the modules of ${T}_\module{m}$ one by one to grow a ball centered at a module $\module{m}$ (using \Cref{lem:teleport}).
If such a ball has radius~12, it is guaranteed to contain a snake subconfiguration.
Although the 2D version of such a result would be much easier to achieve (a similar result is obtained in \cite{a.akitaya_et_al:LIPIcs.ESA.2025.28}), especially with our new technical tool of teleporting modules, the three dimensional version requires quite intricate case analysis, even using \Cref{lem:teleport}'s full power. 
This is because in three dimensions there are more ways that the rest of the configuration ($C\setminus {T}_\module{m}$) can block the advance of modules that are used to fill the desired ball.

\begin{restatable}[$\star$]{lemma}{lemmaGrowball}
	Given a subtree $T_\module{m}$ of $T$ with size $|{T}_\module{m}|\ge 2625$ and $|T_\module{m}|\in\mathcal{O}(1)$, there is a $\mathcal{O}(1)$ schedule that produces a configuration where the meta-cell containing $\module{m}$ is full, and every cell in $T\setminus T_\module{m}$ remains full.
	\label{lem:grow-ball}
\end{restatable}

\subsubsection{Finding a snake}
We now give a precise outline of the snake and its capabilities, and show it can be constructed.

\subparagraph*{Spine path.}
A \newterm{spine path} $P\in\mathbb{Z}^{3\times\ell}$ of length $\ell$ is a simple induced path in the face-adjacency relation of cells, consisting of three types of vertices: endpoints (\newterm{head} and \newterm{tail}), \newterm{major vertices}, and \newterm{minor vertices}.
The head vertex is connected by an induced path of up to~$9$ minor vertices to the first major vertex.
Two successive major vertices are connected by a straight, axis-parallel path of $4$ minor vertices.
The final major vertex is then connected to the tail vertex by another induced path of at most $9$ minor vertices.
To prevent self-intersections, vertices with distance $\geq 5$ along the path must also have $L_1$-distance $\geq 5$.

\begin{figure}[htb]
    \captionsetup[subfigure]{justification=centering}%
    \begin{subfigure}{\columnwidth/3 -0.3em}%
    {\includegraphics[width=\columnwidth]{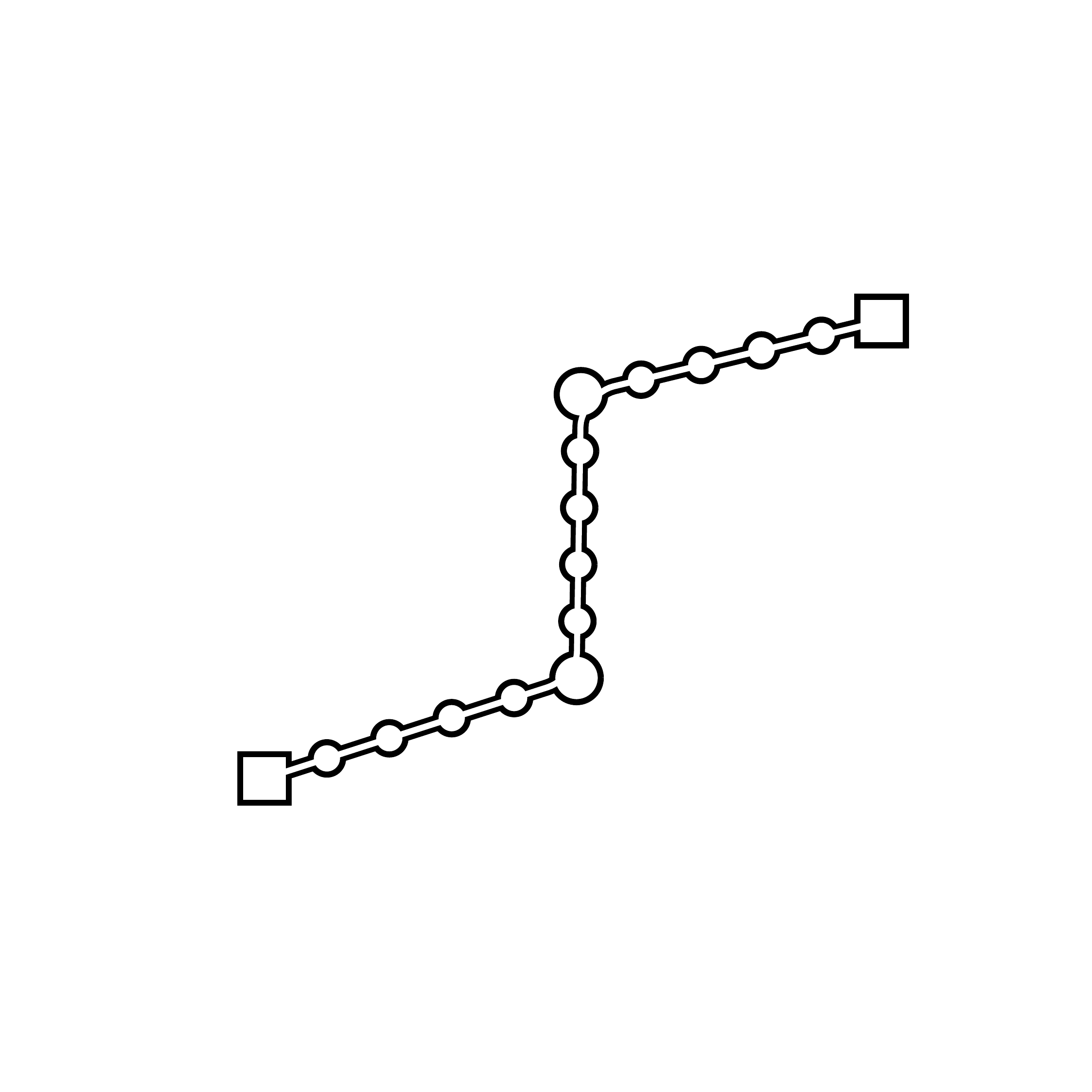}}%
        \subcaption{}
        \label{fig:snake-spine-path}
    \end{subfigure}%
    \hfill%
    \begin{subfigure}{\columnwidth/3 -0.3em}%
        \includegraphics[width=\columnwidth]{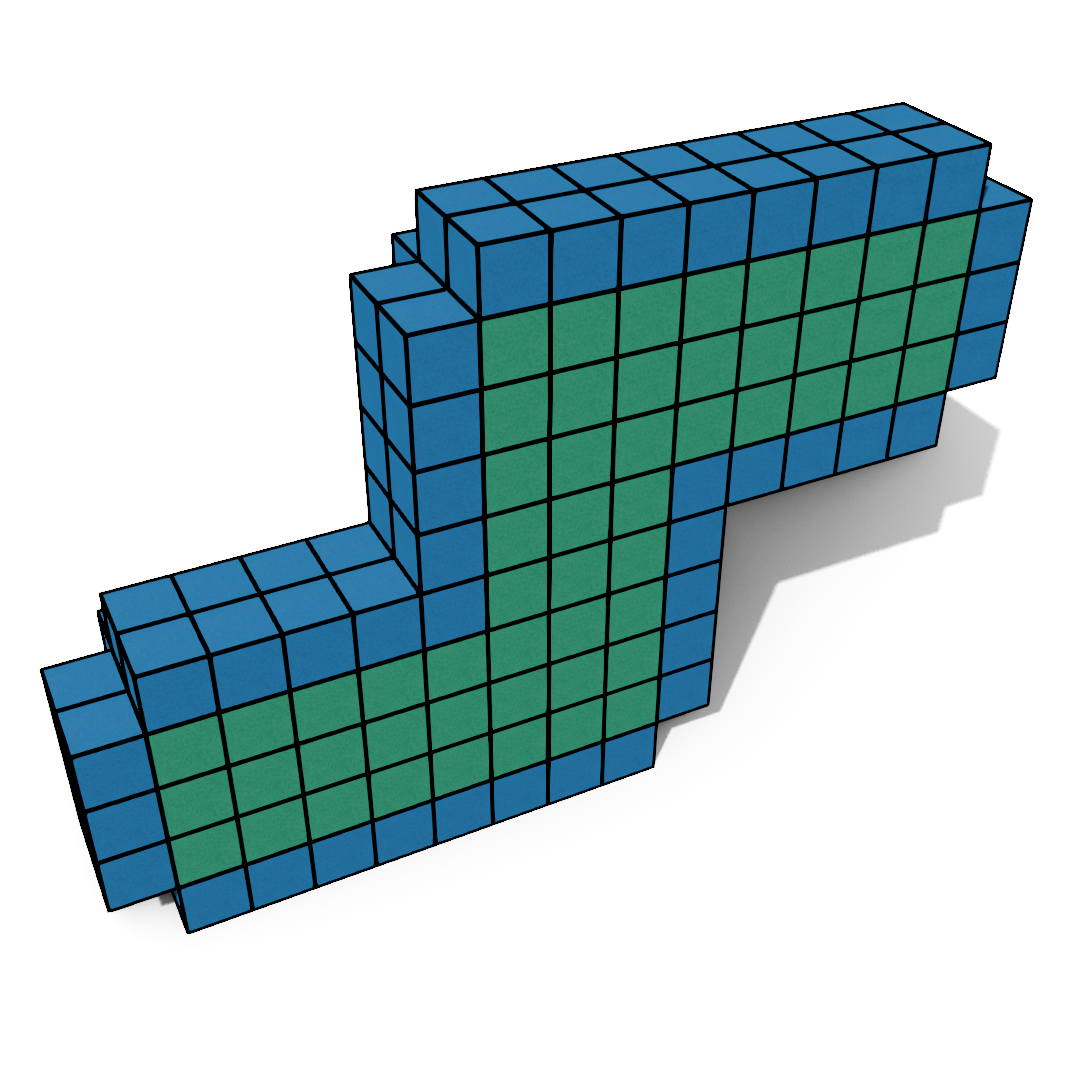}%
        \llap{\includegraphics[width=\columnwidth]{figures/snake/intro/example-snake-spine-path}}%
        \subcaption{}
        \label{fig:snake-skin-interior}
    \end{subfigure}%
    \hfill%
    \begin{subfigure}{\columnwidth/3 -0.3em}%
        \includegraphics[width=\columnwidth]{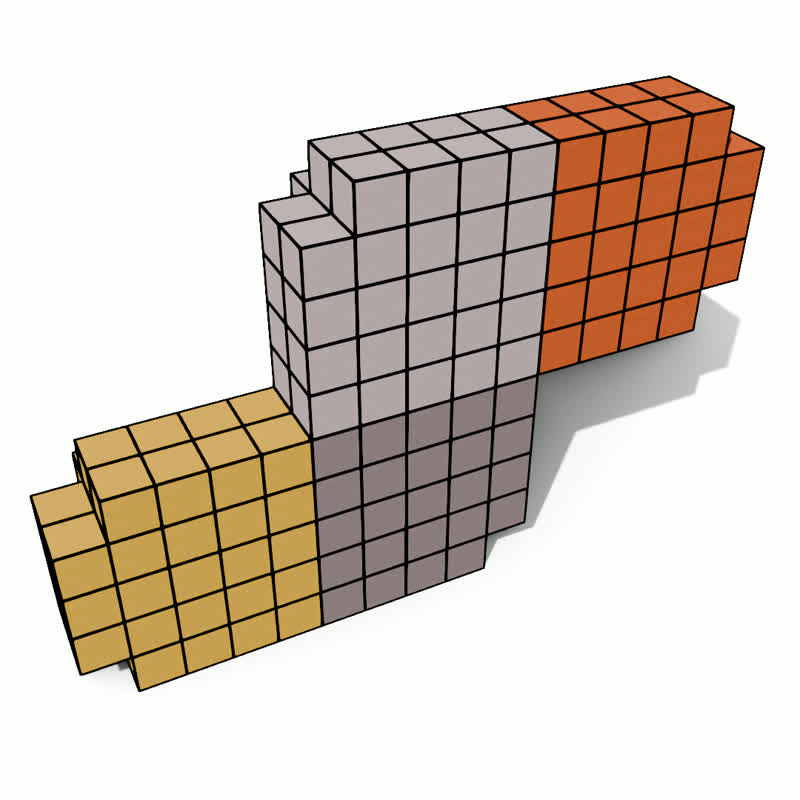}%
        \llap{\includegraphics[width=\columnwidth]{figures/snake/intro/example-snake-spine-path}}%
        \subcaption{}%
        \label{fig:snake-sections}%
    \end{subfigure}%
    \caption{The cell space defined by the spine of length $16$ shown in (a). For clarity, only one half is shown. In (b), skin and interior are colored differently, and (c) illustrates sections.}
    \label{fig:snake-space}
\end{figure}

\subparagraph*{Cell space.}
Let $P$ be a spine path.
The \newterm{interior cells} of the snake with spine~$P$ are then defined exactly as the closed $L_\infty$-neighborhood $N^*[P]$ of $P$, with the \newterm{skin cells} then being the open $L_1$-neighborhood of the interior, $N(N^*[P])$.
Their union forms the snake's \newterm{cell space} $S(P)\subset N^*_4(P)\subset \mathbb{Z}^3$.
We divide the cell space based on major vertices:
The \newterm{section} of a major vertex $v\in P$ that is not an endpoint is then exactly $N^*_2[v]\cap S(P)$, all remaining cells are then assigned to either the head or tail section.
This is illustrated in~\cref{fig:snake-space}.

\subparagraph*{Snake configuration.}
A \newterm{snake (sub)configuration} is an induced (sub)configuration by a spine path's cell space $S(P)\cap C$ in which \emph{all except for} the following cells are occupied:
\begin{enumerate}[\hspace{2em}]
    \item[\textsf{(S.1)}] Major vertices at which the spine bends.
    \item[\textsf{(S.2)}] Any subset of the tail section's interior, except for the interior cells that are adjacent to three skin cells, which must be occupied by \newterm{support modules}.
\end{enumerate}

As the majority of its cell space is occupied, a snake by this definition induces a connected subconfiguration.
In addition, keeping major vertices at bends in the spine unoccupied allows efficient teleportation of modules between the head and tail sections of any given snake.

\begin{figure}[htb]
    \captionsetup[subfigure]{justification=centering}%
    \begin{subfigure}{\columnwidth/3}%
        \includegraphics[width=\columnwidth]{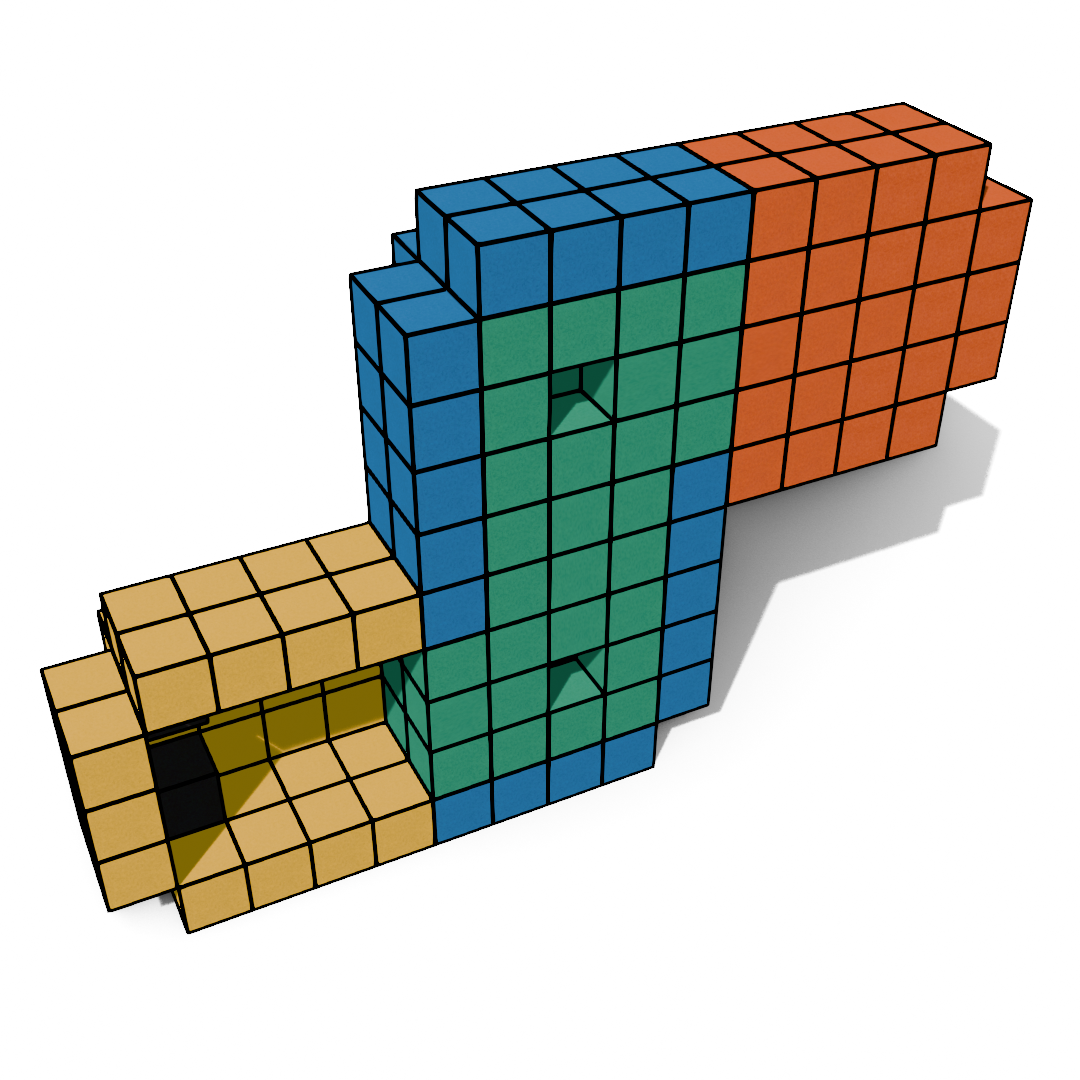}%
        \subcaption{}
        \label{fig:snake-cross-section}
    \end{subfigure}%
    \hfill%
    \begin{subfigure}{\columnwidth*2/3}%
        \centering%
        \includegraphics[page=1]{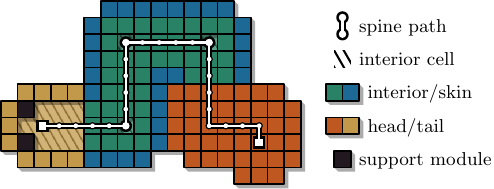}%
        \subcaption{}%
        \label{fig:snake-colors}%
    \end{subfigure}%
    \caption{Structure of a snake configuration. (a) Cross-section of a valid snake configuration. (b) Two-dimensional illustration using color notation for clarity.}
    \label{fig:snake}
\end{figure}

The remainder of this section assumes that we already have a snake configuration.
Once~$C$ contains a large enough $L_1$ ball due to~\cref{lem:grow-ball}, a constant-size snake can always be~found.

\begin{restatable}{lemma}{lemmaSnakeInBall}
    If a configuration contains a fully occupied $L_1$ ball with radius $12$, it contains a free snake configuration with a spine path of length $11$ and $198$ modules.
    \label{lem:snake-in-ball}
\end{restatable}

\medskip
As a crucial first step, we define the \emph{push} operation and its inverse, \emph{pull}.
The push operation can be used to modify an existing snake configuration, moving the head vertex and extending the spine path by one unit.
Similarly, the pull operation shortens the spine path by one unit by moving the tail vertex towards the head along the spine path.

\subparagraph*{Push.}
The \newterm{push operation} can be performed on snake configurations with at least $16$~modules in the tail section's interior.
It moves the head in a given direction and extends the spine path by one unit, forming a valid snake configuration in the resulting cell space.
To fill the new skin cells, it teleports up to~$12$ interior tail modules to the head.

\begin{restatable}[$\star$]{lemma}{lemmaPush}
    The push operation can be performed by a schedule of makespan $\mathcal{O}(1)$.
    \label{lem:push}
\end{restatable}

Having defined pushing as a forward-oriented operation that expands the snake, we note that the role of head and tail in a snake are actually interchangeable by teleportation.
\begin{restatable}[$\star$]{lemma}{lemmaSnakeIsReversible}
    Given a snake configuration $S(P)$, the roles of head and tail vertices can be swapped in $\mathcal{O}(1)$ transformations.
    \label{lem:snake-is-reversible}
\end{restatable}

Each push operation expands the snake's cell space, potentially including additional modules.
Some of these modules may not originally be free, meaning that we need to account for possible connectivity constraints when moving the snake.
We differentiate between a snake's \newterm{owned} and \newterm{held} cells and modules, which are dynamically tracked and updated.
In particular, we say that a snake configuration \emph{holds} a cell as result of a push operation if that cell contained a \emph{non-free module}, which is then also held.
A push operation enters $21$ additional cells into the cell space, some subset of which must then be held.
On the other hand, if such a cell contains a free module before a push operation, the snake may \emph{take ownership} of that module.
Crucially, the initial snake according to~\cref{lem:snake-in-ball} owns all $198$ contained modules.
A snake will hold all cells it moves into unless otherwise stated, and will leave a module when moving away due to the following operation.

\subparagraph*{Pull.}
The \newterm{pull operation} moves the tail by one unit in a given direction, shortening the spine path.
If this results in a held cell being removed from the cell space, the snake leaves a module in this cell; the number of held modules and held cells is therefore always equal.
If the snake leaves a cell that is not held, it does not leave a module behind, so the number of owned modules never decreases as a result of pulling.
This is strictly the reverse of~\cref{lem:push}:

\begin{corollary}
    Given a snake subconfiguration with a tail that has a sufficient number of empty interior cells, the pull operation can be performed by a schedule of makespan $\mathcal{O}(1)$.
    \label{cor:pull}
\end{corollary}

We define two additional high-level operations, forking and joining.
Both rely strictly on pushing, pulling, and the teleportation of modules within the spine of a snake.

\begin{figure}[htb]
    \captionsetup[subfigure]{justification=centering}%
    \begin{subfigure}[t]{\columnwidth/3 -0.5em}%
        \centering%
        \includegraphics[width=\columnwidth,page=1]{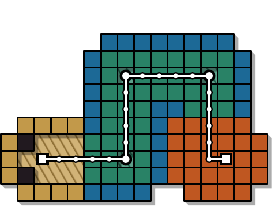}%
        \subcaption{}%
    \end{subfigure}%
    \hfill%
    \begin{subfigure}[t]{\columnwidth/3 -0.5em}%
        \centering%
        \includegraphics[width=\columnwidth,page=3]{snake/intro/snake-fork}%
        \subcaption{}%
    \end{subfigure}%
    \hfill%
    \begin{subfigure}[t]{\columnwidth/3 -0.5em}%
        \centering%
        \includegraphics[page=6,width=\columnwidth]{snake/intro/snake-fork}%
        \subcaption{}%
    \end{subfigure}%
    \caption{An illustration of forking and joining.
    A major vertex becomes the head vertex of a new snake in (b).
    If the union of their spine paths is itself a valid spine as in (c), two snakes can join.}
    \label{fig:fork-and-join}%
\end{figure}

\subparagraph*{Fork.}
The fork operation selects a prefix or suffix of the spine path in arbitrary orientation and teleports a constant number of interior modules along the spine to create a valid snake configuration along this spine with a chosen head vertex, as shown in~\cref{fig:fork-and-join}.
This takes~$\mathcal{O}(1)$ transformations due to~\cref{lem:snake-is-reversible}.
The new snake can then move independently.

\subparagraph*{Join.}
If the union of their spines constitutes a valid spine, two snakes can be joined at a common endpoint to form a larger snake.
This is the reverse of forking, see~\cref{fig:fork-and-join}.

This gives us a universal means of reconfiguration for the snake, as long as it forms a connected configuration with the remaining modules, as well as the fundamental tools to grow the initial snake into a larger structure by taking ownership of additional modules.

\subsection{Growing the snake}
\label{subsec:growing-the-snake}

Recall our goal is to grow the newly determined snake to own $\Theta(A+h)$ many modules, i.e., all modules in the tree $T_\module{s}\subseteq \mathcal{T}$ rooted at a module $\module{s}\in C$.
In this section, we will outline a depth-first-search approach to traversing the tree and taking ownership of all modules.

The snake configuration $S(P)$ determined due to~\cref{lem:snake-in-ball} is fully contained within an~$L_1$ ball of radius $12$ that is adjacent to modules of the subtree $T_\module{s}$.
We start by (i) moving the snake to align with the meta-cells of the $5\times 5\times 5$ grid on which the tree is based and then again (ii) to have its head located within the same such meta-cell as $\module s$.
Both of these steps can be realized using only the push and pull operations; by following any path from its initial position in $T_\module{s}$ to the root, this takes $\mathcal{O}(\abs{T_\module{s}})$ transformations.

Before showing that the snake can be moved along the tree, we establish that it can take ownership of free modules close to its major vertices; an illustration is given in~\cref{fig:grow-consume}.

\begin{restatable}[$\star$]{lemma}{lemmaGrowConsume}
    Let $\mathcal{S}$ be a snake configuration with spine path $P$ that has an empty cell in the tail section, and let $v\in P$ be a major vertex of the spine.
    A free module $\module m \in N^*_2[v]$ can be teleported into its interior by $\mathcal{O}(1)$ transformations, giving the snake ownership.
    \label{lem:grow-consume}
\end{restatable}
\begin{figure}[htb]
    \centering%
    \includegraphics{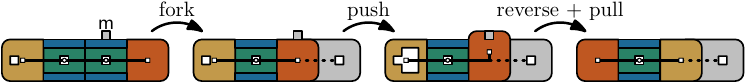}%
    \caption{By forking temporarily, snakes can take ownership of free modules in close proximity.}
    \label{fig:grow-consume}
\end{figure}

Due to~\cref{lem:grow-consume}, if we perform push and pull operations in multiples of five to keep the snake's major vertices aligned with the meta-cells of the grid, it suffices for any major vertex of the snake to be in the same meta-cell as a module it needs to take ownership of.
We show that using join and fork operations in addition to push and pull, the snake can take ownership of all modules within $T_\module{s}$ in $\mathcal{O}(\abs{T_\module{s}})$ transformations by traversing the tree depth-first and taking ownership of leaves.

\begin{restatable}[$\star$]{lemma}{lemmaDFSTraversal}
    \label{lem:ownership}
    Let $\module m_1, \ldots, \module m_\ell$ be the modules of $T_\module{s}$ in DFS order.
    A snake configuration~$\mathcal{S}$ that contains $\module m_1$ can take ownership of the entire sequence within $\mathcal{O}(\ell)$ transformations.
\end{restatable}

The high-level idea is to have the snake move to a leaf of the tree, which it takes ownership of before turning back toward the root $\module s$.
\begin{figure}[htb]
    \captionsetup[subfigure]{justification=centering}%
    \begin{subfigure}[t]{0.25\columnwidth}%
        \includegraphics[page=1]{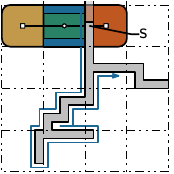}%
        \subcaption{}%
    \end{subfigure}%
    \begin{subfigure}[t]{0.25\columnwidth}%
        \includegraphics[page=2]{snake/grow/dfs}%
        \subcaption{}%
    \end{subfigure}%
    \begin{subfigure}[t]{0.25\columnwidth}%
        \includegraphics[page=3]{snake/grow/dfs}%
        \subcaption{}%
    \end{subfigure}%
    \begin{subfigure}[t]{0.25\columnwidth}%
        \includegraphics[page=4]{snake/grow/dfs}%
        \subcaption{}%
    \end{subfigure}
    \par\bigskip\hfil%
    \begin{subfigure}[t]{0.25\columnwidth}%
        \includegraphics[page=5]{snake/grow/dfs}%
        \subcaption{}%
    \end{subfigure}
    \begin{subfigure}[t]{0.25\columnwidth}%
        \includegraphics[page=6]{snake/grow/dfs}%
        \subcaption{}%
    \end{subfigure}
    \begin{subfigure}[t]{0.25\columnwidth}%
        \includegraphics[page=7]{snake/grow/dfs}%
        \subcaption{}%
    \end{subfigure}
    \caption{Example of our DFS procedure. In (a)--(b), the snake moves to a leaf of $T_\module{s}$, which it consumes before forking and pushing to the next meta-cell on its path in (c)--(d). This is repeated in (e)--(f), after which the forked spine is joined with the previous component in (g).}
    \label{fig:dfs-visualization}
\end{figure}
The sequence $\module m_1, \ldots, \module m_\ell$ could then be viewed as an Euler circuit around the tree $T_\module{s}$ with doubled edges.
Whenever the current snake cannot simply make a turn by reversing and moving, i.e., if its tail extends further up the tree than its head would need to be, it forks temporarily.
As a result, we need to track a tree of ``cast off'' spine paths that connect many tails with a single head; only one path in this tree is an active snake at all times.
A pseudocode description of how this operates is given in~\cref{alg:snake-dfs}, as well as an example covering several iterations in~\cref{fig:dfs-visualization}.

\begin{restatable}[htb]{algorithm}{algSnakeDFS}
    \caption{Procedure for the DFS-traversal of $T_{\module s}$.}
    \label{alg:snake-dfs}
    \textbf{Input:} A snake configuration with spine $v_h,\ldots, v_t$ and a sequence $\module m_1,\ldots, \module m_\ell$.\\
    \textbf{Output:} A schedule of makespan $\mathcal{O}(\ell)$ that gives the snake ownership of the sequence.
    \begin{algorithmic}[1]
        \If{$\module m_1\notin N^*_2[v_h,\ldots, v_t]$}
            \State $v_a\gets $ the first major vertex in an adjacent meta-cell to $\module m_1$.
            \If{$v_a$ is the head vertex $v_h\in P$}
                \State \textbf{push/pull} until head vertex is at the center of the meta-cell containing $\module m_1$.
            \Else
                \State \textbf{fork} the suffix $v_a,\ldots, v_t$ with $v_a$ as head vertex.
                \State \textbf{push/pull} until forked head is at the center of the meta-cell containing $\module m_1$.
            \EndIf
        \EndIf
        \If{$m_1$ is a leaf of $T_\module{s}$}
            \State \textbf{consume} $\module m_1$ due to~\cref{lem:grow-consume}.
        \EndIf
        \State \textbf{join} to a snake of greater size if possible.
        \State \textbf{recurse} with resulting snake and $\module m_2,\ldots, \module m_\ell$.
    \end{algorithmic}
\end{restatable}

\subsection{Scaffold and compress}
\label{subsec:scaffold-and-compress}

Due to space considerations, a detailed description of this phase is left for the appendix.
However, \cref{fig:scaffold-overview} provides intuition.
After growing the snake to an appropriate size, we can move to the top of $C$, and then build the projection of $C$ on this axis $z[C]$, this is our ``scaffolding''.
We turn this projection into a large set of snakes, all of which move towards the bottom of $C$, \newterm{consuming} modules of $C$. 
This turns $C$ into a series of vertical snakes, which we then turn into a compact configuration.

\begin{figure}[h!tb]
	\captionsetup[subfigure]{justification=centering}%
		\hfil
	\begin{subfigure}[t]{0.25\columnwidth}%
		\includegraphics[page=2, width=\textwidth]{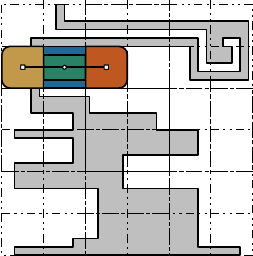}%
		\subcaption{}%
	\end{subfigure}	\hfil%
	\begin{subfigure}[t]{0.25\columnwidth}%
		\includegraphics[page=3, width=\textwidth]{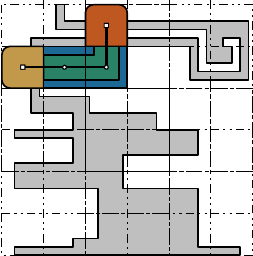}%
		\subcaption{}%
	\end{subfigure}	\hfil%
	\begin{subfigure}[t]{0.25\columnwidth}%
		\includegraphics[page=4, width=\textwidth]{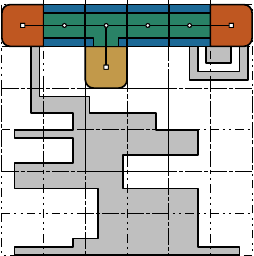}%
		\subcaption{}%
	\end{subfigure}	\hfil%
	\par\hfil%
	\begin{subfigure}[t]{0.25\columnwidth}%
		\includegraphics[page=5, width=\textwidth]{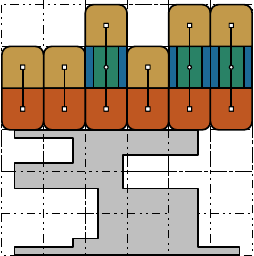}%
		\subcaption{}%
	\end{subfigure}	\hfil
	\begin{subfigure}[t]{0.25\columnwidth}%
		\includegraphics[page=6, width=\textwidth]{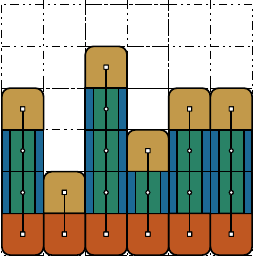}%
		\subcaption{}%
	\end{subfigure}	\hfil
	\begin{subfigure}[t]{0.25\columnwidth}%
		\includegraphics[page=7, width=\textwidth]{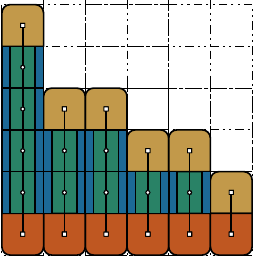}%
		\subcaption{}%
	\end{subfigure}	\hfil
	\caption{An overview of \textsc{Scaffold and Compress}. (a) We have a snake of size $\mathcal{O}(A + h)$. (b) Move the snake to the top of $C$. (c) Use \cref{alg:snake-dfs} to fill $z[C]$. (d-e) Compress the configuration with snakes. (f) Move the resulting towers towards the origin, compacting $C$.}
	\label{fig:scaffold-overview}
\end{figure}

The efficient reconfiguration between arbitrary compact configuration is well-studied in both sequential models~\cite{akitaya.demaine.korman.ea2022compacting-squares,kostitsyna.ophelders.parada.ea2024optimal-in-place} and parallel, two-dimensional models~\cite{a.akitaya_et_al:LIPIcs.ESA.2025.28,fekete.keldenich.kosfeld.ea2023connected-coordinated}.
The techniques from~\cite{a.akitaya_et_al:LIPIcs.ESA.2025.28} in particular have (not necessarily straight-forward) extensions to our three-dimensional setting.
For completeness, we show that compact configurations of meta-modules can be efficiently transformed into one another in the appendix.

\vspace*{-0.5ex}%
\begin{restatable}[$\star$]{lemma}{lemmaCompactReconfiguration}
    Any two compact configurations of meta-modules can be transformed into one another by an in-place schedule of makespan $\mathcal{O}(A+h)$, where $A$ and $h$ are the maximal projection area and maximal extent of the two along some axis $\lambda\in\{x,y,z\}$.
\end{restatable}

%% file: 05-conclusion.tex
\section{Conclusion}
\label{sec:conclusion}
We have provided several complexity and algorithmic results on the parallel motion of \slidingcubes in three dimensions.
In particular, we proved that this problem is not \FPT in makespan or symmetric difference size, as well as \logAPX-hard in general.
This is the first such inapproximability result among related models and emphasizes the relevance of input-sensitive, worst-case optimal algorithms; one such algorithm was also presented here.

There remain open questions.
The results of~\cref{sec:teleport} motivate the investigation of \slidingcubes between obstacles, which provide connectivity but cannot be moved.
Since we wanted to keep the presentation as simple as possible the bound for \cref{lem:teleport} can likely be improved to $\mathcal{O}(|S|^2)$.
Moreover, none of the techniques used for our complexity results~extend to lower-dimensional settings:
Is there a significant leap between two and three dimensions?

%% file: A03-complexity.tex
\section{Missing details and omitted proofs from~\cref{sec:complexity}}
\label{app:sec:complexity}

\subsection{\NP-completeness}

\theoremNPHardness*
\begin{proof}
    We reduce from \textsc{Planar Monotone 3Sat}, which asks whether a given Boolean formula is satisfiable~\cite{berg.khosravi2012optimal-binary}.
    Each clause consists of at most $3$ literals, all either positive or negative, and the clause-variable incidence graph must admit a plane drawing where variables are mapped to the $x$-axis, positive (resp., negative) clauses are mapped to the upper (resp., lower) half-plane, and edges do not cross the $x$-axis.
    Let $\deg(x_i)$ refer to the number of clauses that contain a positive literal over~$x_i$, and $\overline{\deg}(x_i)$ the same for negative literals.
    For the construction, starting from $\varphi$, we arrange multiple instances of three gadget types, along with a single module representing the symmetric difference as depicted in~\cref{app:fig:not-fpt-hardness-overview}.

    \begin{figure}[htb]%
        \includegraphics[width=\textwidth]{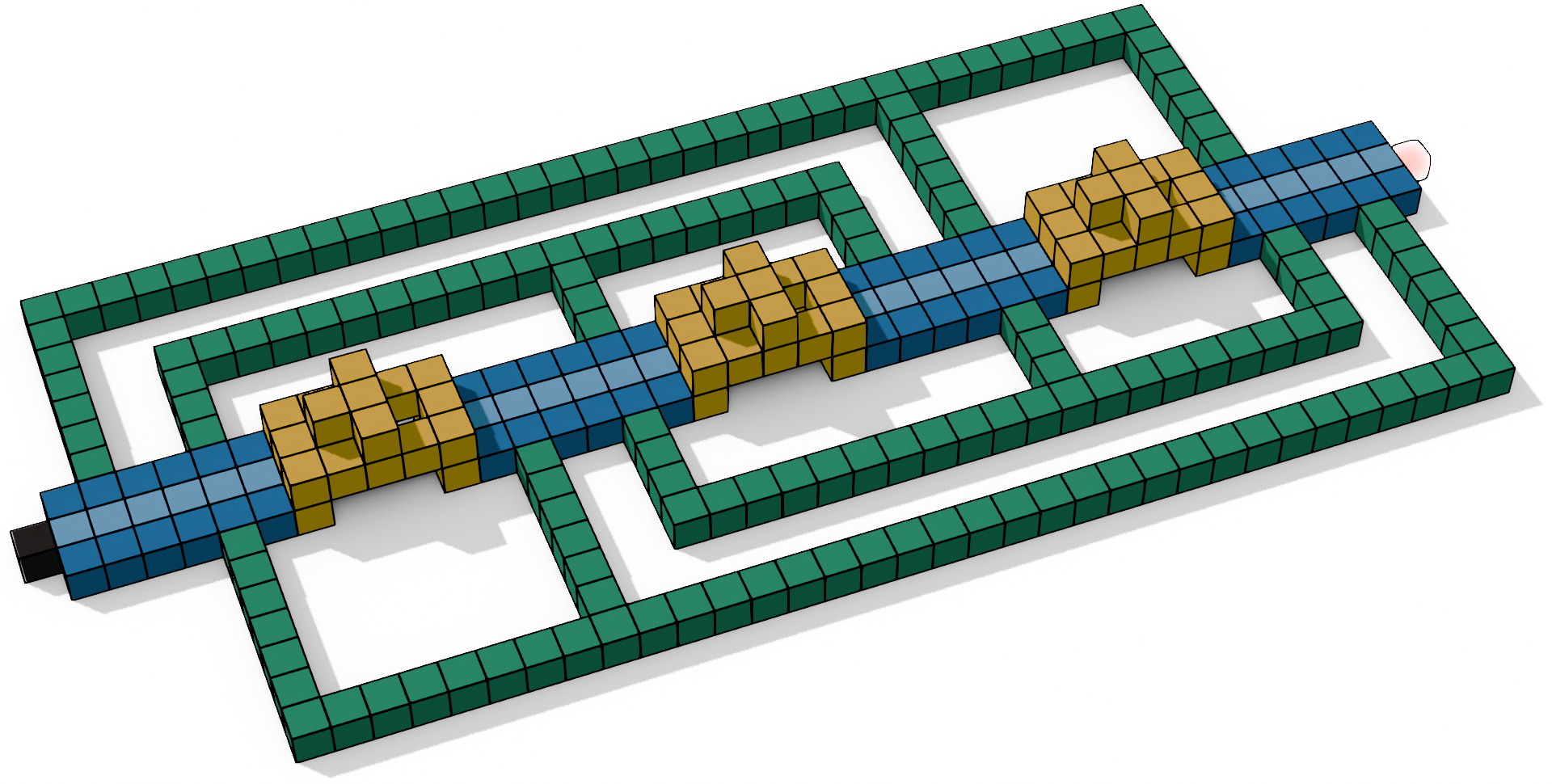}%
        \llap{\includegraphics[width=\textwidth]{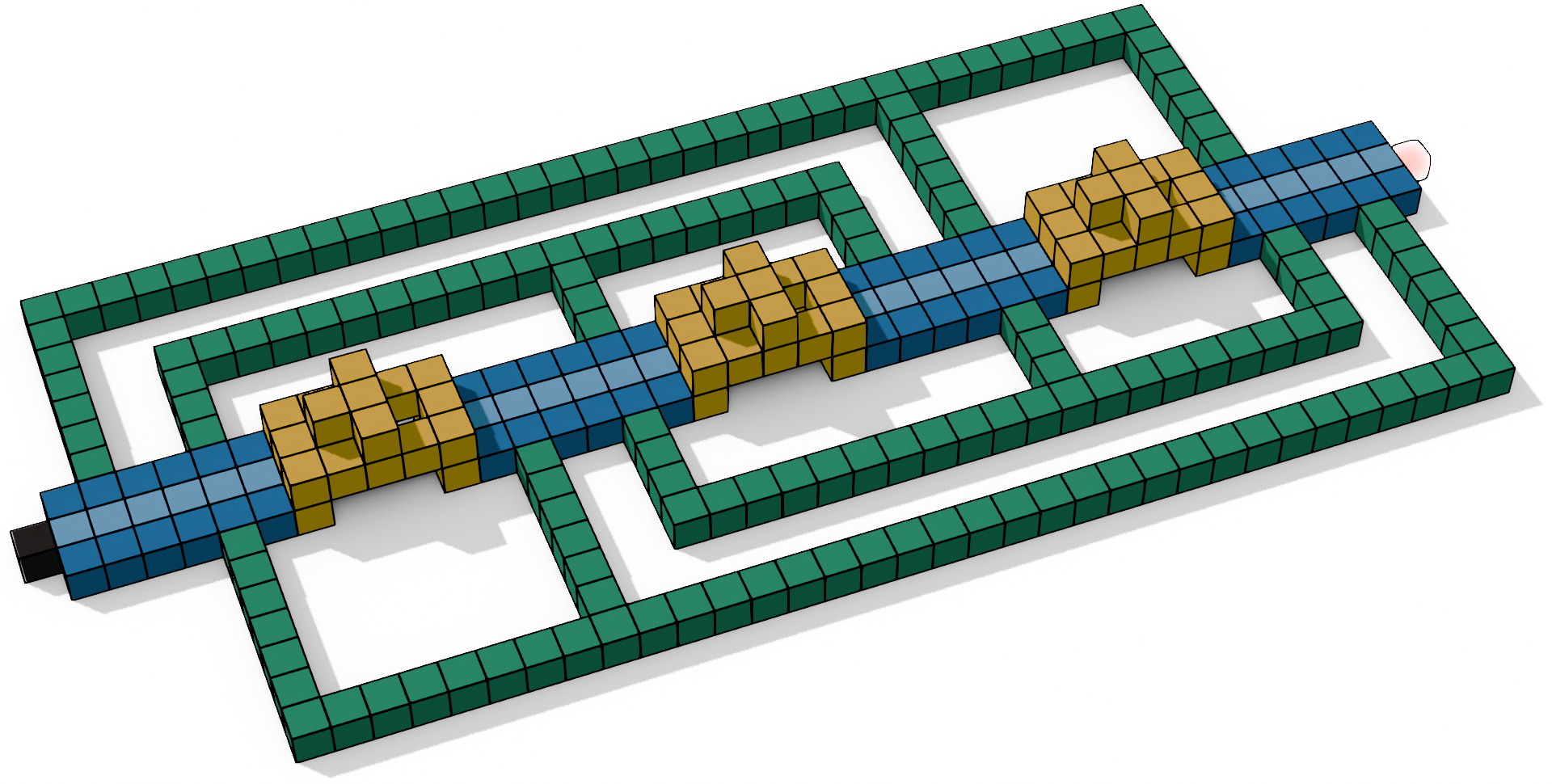}}%
        \caption{Our construction for $\varphi=(x_1\lor x_3\lor x_4)\land(x_1\lor x_2\lor x_3)\land(\overline{x_1}\lor\overline{x_2}\overline{x_4})\land(\overline{x_2}\lor\overline{x_3}\lor\overline{x_4})$. Colors indicate the gadget types, and variable gadgets are placed in ascending order left to right.}
        \label{app:fig:not-fpt-hardness-overview}%
    \end{figure}%

    More precisely, we start by placing $m$ \newterm{variable gadgets} (blue) in the $xy$-plane along the $x$-axis.
    Each variable gadget simply consists of modules arranged in a rectangle of depth exactly $3$, height exactly $1$, and width $3\cdot \max\{\deg(x_i),\overline{\deg}(x_i)\}$.
    We use their coordinate along the $y$-axis to differentiate between \newterm{positive} and \newterm{negative} modules in the variable gadgets; these will encode the Boolean value assigned to each variable during the transformation step.
    Variable gadgets are spaced $5$ units apart from one another along the $x$-axis, leaving space for a \newterm{connector gadget} (yellow) in between subsequent pairs.
    The connector gadget occupies a bounding box of width $5$ and depth and height $3$; it is the only inherently three-dimensional gadget in our construction.
    A detailed breakdown can be seen in~\cref{fig:hardness-connector-blowup}.
    \begin{figure}[htb]
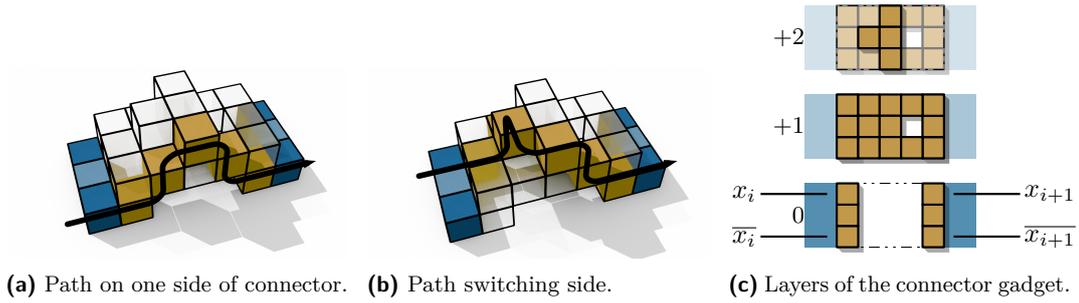
%
        \begin{subfigure}[b]{\columnwidth/3 - 0.5em}%
            \centering%
            \includegraphics[width=\columnwidth]{figures/sat-reduction/variable-connector-stay}%
            \llap{\includegraphics[width=\columnwidth]{figures/sat-reduction/variable-connector-repeat-path}}%
            \subcaption{Path on one side of connector.}
            \label{app:fig:not-fpt-hardness-connector-closeup-a}%
        \end{subfigure}%
        \hfill%
        \begin{subfigure}[b]{\textwidth/3 - 0.5em}%
            \centering%
            \includegraphics[width=\columnwidth]{figures/sat-reduction/variable-connector-swap}%
            \llap{\includegraphics[width=\columnwidth]{figures/sat-reduction/variable-connector-swap-path}}%
            \subcaption{Path switching side.}
            \label{app:fig:not-fpt-hardness-connector-closeup-b}%
        \end{subfigure}%
        \hfill%
        \begin{subfigure}[b]{\columnwidth/3 - 0.5em}%
            \centering%
            \includegraphics[page=3]{figures/sat-reduction/variable-connector-diagram}\par%
            \includegraphics[page=2]{figures/sat-reduction/variable-connector-diagram}\par%
            \includegraphics[page=1]{figures/sat-reduction/variable-connector-diagram}\par%
            \subcaption{Layers of the connector gadget.}
            \label{app:fig:hardness-connector-blowup}%
        \end{subfigure}%
        \caption{The connector gadget with indicated adjacent variable gadgets.}
        \label{app:fig:not-fpt-hardness-connector}
    \end{figure}

    Also attached to each variable gadget is at least one \newterm{clause gadget}~(green).
    Clause gadgets are ``comb-like'' structures that are adjacent to either a positive or negative module of the variable gadgets they contain.
    A positive literal implies the gadget has a face-adjacent module to a positive module of the respective variable gadget.
    Finally, we place a single \newterm{difference module} (black/transparent) at one end of the construction along the $x$-axis in the start configuration, and at the other end in the target configuration; note that this module is the only component that distinguishes the start from the target configuration.
    To reconfigure the resulting instance in a single transformation, this module effectively has to ``teleport'' through the configuration by pushing modules in an uninterrupted sequence of moves from the start to the target position.

    To analyze our construction, we argue about the nature of the path that this sequence of moves describes.
    Because the module that exists solely in the initial configuration must move immediately by performing a convex transition, either the positive or negative modules of the first variable must move as well.
    Moreover, this choice propagates uniformly within the variable gadget: all positive modules move in parallel, or all negative modules do.
    Furthermore, this implies that the position that appears only in the target configuration is occupied through a convex transition performed by a positive or negative module from the final variable gadget.
    Because the first variable ``pushes'' a module into the first connector, and the connector remains unchanged between the initial and target configurations, a module must necessarily exit the connector.
    Consequently, there must be a continuous and uninterrupted motion from the black module to the white position, meaning in particular that each variable gadget moves either all of its positive modules or all of its negative modules.

    It remains to describe the possible movement within the connector gadgets.
    First, assume that the positive modules from the first variable move, and that the positive modules of the second variable are intended to move as well.
    This causes the positive module from the first variable that is adjacent to the connector module to push into the connector through a sliding move.
    Symmetrically, the module leaving the connector exits through a sliding move on the same side.
    The displaced module must then perform a convex transition, and the module that moves into the vacated position must likewise execute a convex transition.
    The only legal sequence of moves allowing this behavior consists of two convex transitions on the same side, as indicated by the path in~\cref{app:fig:not-fpt-hardness-connector-closeup-a}.
    Now, assume that the negative modules from the second variable are intended to move.
    In this case, the movement within the connector proceeds differently, as depicted by the path in~\cref{app:fig:not-fpt-hardness-connector-closeup-b}.
    In particular, the modules in the topmost layer can rotate from one side to the other by means of convex transitions.
    As a result, the overall shape of the connector is preserved, but one of its modules moves to a position previously occupied by a negative module of the second variable.

    Due to the backbone constraint of our model, this implies that all clauses are disconnected from variable gadgets whose assignment does not match a contained literal during the transformation, and that any feasible path has a one-to-one correspondence to a satisfying assignment of $\varphi$.
    In particular, setting a variable to true triggers the movement of its negative variable modules.
    This detaches any clause connected on the negative side from that variable, requiring those clauses to be connected via some other variable that is set to false.
\end{proof}

We now argue the following:

\corollaryInapprox*

\begin{proof}
	It suffices to show that the instances constructed according to our reduction in~\cref{thm:unlabeled-sym-diff-hard} can be solved within two transformations without satisfying the underlying Boolean formula.
    Our strategy is as follows.
    We disregard all positive clauses and move modules exclusively along the negative sides of variable gadgets.

    In the first step, we move sequences of negative modules to the right and into the adjacent connector gadget such that each clause gadget remains connected to its index-minimal incident variable.
    In the second step, we repeat this pattern, instead maintaining index-minimal connections.
    Within two transformations, we can thus teleport the difference module from one end of the configuration to the other.
\end{proof}

\subsection{\logAPX-Hardness}

We now present the omitted details from our earlier overview of the \logAPX-hardness reduction.
includes a thorough description of the construction and its gadgets, followed by a careful analysis of the reduction as a whole.
For clarity, we divide the reduction into two parts. First, we analyze a simplified variant, referred to as the \textsc{Immobile} variant, which is easier to analyze. 
In this variant, we assume that the spike modules (depicted as green cubes in the figures) are immobile, meaning that they cannot move. 
We then show that the assumptions used in this simplified setting remain valid for the original variant once the gadgets are adjusted slightly.

\smallskip
Our reduction is from the \logAPX-hard problem \textsc{Set Cover}~\cite{DinurS14}. 
An instance $\Sigma_{n,k}$ consists of natural numbers $1,\ldots,n$ and a set $S=\{s_1,\ldots,s_k\}$ of subsets of $\{1,\ldots,n\}$ such that $\bigcup_{i=1}^k s_i = \{1,\ldots ,n\}$. 
The task is then to find a set $S^\star$ of minimum cardinality such that $\bigcup_{s_i \in S^\star} s_i = \{1,\ldots, n\}$. 
We use $\ell$ to denote the size of a subset $S'$ of $S$ and $\ell^\star$ to denote the size of the optimal solution $S^\star$.

We begin by describing all gadgets in detail.
In the figures, green cubes represent the spikes.
Although the figures depict these spikes extending three layers upward and downward, in reality they consist of many more cubes; the exact number will be specified later.
Throughout the proof, we denote the initial configuration of our sliding sequence by~$C_I$, and the final configuration by $C_F$.

\begin{description}
	\item[Cover gadget.] For every $j = 1,\ldots,n \in \mathbb{N}$, we introduce a \newterm{cover gadget}, depicted as a violet square in~\cref{fig:sc-reduction-diagram}, and in detail in~\cref{fig:cover-gadget}.
	In a cover gadget, the path of immobile modules is interrupted by a single ``mobile'' module $\module m_j$, highlighted in orange in~\cref{fig:cover-gadget}.
	In $C_I$, the orange module is at level $0$, whereas in $C_F$ it has been raised to level $+1$, directly above its original position.
	The path of immobile modules has spikes going up and down on both modules on the path directly next to the orange module; these are the only spikes that serve a purpose in the \textsc{Immobile} variant apart from immobilizing cubes.
	In the \textsc{Immobile} variant, we assume the spikes to be immobile modules.
	
	\begin{figure}[htb]
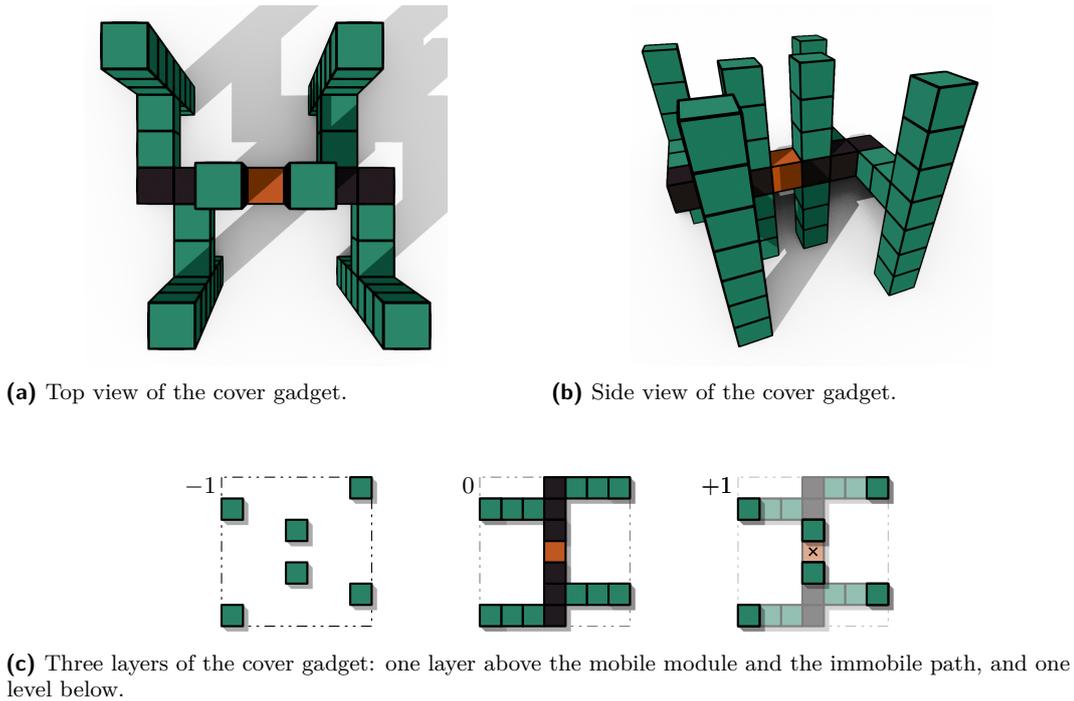
%
		\captionsetup[figure]{justification=centering}%
		\begin{subfigure}[b]{0.5\columnwidth-0.5em}%
			\centering%
			\includegraphics[width=0.7\columnwidth]{figures/sc-reduction/cover-top-view}%
			\subcaption{Top view of the cover gadget.}
		\end{subfigure}
		\hfill%
		\begin{subfigure}[b]{0.5\columnwidth-0.5em}%
			\centering%
			\includegraphics[width=0.7\columnwidth]{figures/sc-reduction/cover-nice-view}%
			\subcaption{Side view of the cover gadget.}
			\label{fig:coverb}
		\end{subfigure}%
		\par%
		\vspace*{0.75cm}%
		\begin{subfigure}[b]{\columnwidth}%
			\hfil%
			\includegraphics[page=1]{figures/sc-reduction/cover-gadget-diagram}\hspace{2em}%
			\includegraphics[page=2]{figures/sc-reduction/cover-gadget-diagram}\hspace{2em}%
			\includegraphics[page=3]{figures/sc-reduction/cover-gadget-diagram}%
			\subcaption{Three layers of the cover gadget: one layer above the mobile module and the immobile path, and one level below.}
		\end{subfigure}
		\caption{The cover gadget. The orange cube describes is the mobile module, with starting position at layer $0$ and target position in layer $+1$. The dark cubes describe the immobile modules. Green cubes correspond to spikes.}
		\label{fig:cover-gadget}
	\end{figure}
	\item[Layer gadget.] To extend~\cref{fig:sc-reduction-diagram} into the third dimension, we introduce~\newterm{layer gadgets}, depicted as a blue box marked with a black arrow; see also~\cref{fig:patha}.
	Layer gadgets signify that the path bends upward, rises by three levels, and subsequently resumes its initial direction.
	In~\cref{fig:sc-reduction-diagram}, the arrow indicates the transition from the lower part of the path to the higher part.
	\begin{figure}[htb]%
		\captionsetup[figure]{justification=centering}%
		\begin{subfigure}[b]{\columnwidth/2 - 0.5em}%
			\centering%
			\includegraphics[width=\columnwidth*2/3]{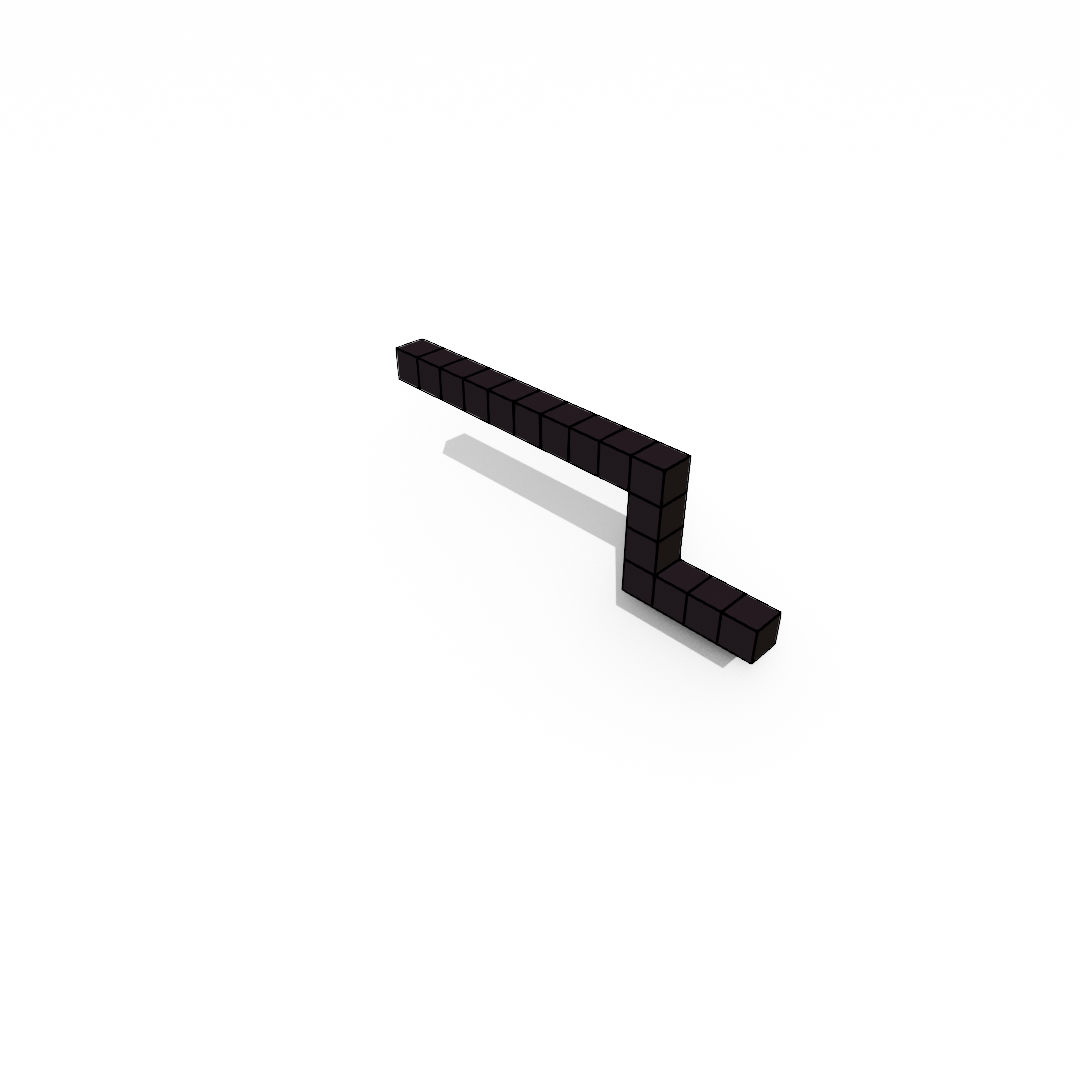}%
			\subcaption{A path climbing three units...}
			\label{fig:patha}
		\end{subfigure}
		\hfill%
		\begin{subfigure}[b]{\columnwidth/2 - 0.5em}%
			\centering%
			\includegraphics[width=\columnwidth*2/3]{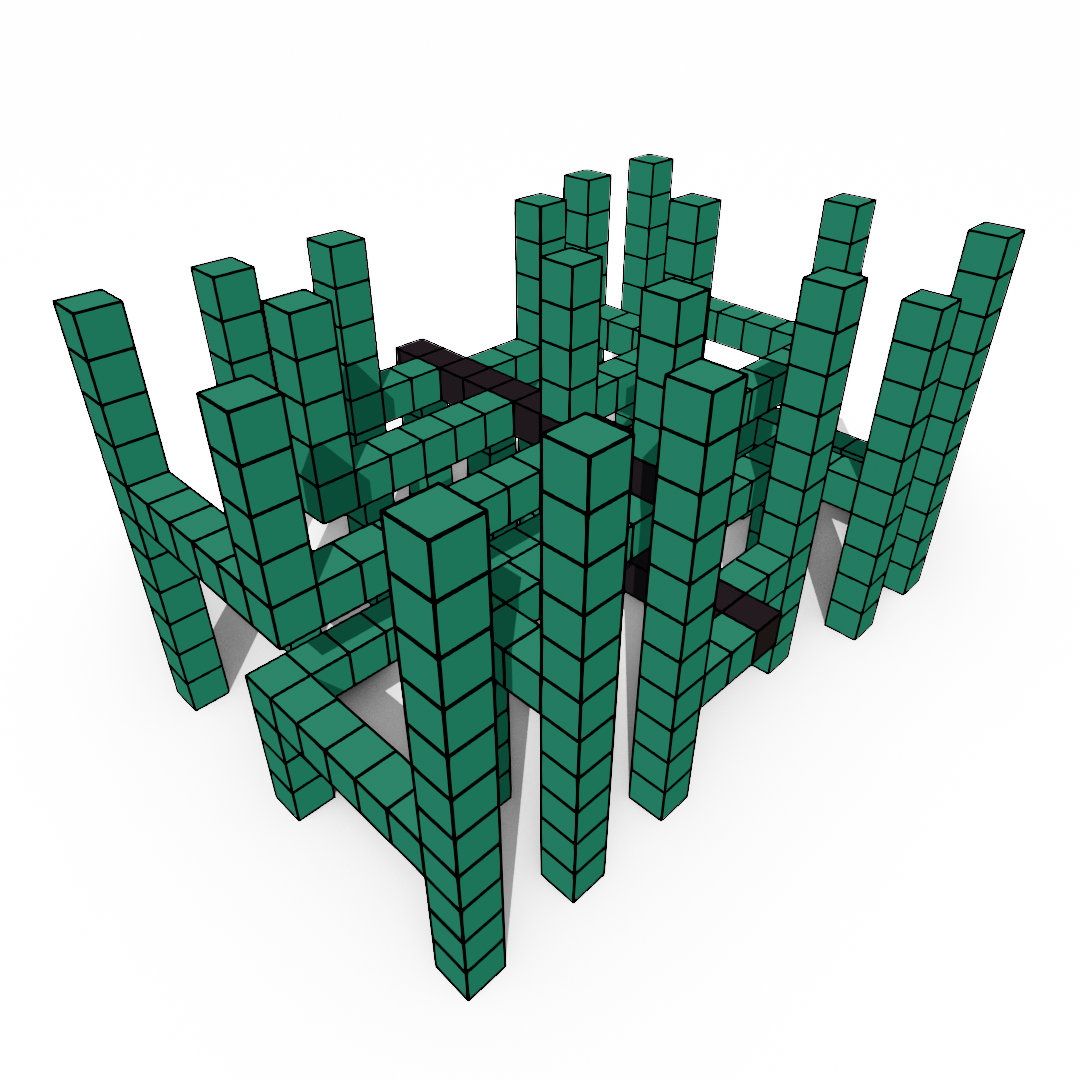}%
			\subcaption{...with lots of spikes attached.}
			\label{fig:pathb}
		\end{subfigure}%
		\par%
		\vspace*{0.75cm}%
		\begin{subfigure}[b]{\columnwidth}%
			\hfil\includegraphics[page=1]{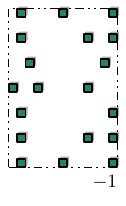}
			\hfil\includegraphics[page=2]{figures/sc-reduction/layer-gadget-diagram}%
			\hfil\includegraphics[page=3]{figures/sc-reduction/layer-gadget-diagram}
			\hfil\includegraphics[page=4]{figures/sc-reduction/layer-gadget-diagram}%
			\hfil\includegraphics[page=5]{figures/sc-reduction/layer-gadget-diagram}
			\hfil\includegraphics[page=6]{figures/sc-reduction/layer-gadget-diagram}%
			\subcaption{Three layers of the path gadget.}
		\end{subfigure}
		\caption{The layer gadget connects path gadget across vertical distances (in multiples of three).}
		\label{fig:layer-gadget}
	\end{figure}
	\item[Selector gadget.] The \newterm{selector gadget} is depicted as a blue square.
	This gadget is used to represent the containment of a number $j$ from $1,\ldots, n$ in a set $s_i$.
	In a selector gadget, one path of immobile modules is oriented orthogonally to a second path.
	The terminal module of the first path is placed at distance $1$ from the second path.
	Each path connects to the cover gadget associated with the number $j$, but from different directions.
	Moving a mobile robot to connect the two paths establishes the backbone constraint that in turn enables the orange module of the cover gadget to move.
	
	\begin{figure}[htb]
		\captionsetup[figure]{justification=centering}%
		\begin{subfigure}[b]{0.5\columnwidth-0.5em}%
			\centering%
			\includegraphics[width=0.7\columnwidth]{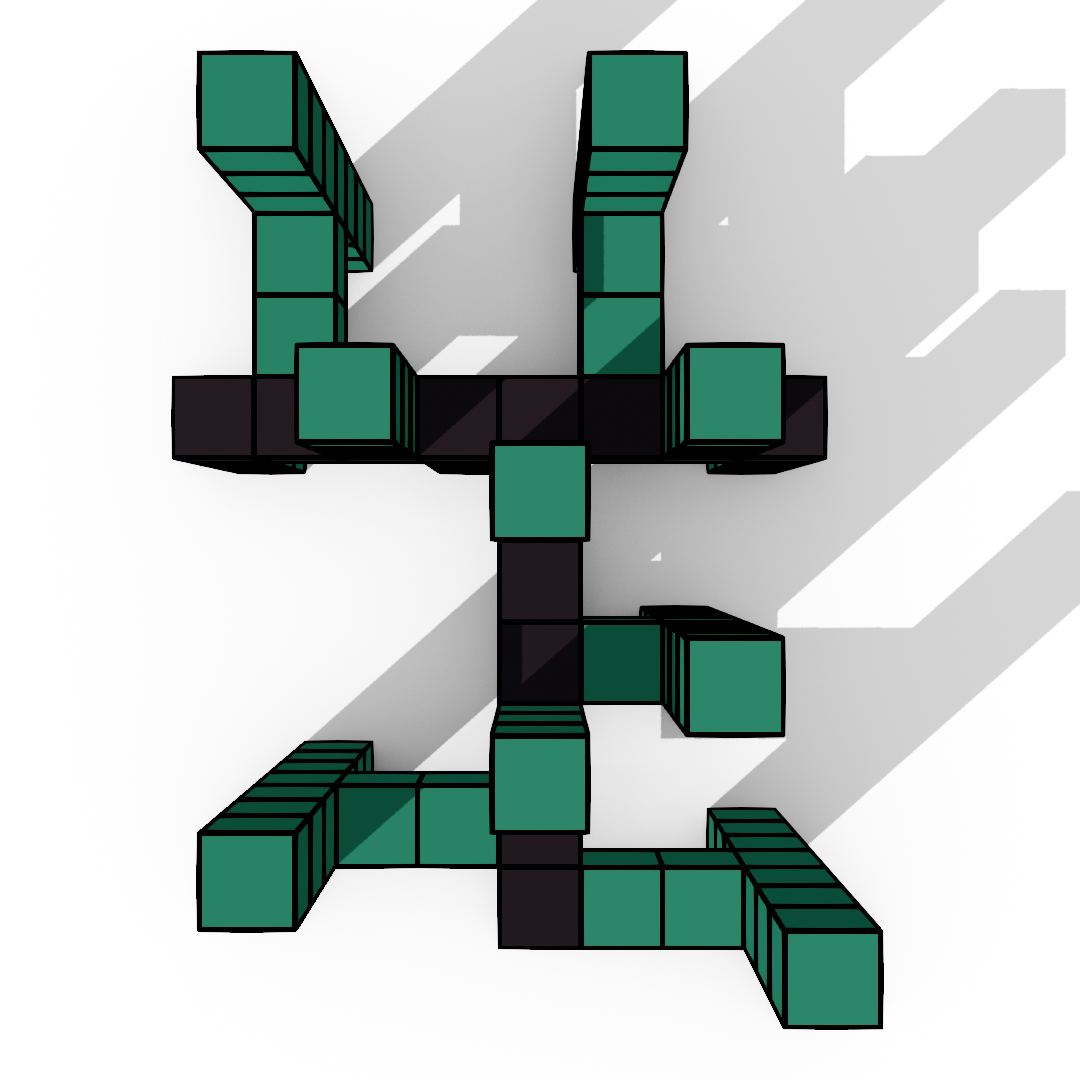}%
			\subcaption{A tee-crossing gadget.}
		\end{subfigure}
		\hfill%
		\begin{subfigure}[b]{0.5\columnwidth-0.5em}%
			\centering%
			\includegraphics[width=0.7\columnwidth]{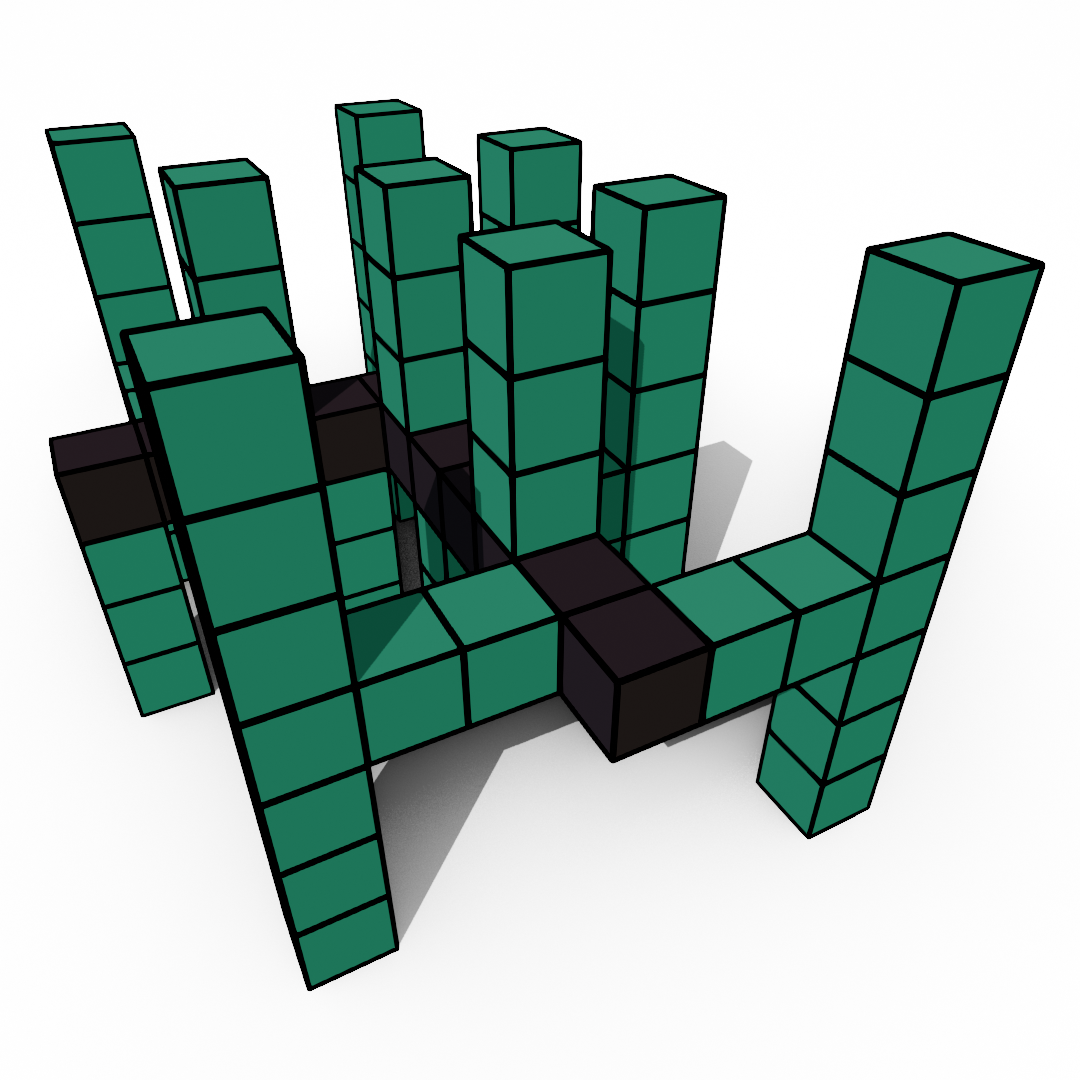}%
			\subcaption{A nicer view.}
		\end{subfigure}%
		\par%
		\vspace*{0.75cm}%
		\begin{subfigure}[b]{\columnwidth}%
			\hfil%
			\includegraphics[page=1]{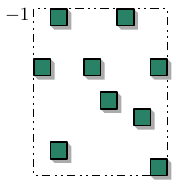}\hspace{2em}%
			\includegraphics[page=2]{figures/sc-reduction/tee-gadget-diagram}\hspace{2em}%
			\includegraphics[page=3]{figures/sc-reduction/tee-gadget-diagram}%
			\subcaption{Three layers of the tee-crossing gadget.}
		\end{subfigure}
			\caption{Tee gadgets connect three path gadgets meeting in a common point.}
			\label{fig:tee-gadget}
	\end{figure}
	
	\begin{figure}[htb]
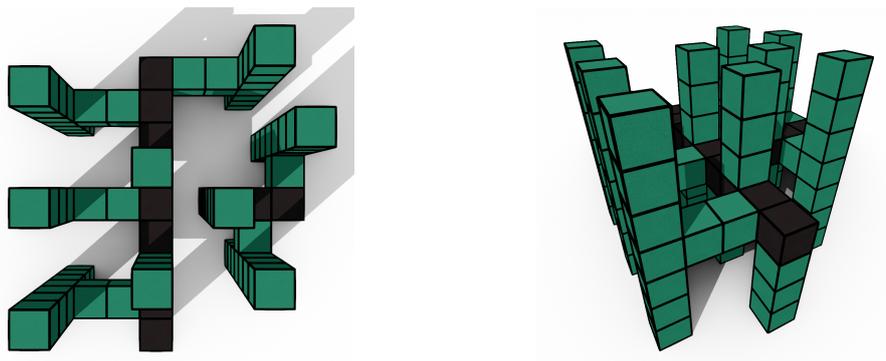

		\captionsetup[figure]{justification=centering}%
		\begin{subfigure}[b]{0.5\columnwidth-0.5em}%
			\centering%
			\includegraphics[width=0.7\columnwidth]{figures/sc-reduction/selector-top-view}%
			\subcaption{A selector gadget.}
		\end{subfigure}
		\hfill%
		\begin{subfigure}[b]{0.5\columnwidth-0.5em}%
			\centering%
			\includegraphics[width=0.7\columnwidth]{figures/sc-reduction/selector-nice-view}%
			\subcaption{A nicer view.}
		\end{subfigure}%
		\par%
		\vspace*{0.75cm}%
		\begin{subfigure}[b]{\columnwidth}%
			\hfil%
			\includegraphics[page=1]{figures/sc-reduction/selector-gadget-diagram}\hspace{2em}%
			\includegraphics[page=2]{figures/sc-reduction/selector-gadget-diagram}\hspace{2em}%
			\includegraphics[page=3]{figures/sc-reduction/selector-gadget-diagram}%
			\subcaption{Three layers of the selector gadget.}
		\end{subfigure}
		\caption{Selector gadgets enable the cover gadgets to perform their operation.}
		\label{fig:selector-gadget}
	\end{figure}
	
	\item[Number gadget.] To encode a number $j$, we construct a chain gadgets (as illustrated in the gray rectangle labeled $1$ in~\cref{fig:sc-reduction-diagram}) along a horizontal path.
	On its right side, it includes a cover gadget. 
	The boxes with orange hatching serve as placeholders, indicating that the actual path is considerably longer than depicted; the figure compresses it for readability.
	For every set that contains $j$, an additional vertical path towards the set the set gadgets (gray rectangles labeled $s_i$) branches off from the main horizontal path via a~\newterm{tee-crossing}.
	
	\item[Combining all numbers $1,\ldots,n$:] Each number  $1,\ldots,n$ is represented by a dedicated horizontal path each. 
	These paths are connected on their right side to a single comb. 
	The largest number is placed foremost, occupying the lowest layer among all number gadgets.
	Between any two consecutive number gadgets, we insert a layer gadget so that the gadget for  $j-1$ is positioned three layers above the gadget for $j$; see the light-blue squares along the rightmost vertical path in~\cref{fig:sc-reduction-diagram}. 
	This ensures that the paths leading to the selector gadgets originate at distinct vertical positions for each number gadget.
	\pagebreak
	\item[Set gadget:] To represent a set $s_i \in S$, we introduce the \newterm{set gadget} (marked with a gray rectangle labeled $s_i$).
	The set gadget consists of a path that is drawn vertically in~\cref{fig:sc-reduction-diagram}, from the bottommost row to the gray rectangles.
	At its upper end, it hosts selectors connected to number gadgets, with a selector for number $j$ present precisely when $j\in s_i$.
	They are ordered so that selectors higher on the path correspond to larger numbers.
	
	\item[Combining set gadgets:] The bottommost horizontal path in~\cref{fig:sc-reduction-diagram}, together with the set gadgets, constitutes a comb structure.
	The set gadgets are attached from left to right in the order prescribed by their indices.
	In $C_I$, the mobile module $\module m$ sits at the corner where the gadget for $s_1$ attaches, whereas in $C_F$ it sits at the tee-crossing where the gadget for $s_k$ attaches to the path.
	
	\item[Putting it all together:] The comb formed by the set gadgets and the comb formed by the number gadgets intersect at a single point (shown at the bottom right of \cref{fig:sc-reduction-diagram}).
	The cube where the two paths meet lies at a considerable distance from all set gadgets as well as all number gadgets.
	This large spacing is represented by the orange-hatched rectangles.
\end{description}

Next, we discuss the different spatial separations among the gadgets as depicted in~\cref{fig:sc-reduction-diagram}.

\begin{itemize}
	\item The distance between two consecutive selector gadgets is set to $B=30$.
	As a result, the vertical distance from a number gadget to a selector gadget is also $B$.
	For paths that connect to number $k$, the $90^\circ$ turns, that is where the paths changes from horizontal to vertical, occur after a horizontal distance of $B(k+1-j)$.
	This ensures that, after the turn, the vertical distance between consecutive paths remains $B$.
	\item To prevent interference, the vertical distance between two set gadgets is chosen so that all $90^\circ$ turns of selectors from the previous gadget lie further to the left. This distance is set to $C = 2\cdot B\cdot k$.
	\item The paths depicted as black-hatched rectangles are of length $A=k\cdot C + 2k(n-1)B + n$.
	\item The paths depicted as orange-hatched rectangles are of length $100\cdot k \cdot A$. Note that this length also specifies how far the spike gadgets extend above and below the path.
\end{itemize}

We provide several results that allow us to draw conclusions about the computational complexity of determining how many steps are required to reconfigure $C_I$ into $C_F$ in the \textsc{Immobile} variant.
Before moving forward, we would like to make a few remarks.
(1) The dual graph of $C_I$, with vertices correspond to modules and two vertices are connected exactly if their corresponding modules share a common 2-dimensional face, forms a tree.
(2) $\module m$ is the only mobile module that can perform a move in~$C_I$.
It is trivial that $\module m_j$ cannot perform a slide because doing so would separate the tree into two parts, thereby violating the backbone constraint.
To allow an $\module m_j$ gadget to move, we can move~$\module m$ to a selector gadget and use it to fill the gap in between the two paths of that separator.
By doing so, we create a cycle that contains both $\module m$ and $\module m_j$, so that $\module m_j$ can perform a slide onto its respective target position without violating the backbone constraint.

We now establish a simple upper bound on the number of moves required to reconfigure $C_I$ into~$C_F$ in the \textsc{Immobile} variant, which lets us rule out certain events in the shortest sequence, since performing them would already exceed this bound.

\begin{lemma}
	\label{lem:trivialbound}
	In the \textsc{Immobile} variant, there exists a sequence of parallel slides from $C_I$ to $C_F$ of length $2\cdot k \cdot A + k\cdot (C+6) + 2k(n-1)(B+6)+n$.
\end{lemma}

\begin{proof}
	The mobile module $\module m$ is initially placed in the bottom-left corner, where the set gadget of $s_1$ connects to the horizontal path at the bottom. 
	We perform a sequence of sequential moves that moves $\module m$ to the first selector of the number $j_1$ within the set gadget of~$s_1$. 
	Next, $\module m$ moves into the gap between the two paths in the selector gadget. 
	This closes a cycle that also includes the module $\module m_{j_1}$. 
	Consequently, we can perform a slide that moves $\module m_{j_1}$ from its initial position to its final position, as $\module m$ provides a backbone.
	After sliding $\module m_{j_1}$ to its final position, we remove $\module m$ from the gap and move it to the next selector gadget of some number $j_2$. 
	Again, we move $\module m$ into the gap between the two paths of the selector gadget, move $\module m_{j_2}$ to its final position, and remove $\module m$ from the gap.
	We repeat this for all selector gadgets within the set gadget of $s_1$. 
	Once all selector gadgets have been processed, we move $\module m$ back to the point where the set gadget of $s_1$ connects to the horizontal path.
	
	We continue by moving $\module m$ to the point were the set gadget of $s_2$ connects to the horizontal path. 
	The entire set gadget of $s_2$ is processed in the same way as $s_1$, with the exception that we skip selector gadgets linked to numbers already handled in $s_1$. 
	This procedure is repeated for all set gadgets corresponding to every $s\in S$. 
	In the end, $\module m$ reaches its final position, and because $\bigcup_{i=1}^k s_i = \{1,\ldots, n\}$, every $\module m_j$ in a cover gadget has been moved to its final position.
	
	Moving $\module m$ up and down a single set gadget requires $2\cdot A$ moves, and $2\cdot k \cdot A$ in total for all set gadgets. 
	Additionally, accessing all selector gadgets within a single set gadget takes at most  $2(n-1)(B+6)$ moves, summing to $2k(n-1)(B+6)$ across all set gadgets. 
	The ``$+6$'' accounts for spikes pointing upwards, which require an extra slide to move past, once when traversing the gadget from bottom to top, and once from top to bottom.
	Two more slides are needed for inserting and removing $\module m$ into the gap between the two paths of a selector gadget. 
	Moving from the attachment point of one set gadget to another requires $(C+6)$ moves, adding up to $(k-1)(C+6)$ in total. 
	Here, the ``$+6$'' accounts again from moving towards spikes pointing upwards in turn and tee gadgets. 
	Finally, moving each $\module m_j$ from its initial position to its target position requires one slide per module, giving $n$ moves in total.
\end{proof}

From~\cref{lem:trivialbound}, we derive the following crucial property of any shortest sliding sequence between $C_I$ and $C_F$.

\begin{lemma}
	In any shortest sliding sequence in the \textsc{Immobile} variant, any $\module m_j$ for any $j=1,\ldots,n$ is involved in exactly one sliding operation, namely the slide from its initial position to its target position.
\end{lemma}

\begin{proof}
	We first note that the term ``its target position'' is indeed well-defined:
	indeed, for distinct $j_1\neq j_2$, sliding a mobile module from the initial position of $j_1$ and the target position of $j_2$ necessarily requires moving it across two orange-hatched rectangles, as visualized in~\cref{fig:sc-reduction-diagram}.
	This requires at least $200\cdot k \cdot A$ many slides, which exceeds the length of the trivial upper bound given in~\cref{lem:trivialbound}.
	
	Consider an arbitrary sliding sequence from $C_I$ to $C_F$ that is strictly shorter than the trivial one.
	Within this sequence, the mobile module~$\module m_j$ may participate in multiple sliding operations.
	However, it can only be involved while $\module m$ sits in a gap of a selector gadget linked to the number gadget for $j$.
	For each set $i$ containing the number $j$, let $\selector(i,j)$ denote the corresponding selector gadget.
	Let $\selector(\cdot,j)$ be the set of all selector gadgets connected to the number gadget of $j$ for which $\module m$ has bridged the gap between the two paths of the selector gadget.
	For every $j = 1,\ldots,n$, we choose exactly one selector gadget $\selector(i,j)$ and include it in the set $\selector^\star$.
	We now construct a sliding sequence as follows:
	Initially, $\module m$ is in its starting position.
	We slide $\module m$ from its initial position to its final position along the horizontal path of the comb formed by all set gadgets.
	Whenever $\module m$ reaches a set gadget that contains a selector gadget $\selector(i,j) \in \selector^\star$, we slide $\module m$ upwards into the set gadget.
	Within it, we place $\module m$ into each selector gadget belonging to the set gadget to connect the two paths.
	This enables us to move every $\module m_j$ to its target position using a single slide per orange module.
	
	The resulting sliding sequence ascends no more set gadgets, uses no more selector gadgets, and performs no more slides of orange mobile modules than the original sequence.
	Moreover, it performs strictly fewer slides of orange modules whenever any orange module participated in more than one slide in the original sequence.
	Consequently, the length of our constructed sliding sequence is at most that of the original one, and strictly smaller whenever the original sequence slid any orange module more than once.
	
	We conclude that in any shortest sliding sequence any orange mobile module is part of at most one slide.
\end{proof}

\begin{proposition}
	\label{prop:nphard}
	In the \textsc{Immobile} variant, there exists a sequence of parallel slides from $C_I$ to $C_F$ of length $2\cdot \ell\cdot A + (k-1)\cdot (C+6) + 2k(n-1)(B+6)+n$ if and only if $\Sigma_{n,k}$ admits a cover of size $\ell$.
\end{proposition}

\begin{proof}
	Suppose $\Sigma_{n,k}$ admits a cover $S'$ of size $\ell$.
	We perform the following sequence of moves.
	The free module moves from left to right:
	whenever it reaches an set gadget corresponding to a set in $S'$, it ascends the gadget and performs a move in every selector gadget linked to cover gadgets that have not yet been transformed.
	
	Moving from left to right requires at most $k\cdot (C+6)$ moves.
	Transforming the selector gadgets then takes one move per element, for a total of $n$ moves.
	Ascending a set gadget takes $A$ slides, and descending it requires an additional $A$ slides. The module repeats this process for $\ell$ set gadgets.
	This gives a total of $2\cdot \ell\cdot A$ slides for set gadgets.
	Within a set gadget, the module must still visit all of the selector gadgets.
	This may take up to $2(n-1)(B+6)$ slides per set gadget, so at most $2k(n-1)(B+6)$ slides in total.
	
	Now assume that the sliding distance between $C_I$ and $C_F$ is at most $2\cdot \ell\cdot A + k\cdot (C+6) + 2k(n-1)(B+6)+n$. 
	To make all necessary slides within the cover gadgets, the mobile module $\module m$ must reach the selector gadgets containing the respective numbers. 
	To do so, it must walk up and down the set gadgets.
	Given the total length of the sliding sequence, it can ascend and descend at most 
	\begin{align*}
		&\leq \bigg\lfloor \frac{2\cdot \ell\cdot A + k\cdot (C+6) + 2k(n-1)(B+6)+n}{2\cdot A} \bigg\rfloor \\ &= \bigg\lfloor \frac{2\cdot \ell\cdot A}{2\cdot A} + \frac{ k\cdot (C+6) + 2k(n-1)(B+6)+n}{2\cdot A} \bigg\rfloor\\
		&= \bigg\lfloor \frac{2\cdot \ell\cdot A}{2\cdot A} + \frac{ k\cdot (C+6) + 2k(n-1)(B+6)+n}{2\cdot ( k\cdot (C+6) + 2k(n-1)(B+6)+n)} \bigg\rfloor  = \ell
	\end{align*}
	many set gadgets. 
	Let $S'\subset S$ denote the sets whose set gadgets the module has walked up and down. 
	Because all $m_j$ were able to move to their target positions, $S'$ forms a cover of size at most $\ell$.
\end{proof}

\begin{remark}
	\label{rem:size}
	Rewriting the estimates in~\cref{prop:nphard} shows that, given a sliding distance~$d$, we can determine the size of the cover $\ell$ by calculating $\ell=\lfloor \frac{d}{2 \cdot A}\rfloor$, or equivalently $\ell = \frac{d-R}{2 \cdot A}$ for some positive $R<((k-1)\cdot (C+6) + 2k(n-1)(B+6)+n)$.
	In particular, $\ell^\ast = \frac{d^\ast-R^\ast}{2 \cdot A}$ with $R^\ast<((k-1)\cdot (C+6) + 2k(n-1)(B+6)+n)$.
\end{remark}

Clearly, any valid sliding sequence in the \textsc{Immobile} variant is also valid for \parallelcubes.
What remains to be shown is that allowing multiple slides in parallel, without the immobile modules, does not lead to any improvement.
To this end, we attach a \newterm{spike} to many of the immobile modules such that removing an immobile module would disconnect its spike from the rest of the configuration, thereby violating the backbone constraint.
The simplest illustration of spike placement appears in~\cref{fig:path-gadget}, where spikes are positioned along a straight path.
Here, the placement alternates in a three-way pattern: one spike extending left, the next extending right, and the following extending downward.
This alternating pattern, however, cannot be applied uniformly throughout the construction.
At turns, as shown in~\cref{fig:turn-gadget}, we temporarily alternate spikes pointing up, down, and outward from the turn to maintain connectivity.

\begin{figure}[htb]
	\captionsetup[figure]{justification=centering}%
	\begin{subfigure}[b]{0.5\columnwidth-0.5em}%
		\centering%
		\includegraphics[width=0.7\columnwidth]{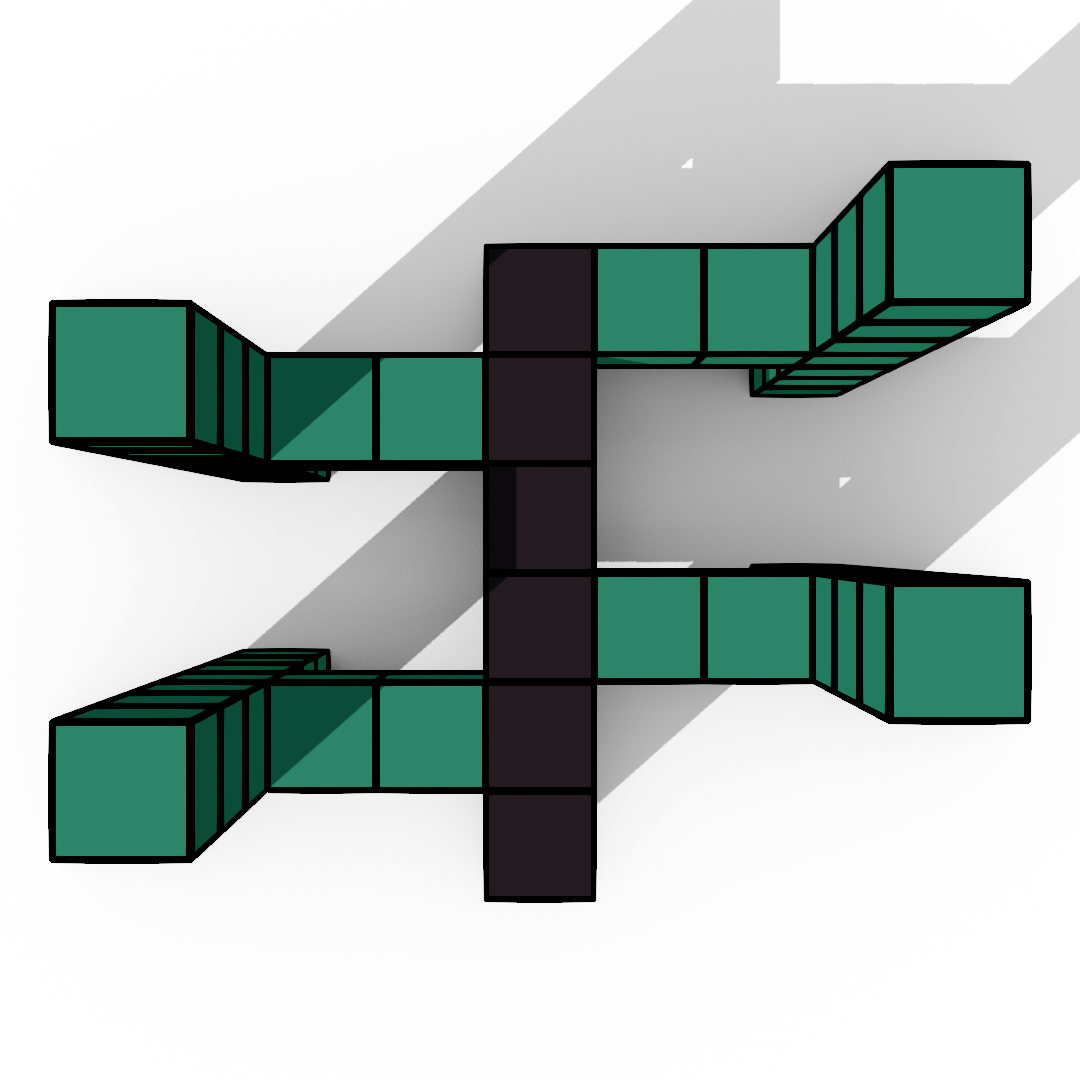}%
		\subcaption{A path gadget.}
	\end{subfigure}
	\hfill%
	\begin{subfigure}[b]{0.5\columnwidth-0.5em}%
		\centering%
		\includegraphics[width=0.7\columnwidth]{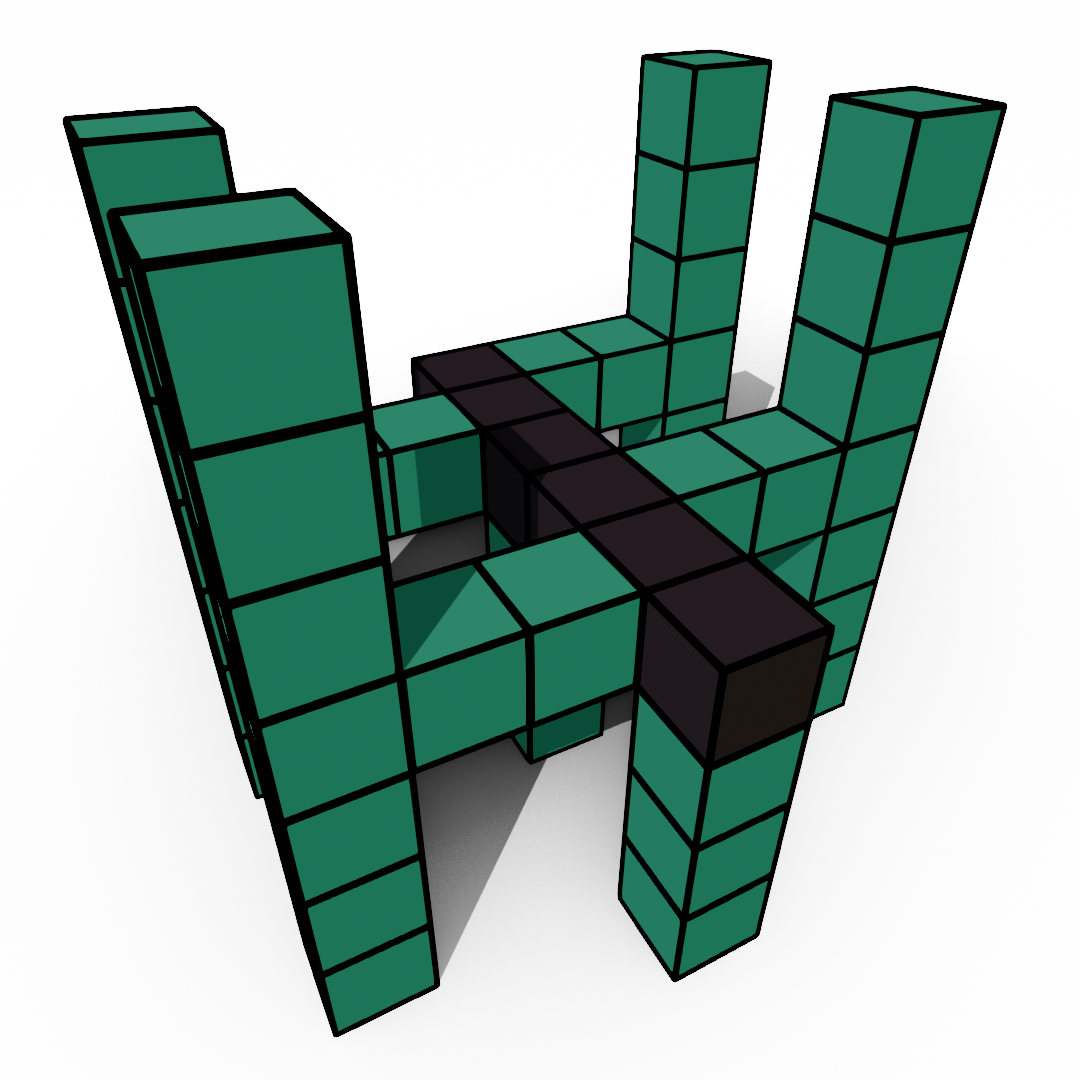}%
		\subcaption{A nicer view.}
	\end{subfigure}%
	\par%
	\vspace*{0.75cm}%
	\begin{subfigure}[b]{\columnwidth}%
		\hfil%
		\includegraphics[page=1]{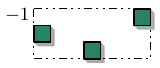}\hspace{2em}%
		\includegraphics[page=2]{figures/sc-reduction/path-gadget-diagram}\hspace{2em}%
		\includegraphics[page=3]{figures/sc-reduction/path-gadget-diagram}%
		\subcaption{Three layers of the path gadget.}
	\end{subfigure}
		\caption{The base component of our reduction is the path gadget.}
		\label{fig:path-gadget}
\end{figure}

\begin{figure}[htb]
	\captionsetup[figure]{justification=centering}%
	\begin{subfigure}[b]{0.5\columnwidth-0.5em}%
		\centering%
		\includegraphics[width=0.7\columnwidth]{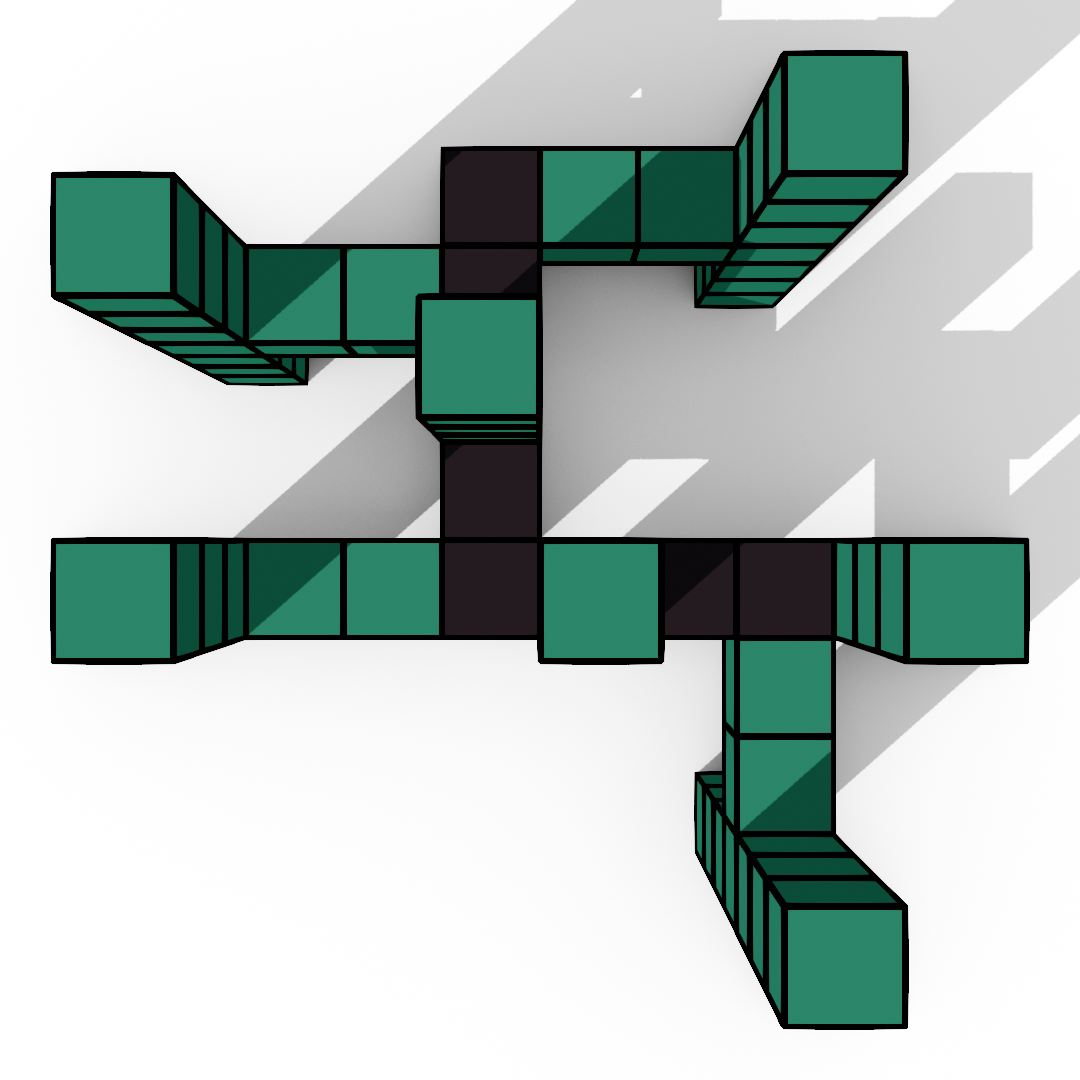}%
		\subcaption{A turn gadget.}
	\end{subfigure}
	\hfill%
	\begin{subfigure}[b]{0.5\columnwidth-0.5em}%
		\centering%
		\includegraphics[width=0.7\columnwidth]{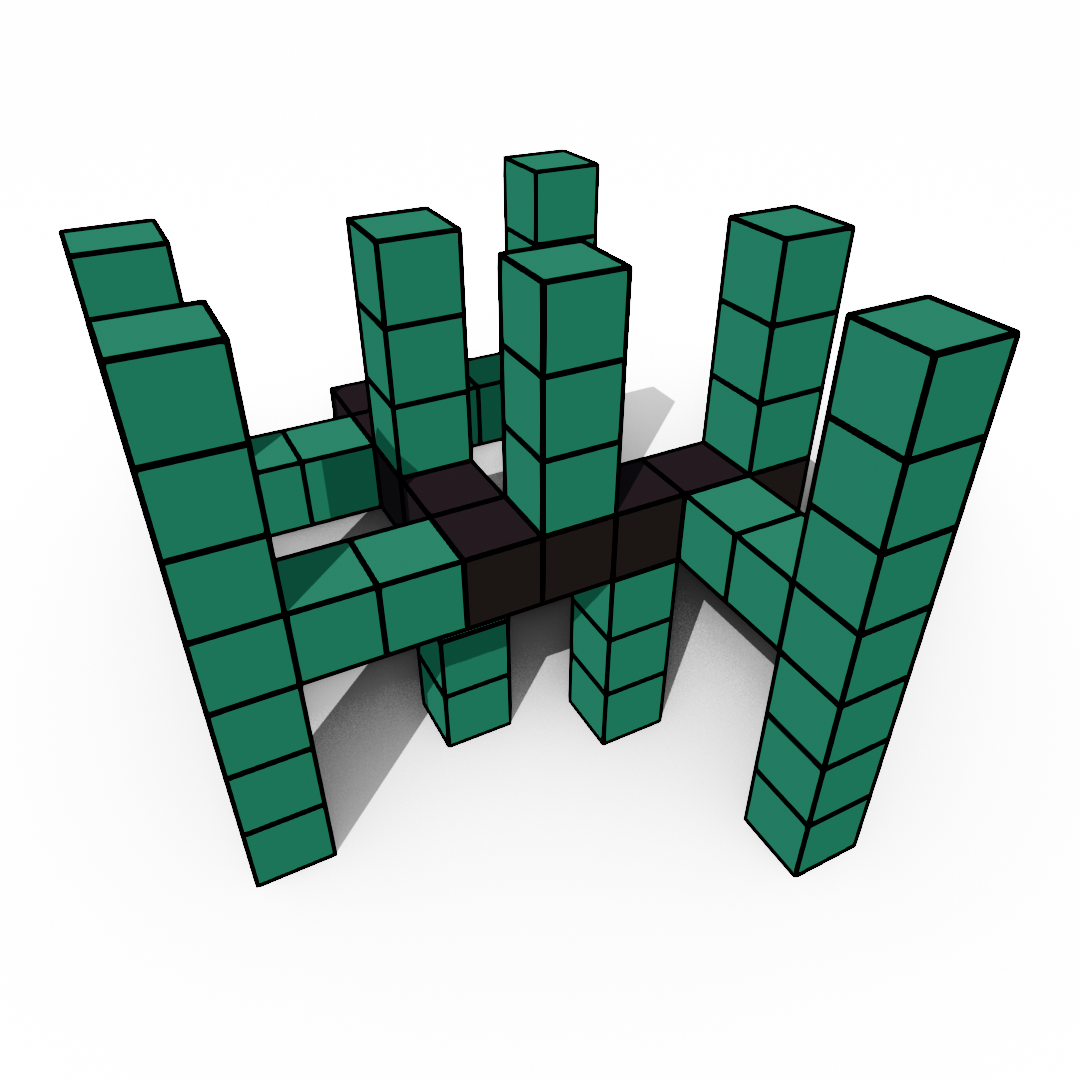}%
		\subcaption{A nicer view.}
	\end{subfigure}%
	\par%
	\vspace*{0.75cm}%
	\begin{subfigure}[b]{\columnwidth}%
		\hfil%
		\includegraphics[page=1]{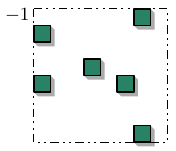}\hspace{2em}%
		\includegraphics[page=2]{figures/sc-reduction/turn-gadget-diagram}\hspace{2em}%
		\includegraphics[page=3]{figures/sc-reduction/turn-gadget-diagram}%
		\subcaption{Three layers of the turn gadget.}
	\end{subfigure}
		\caption{Turn gadgets connect path gadgets orthogonally}
		\label{fig:turn-gadget}
\end{figure}

Additional spike attachment methods are shown for crossings, tee-crossings, selectors, layer gadgets, and cover gadgets, respectively, with crossings specifically involving a horizontal and vertical path separated by a height difference of at least $3$ (see~\cref{fig:crossing-gadget,fig:tee-gadget,fig:selector-gadget,fig:layer-gadget,fig:path-gadget,fig:turn-gadget,fig:cover-gadget}).

\begin{figure}[htb]%
	\captionsetup[figure]{justification=centering}%
	\begin{subfigure}[b]{\columnwidth/2 -0.5em}%
		\centering%
		\includegraphics[width=\columnwidth*2/3]{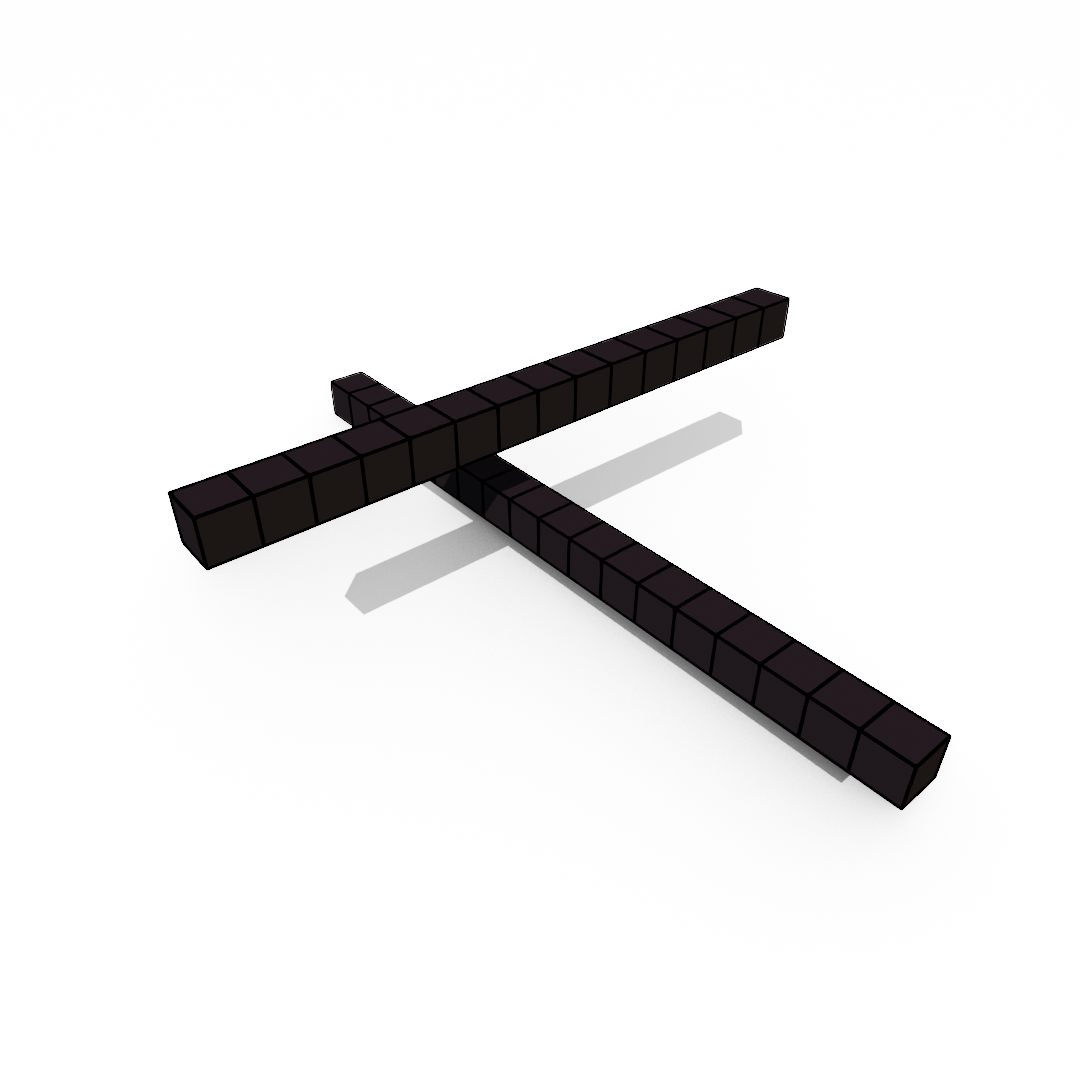}%
		\subcaption{Paths crossing at a vertical distance of 3...}
	\end{subfigure}
	\hfill%
	\begin{subfigure}[b]{\columnwidth/2 -0.5em}%
		\centering%
		\includegraphics[width=\columnwidth*2/3]{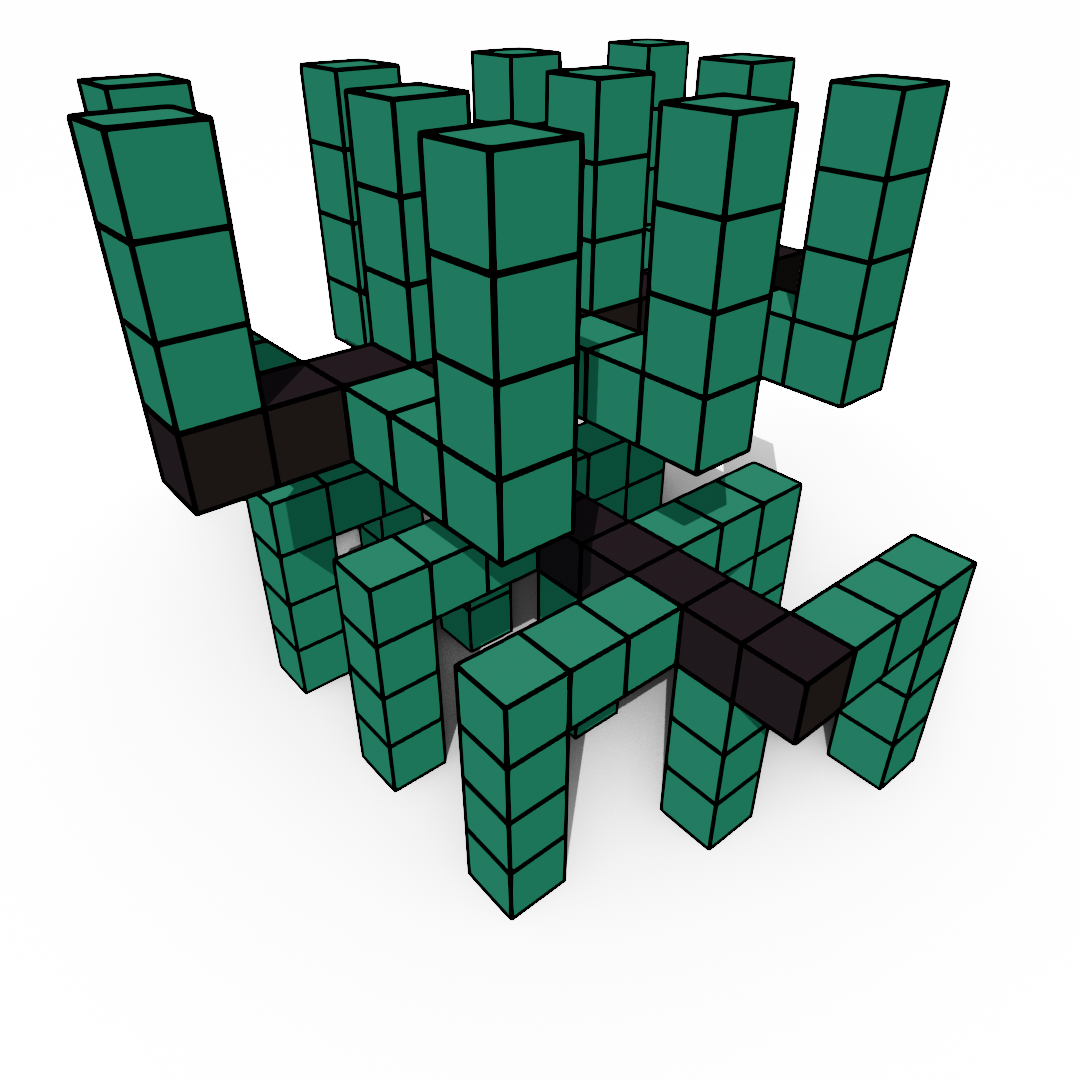}%
		\subcaption{...with lots of spikes attached.}
	\end{subfigure}%
	\par%
	\vspace*{0.75cm}%
	\begin{subfigure}[b]{\columnwidth}%
		\hfil\includegraphics[page=1]{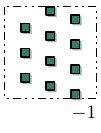}
		\hfil\includegraphics[page=4]{figures/sc-reduction/crossing-gadget-diagram}%
		\hfil\includegraphics[page=2]{figures/sc-reduction/crossing-gadget-diagram}
		\hfil\includegraphics[page=5]{figures/sc-reduction/crossing-gadget-diagram}%
		\hfil\includegraphics[page=3]{figures/sc-reduction/crossing-gadget-diagram}
		\hfil\includegraphics[page=6]{figures/sc-reduction/crossing-gadget-diagram}%
		\subcaption{Six layers of two crossing path gadgets.}
	\end{subfigure}
	\caption{For crossing gadgets, paths temporarily forego full vertical spikes, extending them either up or down.
		This weakens the spike constraint: Crossings have to be far away from free~modules.}
	\label{fig:crossing-gadget}
\end{figure}

In figures, all spikes are indicated by straight lines of length $3$ pointing upward and downward.
Once they extend beyond the instance illustrated in~\cref{fig:sc-reduction-diagram}, the spikes spread out to occupy additional space.
In particular, we require that the endpoint of each spike is at least $100 \cdot k \cdot A$ away from any module that does not belong to the same spike.
We~note that the instance in~\cref{fig:sc-reduction-diagram} can be described by a polynomial $p$ in the size of the original set cover problem, determined by the choices of $A$, $B$, and $C$.
There are at most $p$ immobile modules, each fixed by its own spike.
To ensure that the endpoint of each spike remains at the required distance from all modules of other spikes, one side of the large cube, as shown in~\cref{fig:projection}, must have an area of $(100 \cdot k \cdot A)^2 \cdot p$.
This area requirement remains polynomial in the size of the original set cover instance.

\begin{lemma}
	Without loss of generality, it is sufficient to consider valid sliding sequences for \parallelcubes that are also valid sliding sequences in the \textsc{Immobile} variant.
\end{lemma}

\begin{proof}
	We show that for any sliding sequence that is shorter than the one obtained from~\cref{lem:trivialbound}, there exists another sliding sequence of length at most the same in which none of the modules designated as \emph{immobile} are moved.
	At the beginning of a sliding sequence, the only modules that can move freely are $\module m$ and those at spike endpoints.
	
	The modules at spike endpoints have a distance of $100\cdot k \cdot A$ to any module that is not part of the spike.
	Thus, in a shortest sliding sequence, they cannot be moved into any of the gadgets of $C_I$.
	Moreover, they cannot form cycles with modules beyond those in their own spike or free any modules near the gadgets.
	Thus, any non-trivial sliding sequence that moves end-of-spike modules can be replaced by an equivalent one that performs the same slides without moving those end-of-spike modules.
	The new sliding sequence still transforms $C_I$ into $C_F$ and does not exceed the length of the original.
	
	Thus, to free a module considered immobile in the \textsc{Immobile} variant, a mobile module must serve as a backbone for its removal.
	We claim that this is impossible.
	To see why, note that all immobile modules $\module M$ and spike modules fall into one of the following three categories:
	
	\begin{enumerate}[(i)]
		\item Removing $\module M$ from $C_I$ splits $C_I$ into three connected components, and any additional module can connect at most two of these components simultaneously.
		\item Removing $\module M$ from $C_I$ splits it into two connected components, and the only way to reconnect them is by placing a module at $\module M$’s original position.
		\item Removing $\module M$ from $C_I$ splits it into two connected components, and $\module M$ is part of a spike attached to a crossing gadget. Since all crossing gadgets are at least a distance of $100 \cdot k \cdot A$ away from any of the modules $\module m$ and $\module m_j$ for any $j=1,\ldots, n$, no sliding sequence shorter than the trivial one can move a mobile module far enough to serve the backbone constraint to free a module.
	\end{enumerate}
	
	In all cases, we conclude that for any shortest sliding sequence, there exists a sequence of the same length in which no immobile or spike modules are moved.
\end{proof}

\begin{figure}
	\centering
	\includegraphics[page=2]{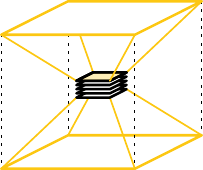}
	\caption{Schematic idea how to project spikes to the boundary of a large cube.}
	\label{fig:projection}
\end{figure}

In order to finally derive \logAPX-hardness of \parallelcubes, we perform a so-called \AP-reduction.

\begin{definition}[{\cite[pp. 257--258]{ausiello2012complexity}}]
	\label{def:ap}
	Let $P_1$ and $P_2$ be two \NP\ optimization problems. $P_1$ is \AP-reducible to $P_2$ if two functions $f$ and $g$ and a constant $\alpha\geq1$ exist such that
	\begin{itemize}
		\item for any instance $X$ of $P_1$ and any rational $r>1$, $f(X,r)$ is an instance of $P_2$.
		\item for any instance $X$ of $P_1$ and any rational $r>1$, if there is a feasible solution of $X$, then there is a feasible solution of $f(X,r)$.
		\item for any instance $X$ of $P_1$ and any rational $r>1$, and for any $Y$ that is a feasible solution of $f(X,r)$, $g(X,Y,r)$ is a feasible solution of $X$.
		\item $f$ and $g$ are polynomial-time computable.
		\item for any $X$ of $P_1$ and any rational $r>1$, and any feasible solution $Y$ for $f(X,r)$, a performance ratio of at most $r$ for $Y$ implies a performance ratio of at most $1+\alpha(r-1)$ for $g(X,Y,r)$.
	\end{itemize}
\end{definition}

In our case, $f$ corresponds to the reduction from \textsc{Set Cover} to \parallelcubes, and $g$ is the operation, in which we obtain a set cover from a sequence of slides as described in the proof of~\cref{prop:nphard}.
We require $f$ and $g$ to be polynomial-time computable.
For a regular \AP-reduction, a weaker requirement suffices: it is enough that $f$ and $g$ are polynomial-time computable for any fixed value of $r$.
This is a standard technique for proving that a problem is \APX-hard.
In contrast, to establish \logAPX-hardness, a stronger condition is necessary: the runtimes of $f$ and $g$ must be independent of $r$.
We further choose $\alpha = 4$.

\theoremLogAPX*

\begin{proof}
	Suppose we can approximate the length of the optimal sliding sequence within a factor of at most $r$.
	Then we can compute a sliding sequence of length at most $d<r\cdot d^\ast$ in polynomial time.
	By~\cref{rem:size}, this allows us to obtain a set cover of size
	\begin{align*}
		\ell &= \frac{d-R}{2 \cdot A} < \frac{r\cdot d^\ast - R}{2\cdot A} = \frac{r\cdot d^\ast - r\cdot R^\ast + r\cdot R^\ast - R}{2\cdot A}\\
		&=  \frac{r\cdot d^\ast - r\cdot R^\ast}{2\cdot A} +  \frac{r\cdot R^\ast - R}{2\cdot A} < r\cdot\ell^\ast + \frac{r\cdot R^\ast}{2\cdot A} < r\cdot \ell^\ast + \frac{r}{2}
	\end{align*}
	We now have an estimate for $\ell$ based on $\ell^\ast$ and $r$.
	The next step is to rewrite it in the form specified by~\cref{def:ap}, that is, $\ell \leq \ell^\ast( 1 + \alpha(r-1))$.
	
	To this end, we distinguish two cases.
	First, assume $(r-1)\geq \frac{1}{2\ell^\ast + 1}$.
	We make the following estimate
	
	\begin{align*}
		r\cdot \ell^\ast + \frac{r}{2} &= \ell^\ast + \frac{1}{2} + (r-1)\ell^\ast + \frac{(r-1)}{2} + (3\cdot\ell^\ast - 3\cdot\ell^\ast)\\
		&= \ell^\ast + 4(r-1)\ell^\ast + \frac{1}{2} - (r-1)(3\cdot\ell^\ast-\frac{1}{2}) \leq \ell^\ast ( 1 + \alpha(r-1))
	\end{align*}
	
	For the second case, assume that $(r-1)<\frac{1}{2\ell^\ast + 1}$.
	We estimate in the following way.
	
	\begin{align*}
		r\cdot \ell^\ast + \frac{r}{2} = \ell^\ast + \frac{1}{2} + (r-1)\ell^\ast + \frac{(r-1)}{2} < \ell^\ast + (r-1)(\ell^\ast + \frac{1}{2}) + \frac{1}{2} < \ell^\ast + 1
	\end{align*}
	
	In the second case, the estimate allows us to recover the optimal solution, while in the first case, the inequality holds as expected and required by the \AP-reduction.
	This concludes our proof.
\end{proof}

%% file: A04-teleport.tex
\section{Omitted details and missing proofs from~\cref{sec:teleport}}
\label{app:sec:teleport}

\lemmafreeexists*
\begin{proof}
	Let $\mathcal{F}(C)$ denote the set of module faces comprising the outer extended face of configuration $C$. We proceed by induction on $|\mathcal{F}(C)|$.
	
	\descriptionlabel{Base case:} 
	 Consider a configuration $C$ where $|\mathcal{F}(C)| \leq 3$. This is only possible if $C = \mathcal{F}(C)$. If $|\mathcal{F}(C)| = 1$ then $C$ is one module and it is trivially free. Otherwise, $C$ has two free modules, as every connected configuration does (Lemma 6 \cite{rus.vona2001crystalline-robots}).
	 Since $C = \mathcal{F}(C)$ there are two free modules in $\mathcal{F}(C)$. 
	
	Now assume that for any $C$ with an outer extended face of size $ |\mathcal{F}(C)| < k$, there exists two free modules in $\mathcal{F}(C)$ (unless $|C| = 1$).
	Consider a configuration $C$ with $|\mathcal{F}(C)| = k$. We identify two extremal modules, the highest, leftmost, topmost module $\module m_{max}$ and the lowest, rightmost, bottom most module $\module{m}_{min}$. Due to their extremal nature, both modules belong to the extended outer face.
	
	We then check the following:
	\begin{itemize}
		\item If either $\module{m}_{min}$ or $\module{m}_{max}$ is free, the lemma holds.
		\item If both are not free, consider $C \setminus \module{m}_{max}$, which has at most three components, as $\module m_{max}$ could only have three neighbors $\module a , \module b , \module c$.
	\end{itemize}
	
	We select one of these components, without loss of generality, say the one containing $\module a$, $K_a$.
	If $\module m_{max}$ was a cut module, then at least one of $\module b$ or $\module c$ are not elements of $K_a$. 
	Again, without loss of generality, say $\module b \notin K_a$, or, $K_b \cap K_a = \emptyset$. 
	We claim that $\mathcal{F}(K_a) < \mathcal{F}(C)$. 
	Observe, $\module a$ and $\module b$ now form a pinched edge in $C \setminus {\module m_{max} }$. 
	$\module a$ and $\module b$ have two cells in their shared neighborhood, one of which was occupied by $\module m_{max}$, the other we call $z$. The latter, $z$, must be empty, or else $\module a$ and $\module b$ would be connected in $C \setminus {m_{max}}$.
	If $\mathcal{F}(K_a)$ was in fact larger than $\mathcal{F}(C)$, then the face of $a$ that borders $z$, call it $f_z$, must be an element of $\mathcal{F}(K_a)$ but not of $\mathcal{F}(C)$.
	 
	For a contradiction suppose that this is the case ($f_z \in \mathcal{F}(K_a)$ but $f_z \notin \mathcal{F}(C)$).
	Clearly, $\module b$ and $\module c$ are adjacent to $f_z$ as well.
	Consider the shared edges between $\module b, \module c$ and $\module a , \module c$.
	If~$f_z \notin \mathcal{F}(C)$ then neither of these edges cannot be pinch edges in $C$. 
	However, this is only possible if $\module a$ and $\module c$ have a shared neighbor as well as $\module b$ and $\module c$.
	But if these shared neighbors existed, then $\module a$, $\module b$, and $\module c$ would all be connected and $K_a \cup K_b \neq \emptyset$, a contradiction.
	Hence, $|\mathcal{F}(K_a)| < |\mathcal{F}(C)|$ and our inductive hypothesis applies.  
	$\mathcal{F}(K_a)$ has at least two free modules. 
	Therefore at one of them is not $\module a$, and so that module is free in $C$. 
	We can also induct on~$K_b$ to get another free module, satisfying our inducting hypothesis.
\end{proof}

\lemmaTeleport*

\begin{proof}
	We use induction on $|S|^2+\delta_F(\module m, e)$ where $\delta_F(\module m, e)$ denotes the distance between $\module m$ and $e$ along the extended face $F$ of  $S\setminus\{\module m\}$ containing both. That is, in the adjacency graph of module faces induced by $F$, $\delta(\module m, e)$ is the length of the shortest path between the face of~$\module m$ in $F$ and the face of $e$ in $F$.
	Note that $|S|$ is an upper bound for $\delta_F(\module m, e)$.
	The base case is when $\module m$ is slide-adjacent to $e$ and $\module m$ can move to $e$, in which case we are done.
	Note that the case when $\module m$ is adjacent to $e$ and $\module m$ can perform a convex transition through $e$ cannot happen since that would make  $(C\setminus\{\module m\})\cup\{e\}$ disconnected.
	
	We now proceed with a case analysis for the inductive step.
	Note that if there is a sequence of moves bringing $\module m$ to $e$ on the surface of $S$ the claim is trivial. Thus, below we assume no such schedules exist, i.e., all such paths are blocked.
	Refer to \Cref{fig:teleport}.
	
	\begin{figure}
		\centering
		\includegraphics[page=2]{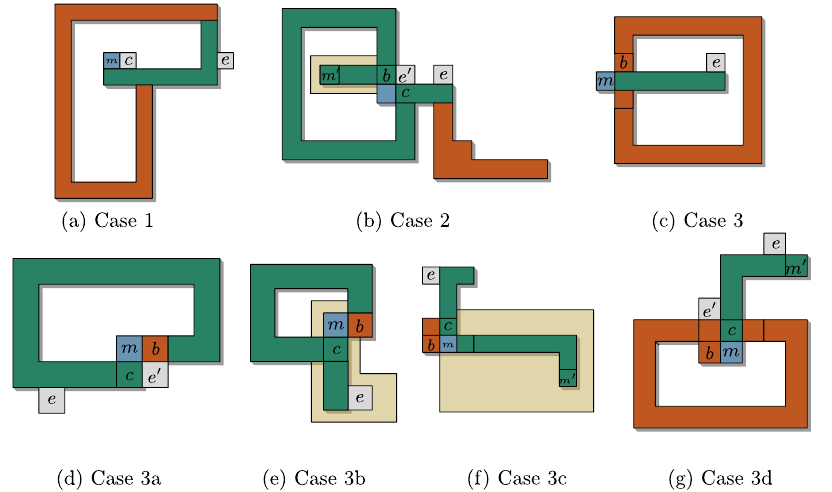}
		\caption{Cases for the proof of~\cref{lem:teleport}.}
		\label{fig:teleport}
	\end{figure}
	
	\begin{description}
		\item[Case 1.]  $\module m$ is slide-adjacent to an empty cell $c$ with $\delta_F(c, e)<\delta_F(\module m, e)$. Then we move $\module m$ to $c$ and the inductive hypothesis holds.
		
		\item [Case 2.] The path from $\module m$ to $e$ is blocked by a module $\module b\in S$.
		Then, $\module b$ has a pinched edge in $S\setminus\{\module m\}$  with another module  $\module c\in S$ and $\module m$ is adjacent to both.
		Let $e'$ be the other cell adjacent to $\module b$ and $\module c$. Note that  $e'\notin C$ or else $\module m$ to $e$ would not be in the same extended face.
		If $\module b$ is free, slide it to $e'$ and slide $\module m$ to $\module b$'s original position. Then the inductive hypothesis holds by swapping the labels $\module m$ and  $\module b$.
		Else, let $S'$ be a cut component of $\module b$ that does not contain $\module c$ or $\module m$.
		We can find a free module $\module m'\neq \module b$  in $S'$ by \Cref{lem:free-existance}.
		We can then use inductive hypothesis in $S'$ to teleport $\module m'$ to $e'$ and then to teleport $\module m$ to the original position of $\module m'$, effectively teleporting $\module m$ to $e'$.
		Then, the inductive hypothesis holds since $\delta_F(e', e)<\delta_F(\module m, e)$.
		
		\item [Case 3.] The path from $\module m$ to $e$ is blocked by a module $\module b\in C$.
		Let $\module c\in S$ be the module of $S$ that $\module m$ is adjacent to when blocked by $\module b$.
		Either $\module m$ is already adjacent to $\module b$ or it can slide to be adjacent to $\module b$.
		Thus we can assume that $(C\setminus S)\cup\{\module m\}$ is connected.
		We split into subcases.
		
		\begin{description}
			\item [Case 3a.] $\module c$ is free. Then, if $\module c$ and $\module b$ form a pinched edge in $C\setminus\{\module m\}$, let $e'$ be the empty cell adjacent to both $\module c$ and $\module b$. Note that $\delta_F(e', e)<\delta_F(\module m, e)$ by the definition of $\module b$.
			We can then move $\module c$ to $e'$ and $\module m$ to $\module c$'s original position, effectively teleporting $\module m$ to $e'$.
			Then, the inductive hypothesis holds.
			If $\module c$ and $\module b$ do not form a pinched edge, then an even simpler procedure works.
			We can use the inductive hypothesis to move $\module c$ to $e$, and afterwards move $\module m$ to $\module c$'s original position.
			
			\item [Case 3b.] $\module c$ is a cut module in $S\setminus \{\module m\}$ and a component $S'$ of $S\setminus \{\module c, \module m\}$ that is not adjacent to $e$ is adjacent to $C\setminus S$.
			Then $C \cup S'$ is a connected configuration. We can then apply the inductive hypothesis to $S \setminus S'$ teleporting $\module m$ to $e$ on an extended face of $S\setminus S'$.
			
			\item [Case 3c.] $\module c$ is a cut module in $S \setminus \{\module m\}$ and all components $S'$ of $S\setminus \{\module c, \module m\}$ that are not adjacent to $e$ are also not adjacent to $C\setminus S$.
			Then, $\module m$ is either adjacent to a $S'$ or can move to be adjacent to a $S'$.
			We apply the inductive hypothesis to teleport a free module $\module m'\in S'$ (which exists by \Cref{lem:free-existance}) to $e$ in $S\setminus\{m\}$.
			We then use it again to teleport $\module m$ to the original position of $\module m'$ on the extended face of $S'\cup\{\module m\}$.
			
			\item [Case 3d.] $\module c$ is a not a cut module in $S\setminus \{\module m\}$, which means that its deletion separates two components: $C'=(C\setminus S)\cup \{m\}$ and  $S'=S\setminus \{\module c, \module m\}$.
			Then, there is an empty position $e'$ adjacent to both $C'$ and  $S'$, or else a direct path from $\module m$ to $e$ would exist.
			This cell must be closer to $e$ because it is adjacent to $S'$.
			If $|S|\ge3$, we can apply the inductive hypothesis to teleport a free module $\module m'\in S'$ to $e'$ on $S\setminus\{\module m\}$.
			We can apply the inductive hypothesis again to teleport $\module m$ to the original position of $\module m'$ on the extended face of $S\setminus\{\module m'\}$.
			That effectively teleported $\module m$ to $e'$ and since $\delta_F(e', e)<\delta_F(\module m, e)$, the inductive fills $e$.
			Finally, if $|S|=3$, moving $\module m'$ to $e'$ disconnects $S$. We can obtain the same result as above by moving $\module m'$ to $e'$,  $\module c$ to $\module m'$'s original position, and $\module m$ to $\module c$'s original position.
		\end{description}
	\end{description}
	
	We can directly transform the induction above into a recursive algorithm.
	To do so, let $T(|S|,\delta_F(\module m, e))$ be the runtime of the recursive algorithm.
	The worst case recurrence occurs in Case 3d, when:
	\[
	T(|S|,\delta_F(\module m, e)) \le 2\cdot T(|S|-1,|S|-1) + T(|S|,\delta_F(\module m, e)-1).
	\]
	
	However, in that case the $T(|S|,\delta_F(\module m, e)-1)$ term is in reality linear on $|S|$ since we have the guarantee of a direct path from $\module m'$ to $e'$.
	Thus the number of moves performed is~$2^{\mathcal{O}(|S|)}$.
\end{proof}

%% file: A05-algorithm.tex
\section{Omitted details and missing proofs from~\cref{sec:alg}}
\label{app:sec:alg}

\subsection{Finding ${T}_\module{s}$}
\label{app:sec:alg:Ts}

\begin{lemma}
	\label{lem:subtree}
	Given a constant $t\in [n]$, ${T}$ contains a subtree ${T}_\module{p}$ rooted at a module $\module{p}$ so that $t\le |{T}_\module{p}|\le 5 t$.
\end{lemma}
\begin{proof}
	The claim follows from the fact that the maximum degree in ${T}$ is 6.
	If $n/5\le t\le n$, $T_\module r$ satisfies the claim where $\module r$ is the root of $T$.
	Else, $\module r$ has a child $\module p$ with $|T_\module p|> t$ since $\module r$ has at most 5 children.
	We can then recurse on $T_\module p$.
\end{proof}

We use \Cref{lem:subtree} to obtain a tree ${T}_\module{s}$ of size $\Theta(A+h)$.

\subsection{Omitted proofs}

\lemmaGrowball*

\begin{proof}
	Let $(\module m, \module m_1, \module m_2, \module m_3, \module m_4)$ be a path in $T_\module{m}$ (which exists due to the size of $T_\module m$).
	We first claim that we can build a ball of radius 2 centered at either $\module{m}_1$, $\module{m}_2$, or $\module{m}_3$.
	Note that by construction, $T$ has the following property:
	
	\begin{enumerate}[\hspace{2em}]
		\item[(\textsf{M})] \label{prop:M} Let $T_\module m$ be a subtree of $T$ rooted at $\module m$. If $\module p \in T_\module m$ is adjacent to a module $\module q$ in the same meta-cell as $\module p$, then either $\module q$ is the parent of $\module p = \module m$, or  $\module q \in T_\module m$.
	\end{enumerate}
	
	\begin{claim}
		\label{cl:ball-2}
		Within $\mathcal{O}(1)$ makespan, there is a schedule fills an $L_1$ ball of radius 2 centered at either $\module{m}_1$, $\module{m}_2$ or $\module{m}_3$.
	\end{claim}
	\begin{claimproof}
		Let $\mathcal{B}_1$ and $\mathcal{B}_2$ denote the set of cells corresponding to the
		$L_1$ ball of radius 2 centered at $\module{m}_1$ or $\module{m}_2$, respectively.
		Let $S_\module m$ be the set of modules originally in $T_\module{m}$.
		We use the following case analysis.
		
		\begin{description}
			\item[Case 1.] Assume that $\module{m}_1$ is not on the boundary of its meta-cell.
			Then, every cell in $N(\module{m}_1)$ must be either empty or in $T_\module{m}$, by the definition of $T$.
			If they are empty, we can apply \Cref{lem:teleport} with $S = S_\module{m}$ and a free module outside of the component of $T_\module{m}\cap \mathcal{B}_1$ containing $\module m$, which must exists due to the size of $T_\module{m}$.
			Thus, we can assume that $N(\module{m}_1)$ is full with modules in $S_\module m$.
			Similarly, if there are cells in $\mathcal{B}_1$ that remain empty, we can apply \Cref{lem:teleport} to fill them since every such empty cell is adjacent to $S_\module m$.
			
			\item[Case 2.] If all neighbors of $\module{m}_1$ that are outside of its meta-cell are either empty or in $S_\module m$, the same arguments of the first case apply, and we can apply \Cref{lem:teleport} until $\mathcal{B}_1$ is full.
			
			\item[Case 3.] Assume, without loss of generality, that $\module{m}_1$ is on the bottom face of its meta-cell, and that its bottom neighbor $\module m'$ is its only neighbor which is a module not in $S_\module m$.
			Then, similarly to case 1, by applying \Cref{lem:teleport} we can fill every position of $\mathcal{B}_1$, except for the cell $e$ below $\module m'$.
			If $e$ is already full we are done.
			Else, note that $\module m'$ is free: each of its original neighbors not in $S_\module m$ are now adjacent to a module in $S_\module m$, and $S_\module m$ is connected to the remaining configuration through $\module m$ which cannot be a neighbor of $\module m'$.
			Also note that  $\{\module m_1,\module m'\}$ is free since $\mathcal{B}_1$ is full except for $e$.
			Thus, we can apply \Cref{lem:teleport} with $S = \{\module m_1,\module m'\}$ to teleport $\module m_1$ to $e$.
			This leaves the original cell of $\module m_1$ empty, and $S_\module m \setminus\{\module m_1\}$ still connected.
			Since this cell is adjacent to $S_\module m \setminus\{\module m_1\}$, we can fill it with \Cref{lem:teleport}.
			
			\item[Case 4.] Assume, without loss of generality, that $\module{m}_1$ is on the bottom-front edge of its meta-cell, and that its bottom ($(-z)$-direction) and front ($(-y)$-direction) neighbors $\module m_b$ and $\module m_f$ are its only neighbors not in $S_\module m$.
			Note that cases 1--3 also apply for $\module{m}_2$. Therefore, we can also assume $\module{m}_2$ is on bottom-front edge of its meta-cell, else we apply one of the previous cases to $\module{m}_2$.
			We also assume that $\module{m}_2$ has exactly two neighbors not in $S_\module m$. (If it has only 1, case 3 applies.)
			Using the same argument as case 1, we can fill all positions in $\mathcal{B}_1$ except for the cell $c_f$ in front of $\module m_f$ ($\module m_1 + (0,-2,0)$), the cell  $c_d$ to the bottom of $\module m_f$ ($\module m_1 + (0,-1,-1)$), and the cell  $c_b$ to the bottom of $\module m_b$ ($\module m_1 + (0,0,-2)$).
			We break into sub-cases.
			
			\begin{description}
				\item[Case 4a.] If $c_d$ is full, and $c_b$ is empty, the same argument as Case 3 applies, and we can teleport $\module m_1$ to $c_b$ by \Cref{lem:teleport}.
				Symmetrically, the same holds for when $c_d$ is full, and $c_f$ is empty.
				\item[Case 4b.] Now, assume $c_d$ is empty.
				If $\module m_b$ is free, the set  $\{\module m_1,\module m_b\}$ is free, and we can teleport $\module m_1$ to $c_d$ and fill the original cell of $\module m_1$ similar to Case 3. The same holds for  $\module m_f$.
				If $\module m_b$ is not free then $c_b$ must be full because every other neighbor of $\module m_b$ is now adjacent to $S_\module m$.
				Then, deleting $\module m_b$ disconnects the module at $c_b$ with the other neighbors of $\module m_b$.
				Thus, the cells  $c_b+(1,0,0)$ and $c_b+(-1,0,0)$ must also be empty.
				This implies that the module $\module m_{2b}$ below $\module m_2$ is free (one of $c_b+(1,0,0)$ or  $c_b+(-1,0,0)$ is below $\module m_{2b}$ and all its other neighbors are adjacent to other full modules in $\mathcal{B}_1$).
				If $\module m_{2b}$ can preform a convex transition to $c_d$, we can move $\module m_2$ down and fill the empty position left by $\module m_2$ using \Cref{lem:teleport}.
				Else, the common neighbor of $\module m_{2b}$ and $c_d$ is full.
				In that case, we can apply the first sub-case of case 4 to $\module m_2$ since the corresponding $c_d = \module m_2+ (0,-1,-1)$ is full.
			\end{description}
			\item[Case 5.] Assume $\module{m}_1$ has 3 neighbors not in $S_\module m$.
			Without loss of generality, $\module{m}_1$ is the bottom front corner of its meta-cell.
			Then, $\module{m}_2$ and $\module{m}_3$ are both on an edge of the same meta-cell and thus have at most two neighbors not in $S_\module m$.
			We can then repeat case 4 with $\module{m}_1$ and $\module{m}_2$ replaced by $\module{m}_2$ and $\module{m}_3$.\claimqedhere
		\end{description}
	\end{claimproof}
	
	\begin{figure}
		\hfil%
		\begin{subfigure}{0.4\columnwidth}
			\includegraphics[width=\columnwidth,page=1]{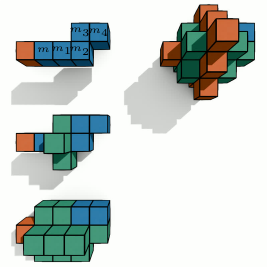}%
			\subcaption{Case 1.}
			\label{fig:cl-ball-2-case-1}
		\end{subfigure}%
		\hfil%
		\begin{subfigure}{0.4\columnwidth}
			\includegraphics[width=\columnwidth,page=2]{lem14/annotated}%
			\subcaption{Case 3.}
			\label{fig:cl-ball-2-case-3}
		\end{subfigure}%
		\par\hfil%
		\begin{subfigure}{0.4\columnwidth}
			\includegraphics[width=\columnwidth,page=3]{lem14/annotated}%
			\subcaption{Case 4.}
			\label{fig:cl-ball-2-case-4}
		\end{subfigure}%
		\hfil%
		\begin{subfigure}{0.4\columnwidth}
			\includegraphics[width=\columnwidth,page=4]{lem14/annotated}%
			\subcaption{Case 5.}
			\label{fig:cl-ball-2-case-5}
		\end{subfigure}%
		\caption{Proof of \Cref{cl:ball-2}. Modules $\module m_1$, $\module m_2$, $\module m_3$ and  $\module m_4$ are shown in blue, modules that are guaranteed to be in $S_\module m$ are in green and modules that are not guaranteed to be in $S_\module m$ are red.}
		\label{fig:ball-2}
	\end{figure}
	
	Once we have a ball of radius 2 we can grow it using the following claim.
	Let $\mathcal{B}$ be the maximal $L_1$ ball of full cells centered at a module $\module m^*$ produced by \Cref{cl:ball-2}.
	Let $S'$ be the set of modules of $T_{\module m}\setminus \mathcal{B}$ that were left in place after the construction of $\mathcal{B}$.
	
	\begin{claim}
		\label{cl:grow-ball}
		Let $\mathcal{B}$ be the maximal $L_1$ ball of full cells centered at a module $\module m$.
		If the $L_1$-radius $\text{radius}(\mathcal{B})$ is greater than or equal to $2$, there is a schedule of makespan $\mathcal{O}(\text{radius}(\mathcal{B}))$ that teleports $\module m$ to an empty position adjacent to $\mathcal{B}$.
	\end{claim}
	\begin{claimproof}
		We prove the claim by induction on the radius of $\mathcal{B}$. For the inductive step, assume  $r = \text{radius}(\mathcal{B})> 2$.
		
		\descriptionlabel{Inductive step.}
		Assume the claim holds for any ball of radius $r-1$. We call a cell of $\mathcal{B}$ \emph{extremal} if it is one of the 6 cells that are maximal or minimal in one of the axis directions.
		We wish to fill an empty position $e$ adjacent to the boundary of $\mathcal{B}$. We distinguish two cases:
		
		\begin{description}
			\item [Case 1.] There exists an empty position $e$ adjacent to a module $\module{u} \in \mathcal{B}$ that is not extremal.
			Since $r > 2$ and $\module{u}$ is not extremal, $\module{u}$ has sufficient connectivity within $\mathcal{B}$: it is not a cut-vertex since all its neighbors not in $\mathcal{B}$ are also adjacent to other modules in $\mathcal{B}$.
			We move $\module{u}$ to $e$. This creates a vacancy at the original position of $\module{u}$, denoted $p_{vac}$, which lies on the boundary of the ball of radius $r-1$.
			By the inductive hypothesis, we can teleport the center module $\module{m}^*$ to fill $p_{vac}$.
			
			\item [Case 2.] All empty cells adjacent to $\mathcal{B}$ are only adjacent to extremal module of $\mathcal{B}$.
			Without loss of generality, assume $e$ is an empty cell above the topmost module $\module{u}$ of $\mathcal{B}$.
			By the definition of the case $N^*(\module{u})\setminus\{e\}$ is full, ensuring that $\module{u}$ is free: all its neighbors not in $\mathcal{B}$ are also adjacent to other modules in $\mathcal{B}$.
			We move $\module{u}$ up towards $e$. Depending on the local support, this occurs in two ways:
			\begin{description}
				\item [Case 2a (Direct Slide).] If $e$ is adjacent to another occupied cell, $\module{u}$ slides directly into $e$. We then apply the inductive hypothesis to fill the vacancy at $\module{u}$'s original position.
				
				\item [Case 2b (Convex Transition).] If $e$ has no other occupied neighbors to support a slide, $\module{u}$ performs a convex transition to a temporary position $e'$ adjacent to $e$. This vacates $\module{u}$'s original position, $p_{vac}$.
				We then apply the inductive hypothesis to teleport $\module{m}^*$ into $p_{vac}$.
				Once $p_{vac}$ is refilled, it provides the necessary support for $\module{u}$ to perform a slide from $e'$ into its final destination $e$.
			\end{description}
		\end{description}

		\descriptionlabel{Base case.} The base case happens when $\mathcal{B}$ has radius 2.
		\begin{description}
			\item [Case 1.] Assume the only empty positions adjacent to $\mathcal{B}$ are extreme in one of the cardinal directions.
			Without loss of generality, let $e = \module m^*+(0,0,3)$.
			Then, the set $\{\module m^*, \module m^*+(0,0,1), \module m^*+(0,0,2)\}$ is free and they can all move up in one parallel move.
			\item [Case 2.] Assume that  $e = \module m^*+(1,0,2)$ is empty.
			Move the
			\item module below $e$ ($\module m^*+(1,0,1)$) up to $e$.
			If module $\module m^*+(0,0,1)$ is free, move it right to $\module m^*+(1,0,1)$ and then move $\module m^*$ up to $\module m^*+(0,0,1)$.
			Else, module $\module m^*+(0,0,1)$ is not free and, therefore, its deletion disconnects $\module m^*+(0,0,2)$ with the rest of $\mathcal{B}$.
			Thus, all horizontally neighboring cells (with the same $z$-coordinate) of $\module m^*+(0,0,2)$ are empty, else they would connect it with $\module m^*+(1,0,1) \in \mathcal{B}$.
			We move $\module m^*+(1,0,1)$ to $e$ and further break down into subcases.
			\begin{description}
				\item [Case 2a.] Assume that $\module m^*+(1,-1,2)$ is empty. Then, we can convex transition the module in $e$ to $\module m^*+(0,-1,2)$. This bridges $\module m^*+(0,0,2)$ with $\module m^*+(0,-1,1) \in \mathcal{B}$ freeing $\module m^*+(0,0,1)$. We can proceed by moving $\module m^*+(0,0,1)$ right and $\module m^*$ up as before. We are done since we filed an empty position adjacent to $\mathcal{B}$.
				
				\item [Case 2b.] Assume that both $\module m^*+(1,-1,2)$ and $\module m^*+(1,-1,1)$ are full.
				Then, $\module m^*+(0,0,1)$ is free (se the path $(\module m^*+(0,0,2), e, \module m^*+(1,-1,2), \module m^*+(1,-1,1), \module m^*+(0,-1,1))$) and we can proceed as the standard Case 2 moving $\module m^*+(0,0,1)$ right and $\module m^*$ up.
				
				\item [Case 2c.] Assume that $\module m^*+(1,-1,2)$ is full and $\module m^*+(1,-1,1)$ is empty.
				Then, we move the module in $\module m^*+(1,-1,0)$ up to $\module m^*+(1,-1,1)$.
				Then we can apply Case 2b. After $\module m^* +(1,0,1)$ is full, $\module m^*+(1,-1,1)$ is not needed for connectivity, and we can move the module there down to its original position $\module m^*+(1,-1,0)$.
			\end{description}
			\item [Case 3.] Assume that $\module m^*+(1,-1,1)$ is empty and all horizontally neighboring cells of $\module m^*+(0,0,2)$ are full.
			Then, we can move module $\module m^*+(1,0,1)$ to $\module m^*+(1,-1,1)$.
			Since $\module m^*+(0,-1,2)$ is full, $\module m^*+(0,0,1)$ is free and we can proceed as Case 2 moving $\module m^*+(0,0,1)$ right and $\module m^*$ up.\claimqedhere
		\end{description}
	\end{claimproof}

	We use \Cref{cl:grow-ball} to fill an empty position adjacent to $\mathcal{B}$ and empty the center of $\mathcal{B}$. If~$S'$ has a free module adjacent to $\mathcal{B}$ we can then use \Cref{cl:grow-ball} in reverse to fill the center of $\mathcal{B}$.
	We delete such module from $S'$ since it is now part of the ball $\mathcal{B}$.
	If $S'$ has no free module adjacent to $\mathcal{B}$, there must be an empty position adjacent to both $S'$ and $\mathcal{B}$. This is because by definition of $T$ and $S'$, modules in $S'$ are only adjacent to other modules of $S'$ and module in $\mathcal{B}$, and \Cref{cl:grow-ball} maintains that $S'$ is adjacent to $\mathcal{B}$.
	We can then teleport a free module of $S'$ to an empty position adjacent to $\mathcal{B}$ by \Cref{lem:teleport}.
	By repeating this process, we can grow $\mathcal{B}$ until it contains an entire meta-cell.
\end{proof}

\subsection{Details of Scaffold and Compress}

After growing the snake to an appropriate size, we construct a ``scaffolding'', a temporary substructure that allow us to preform many snake operations in parallel to ``compress'' the configuration and eventually reconfigure into a compact configuration.
Once the snake reaches a size of $\mathcal{O}(A + h)$ we stop growing and start building a scaffolding, which will be a scaled version of $z[C]$. We then use the scaffolding to ``compress'' $C$ into a compact configuration.

We use push and pull operations to move the head of the snake to a position of height equal to the maximal module of $C$.
Assume that the head is aligned with a $5\times 5 \times 5$ meta-cell of $B_1$.
We then compute a spanning tree $T$ on the scaled cells of the $z[C]$ rooted at the meta-cell containing the snake's head.
Using \cref{alg:snake-dfs} we traverse $T$.
Note, that as the starting size of our snake is $\mathcal{O}(A + H)$, the initial meta-cells of the snake outnumber the metacells of $T$.
Hence, at the end of this traversal, our snake can occupy every meta cell of~$T$.
We fork at each meta cell of $T$ splitting into many snakes.

We now have our scaffolding a connected component of snakes, each of which have their head at the same height, we denote this collective of heads by  $\text{\ss}$.
We can move snake each downwards in parallel via Push operations.
Note for a single snake, at every stage of the push operation a large portion the skin cells are stationary.
Hence every snake in $\text{\ss}$ can preform a push in parallel and their skins will maintain connectivity.
As a snake encounters module $\module m \in C \setminus \text{\ss}$ the snake \newterm{consumes} $\module m$.

We terminate this process when $\text{\ss}$ reaches the bottom of $C$.
At this point $C$ is a set of vertical snakes, arranged in the scaled projection of $z[C_1]$.
We now rearrange these to form a compact configuration.
First, we ``compress'' each snake filling the empty spine cells by ``teleporting'' modules from the tail, to any empty spine positions.
This results in a configuration consisting of solid vertical towers, i.e. for every module, $\module m \in C$, any cell directly below $\module m$ is full.
The following will reconfigure $C$ so every position to the left of $\module m$ is full:
\begin{enumerate}
	\item Take every module of even $y$ coordinate.
	\begin{enumerate}
		\item Of this set, move every module with even $z$ coordinate left, unless it would move to a position of negative $x$ coordinate, or its left neighbor does not move. (A module with no left neighbor will have to preform a convex transition up or down, and then slide into the appropriate position)
		\item Now move the remaining modules with even $y$ coordinate and odd $z$ to the left, unless it would move to a position of negative $x$ coordinate, or its left neighbor does not move.
	\end{enumerate}
	\item Now move every module with odd $y$ coordinate left, following the same procedure, first those with even $z$ and then those with odd $z$.
	\item Repeat the above until $C$ is monotone with respect to the $y$-axis
\end{enumerate}

By alternating even and odd $z$ coordinates, we assure that connectivity is maintained.
When every module of even $y$ and $z$ coordinate are moving, every module with odd $y$ coordinate is connected through those with even $y$ and odd $z$. 
The same argument applies for the other three combinations of even and odd $y$/$z$ coordinates.
After this process, it is still true that for any module $\module m$ every cell directly below it is full.
Further this is true of every position directly left of $\module m$, as otherwise we would not have terminated.
This is transitive, so for every such $\module m$, every cell with lesser $z$ or $x$ coordinate is full. 

If we repeat this procedure, replacing any reference to ``$y$ coordinate'' with ``$x$ coordinate'', and replacing ``move left'' with ``move down''.
$C$ is now compact, it is still true that for every $\module m$ every cell of lesser $z$ or $x$ coordinate is full.
However, by a symmetric argument it is also true that every position of lesser $y$ is also full. 
This is the definition of a compact configuration.

%% file: A05_1-gather.tex
\lemmaSnakeInBall*
\begin{proof}
    \begin{figure}[htb]
        \hfil%
        \includegraphics[page=1]{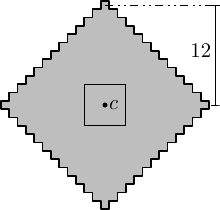}%
        \hfil%
        \includegraphics[page=2]{snake/snake-in-ball}%
        \caption{A cross-section of the $L_1$ ball and the snake it contains.}%
        \label{fig:starter-snake}
    \end{figure}
    Let $C$ be a configuration that contains an $L_1$ ball of radius $12$ which has its center at~$c\in\mathbb{Z}^3$.
    We denote by $P=(c-(0,0,5),\ldots,c+(0,0,5))$ a spine path of length eleven and argue that $S(P)$ (i)  induces a valid snake configuration $S(P)\cap C$ and (ii) corresponds to a free subconfiguration of $C$.
    Recall that by definition, $S(P)\subset N^*_4(P)$.
    Furthermore, every vertex of the spine path is within five units of the center $c$, i.e., $P\subset N^*_5(c)$.
    By triangle inequality, it follows that $S(P)\subset N^*_9(c)$, so the snake subconfiguration is fully contained within the ball.
    As there remain three full outer layers of the $L_1$ ball, $S(P)$ is free.
\end{proof}

\lemmaPush*
\begin{proof}
    Consider a snake subconfiguration ${S}(P)\subseteq C$ and let $v_h\in\mathbb{Z}^3$ be the endpoint of the spine $P$ that defines the snake's head.
    We assume, without loss of generality, that $v_h=(0,0,0)$ and that the direction of the push operation is upward.
    Let the nine skin cells of the snake head at $z=2$ be the \emph{old front} $F_c$, and their upward neighbors the \emph{new front} $F_n$.

    \begin{figure}[htb]%
        \captionsetup[subfigure]{justification=centering}%
        \hfil%
        \begin{subfigure}[t]{\columnwidth/4}%
            \includegraphics[width=\columnwidth]{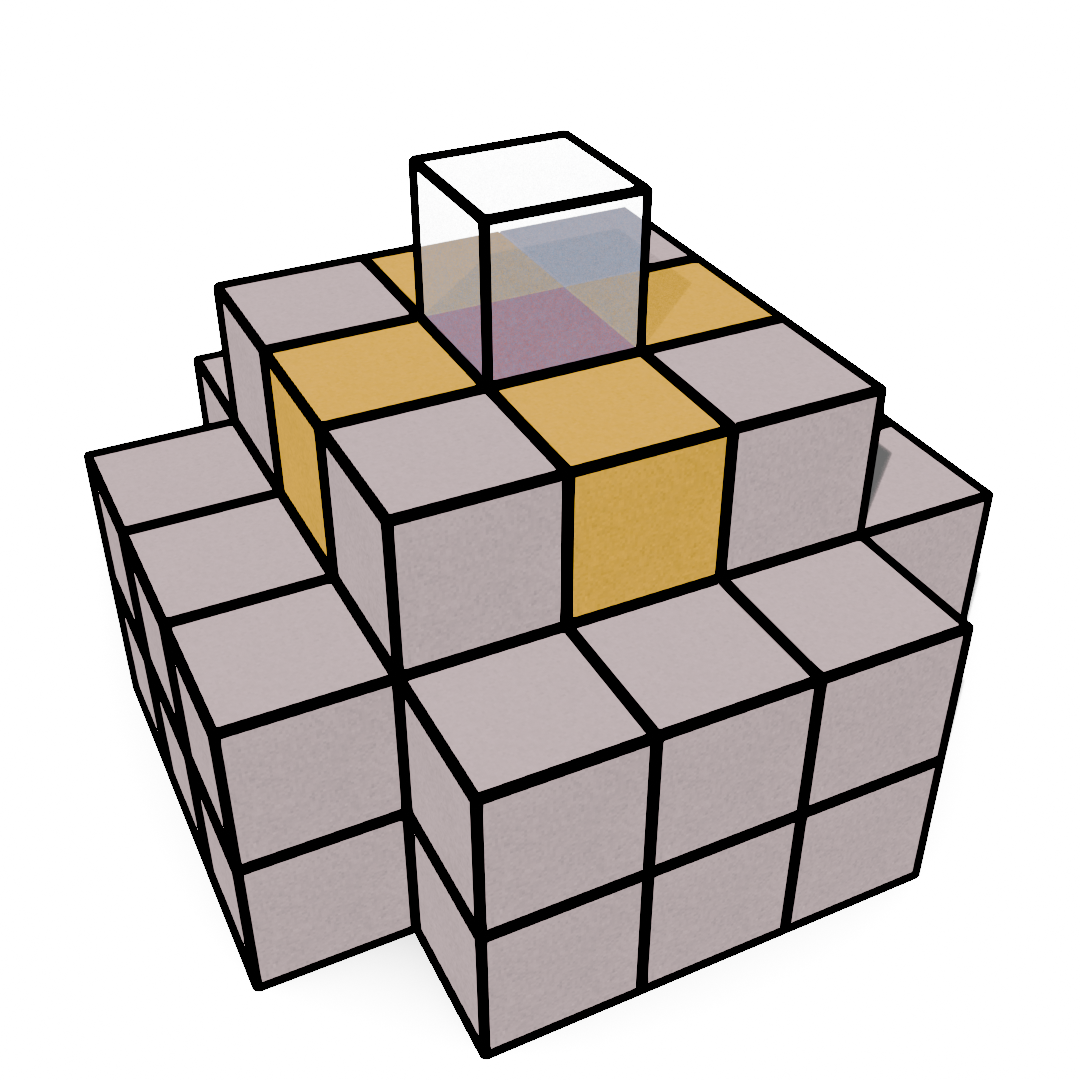}%
            \llap{\includegraphics[width=\columnwidth]{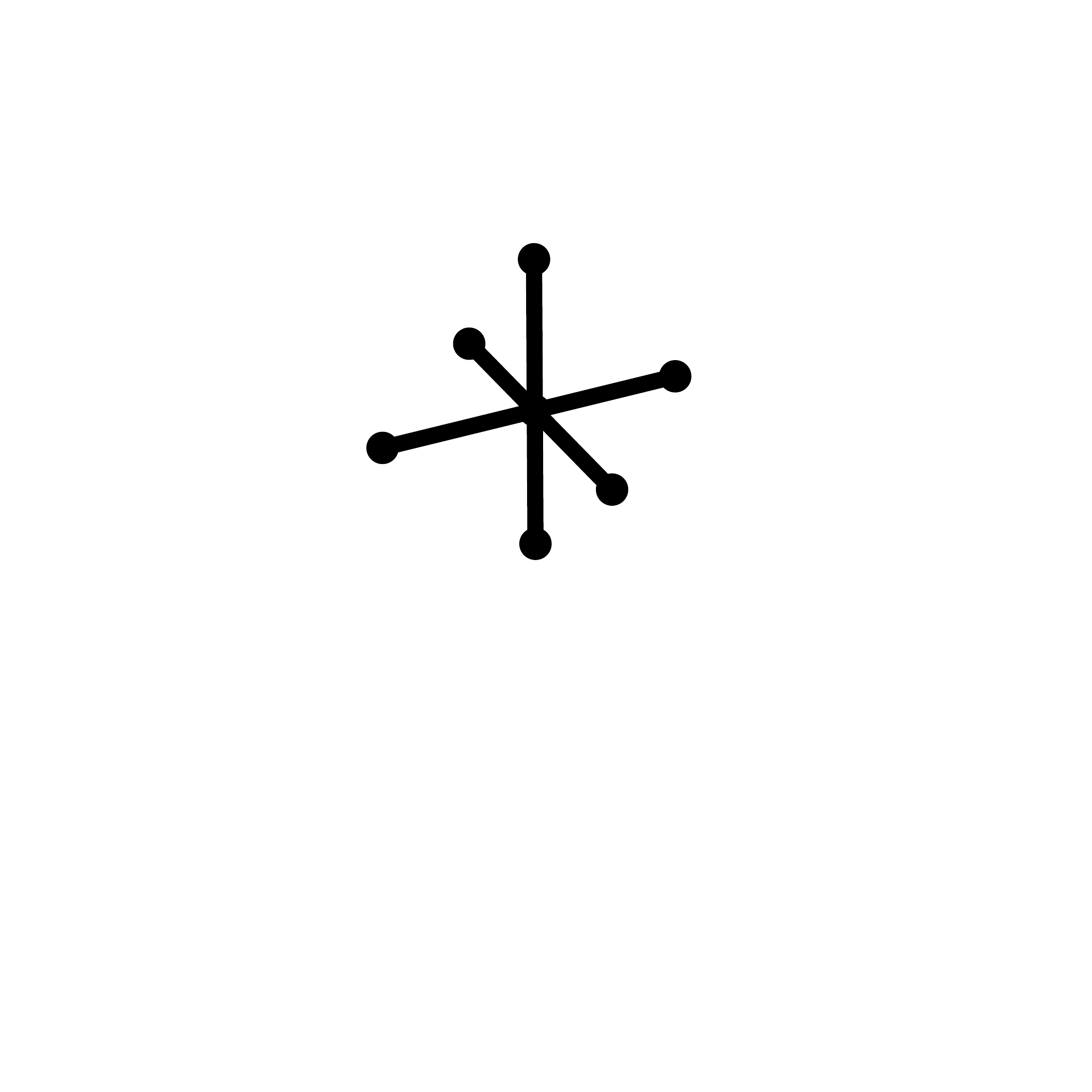}}%
            \subcaption{}
            \label{fig:head-push-center-free}%
        \end{subfigure}%
        \hfil%
        \begin{subfigure}[t]{\columnwidth/4}%
            \includegraphics[width=\columnwidth]{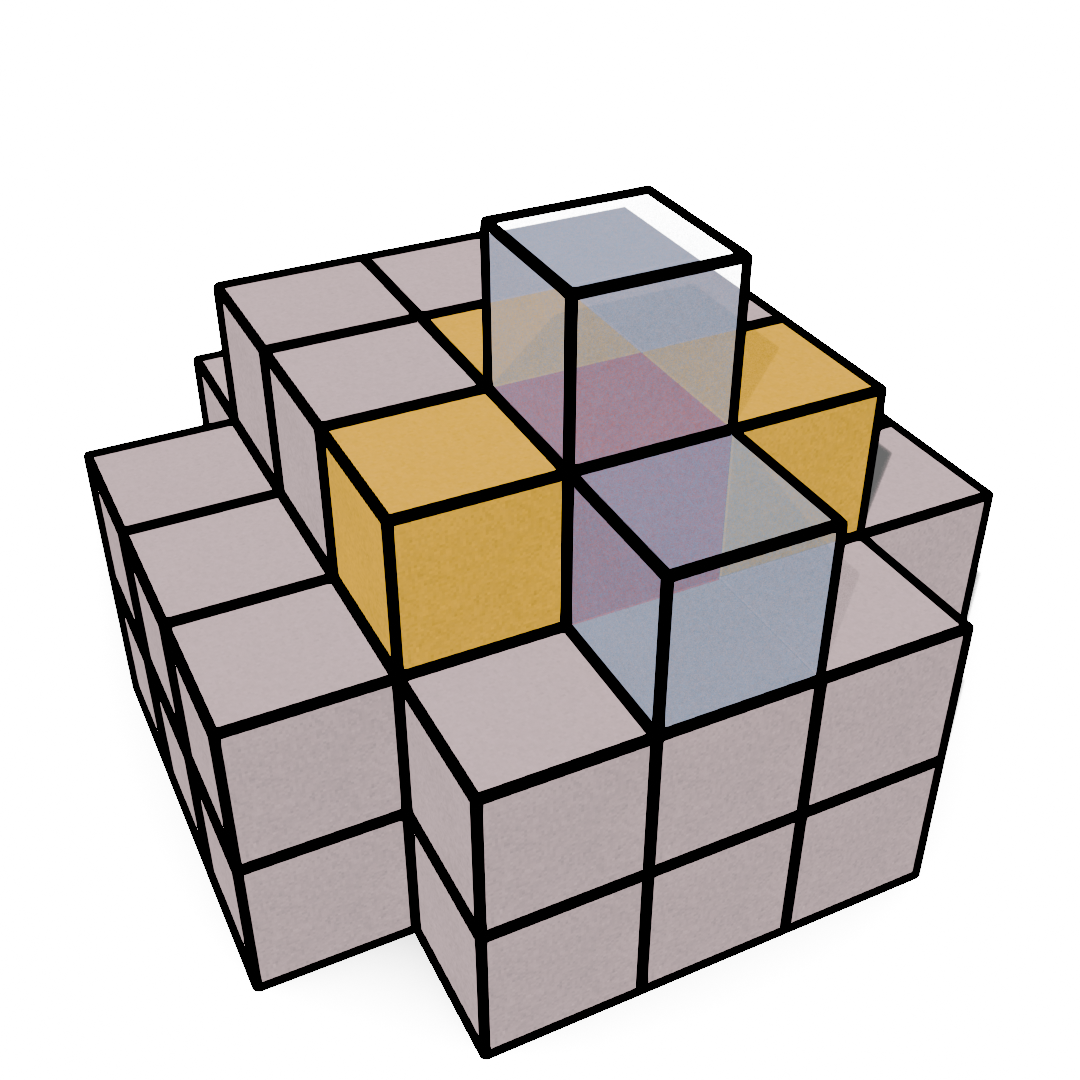}%
            \llap{\includegraphics[width=\columnwidth]{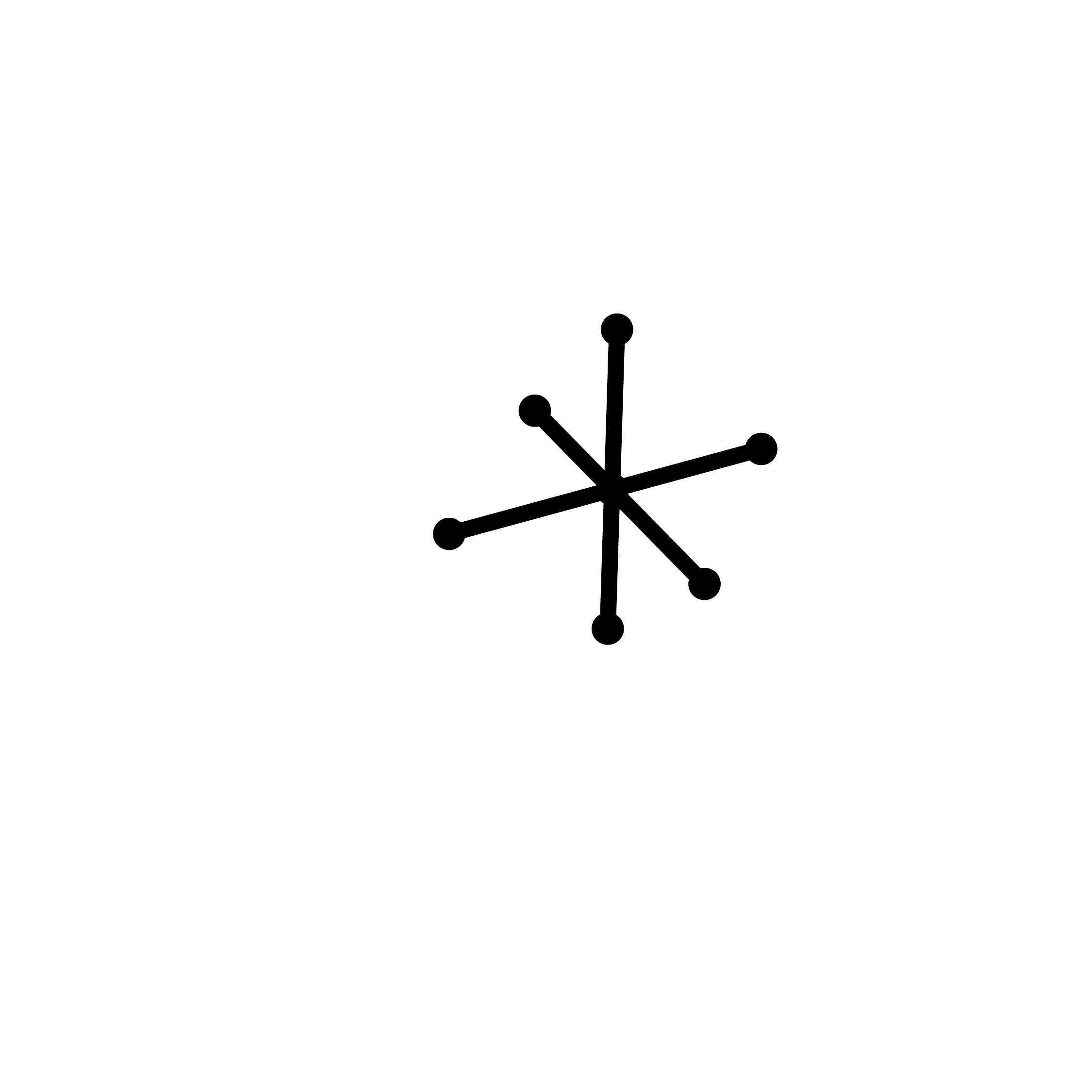}}%
            \subcaption{}
            \label{fig:head-push-edge-free}%
        \end{subfigure}
        \hfil%
        \begin{subfigure}[t]{\columnwidth/4}%
            \includegraphics[width=\columnwidth]{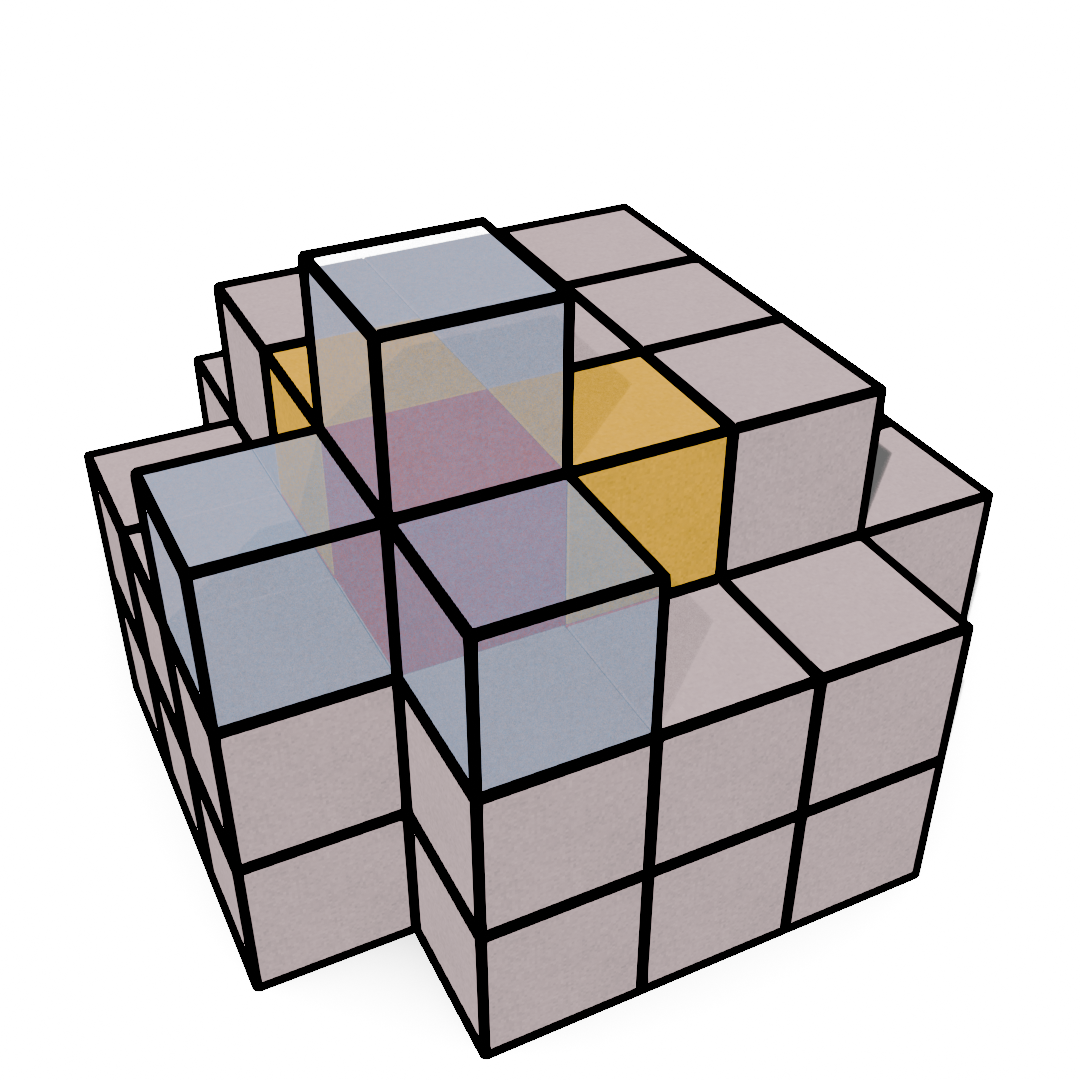}%
            \llap{\includegraphics[width=\columnwidth]{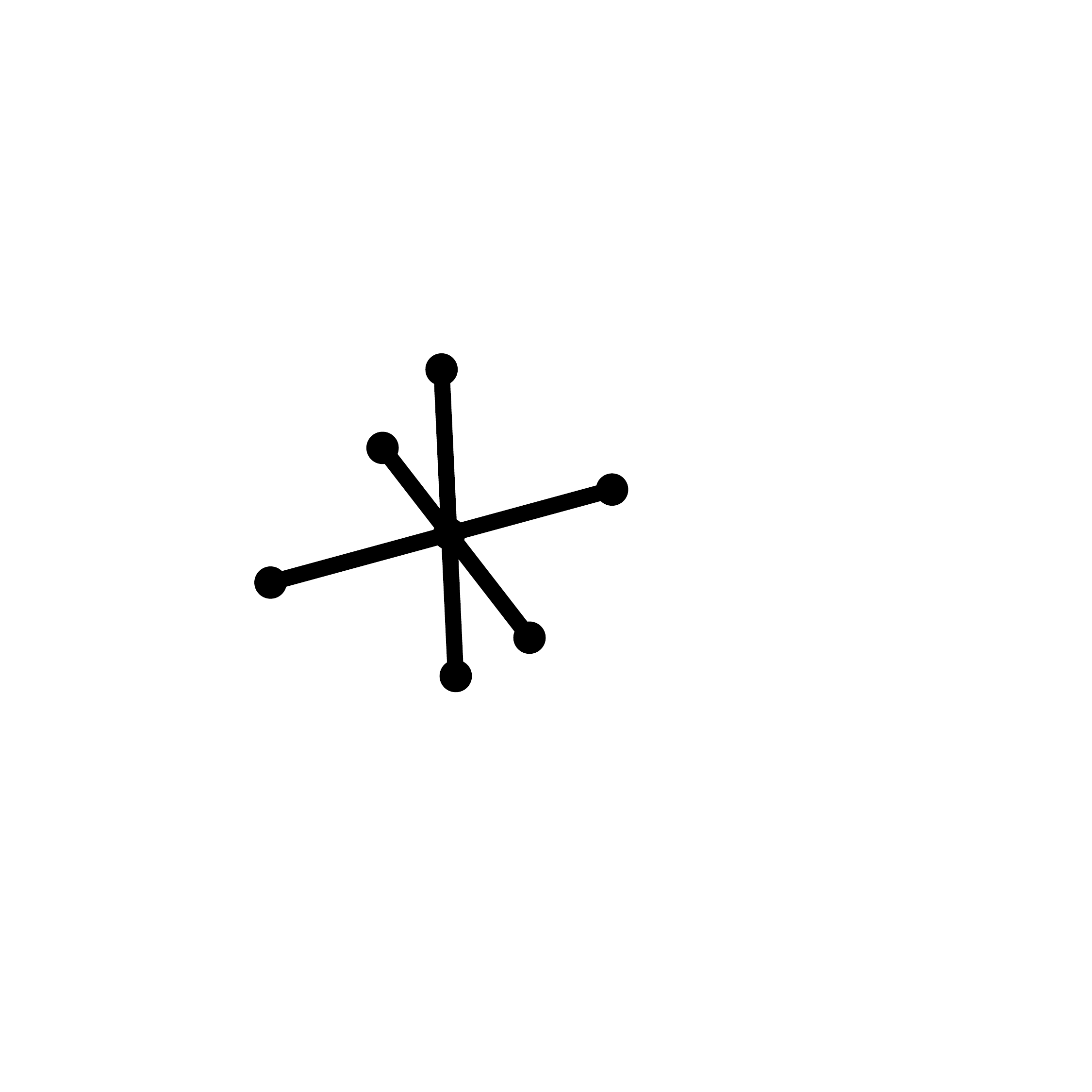}}%
            \subcaption{}
            \label{fig:head-push-corner-free}%
        \end{subfigure}%
        \caption{Cells in the current front are either free or adjacent to a module in the new front:
        For~${\module{m}\in F_c}$ (red), we split $N(\module{m})$ into cells within (yellow) and outside the snake (blue, transparent).}%
        \label{fig:push}
    \end{figure}

    We start by arguing that we can iteratively fill the new front by emptying the eight occupied interior cells of the head at $z(v_h)$ and the module above $v_h$.

    \begin{claim}
        If a module in $F_c$ is not free, the face-adjacent cell in $F_n$ is already occupied.
        \label{clm:push-front-free}%
    \end{claim}
    \begin{claimproof}
        Let $\module{m}\in F_c$; we argue based on its open $L_1$-neighborhood $N(c)$.
        Note that every module in the old front lies on a $4$-cycle in the face-adjacency relation of modules in the head, i.e., the face-adjacency relation is two-vertex-connected.
        In particular, this means that if $\module{m}$ is not free, there has to exist an adjacent module that is not part of the head.

        For $x(\module{m}) = y(\module{m}) = 0$, the claim is trivial: the only cell of $N(\module{m})$ not in the head is its upward neighbor, which is then also the only reason why $\module{m}$ may not be free, see~\cref{fig:head-push-center-free}.
        Assume now that $\module{m}$ is an outer module of the old front, e.g., $x(\module{m}) = 1$ and $y(\module{m}) = 0$; we illustrate the possible cases in~\cref{fig:head-push-edge-free,fig:head-push-corner-free}.
        If $\module{m}$ is not free and its upward neighbor cell is not occupied, there has to exist a module in its east neighbor cell that causes it to be a cut module.
        However, a module in this cell would necessarily have to lie on a common cycle with $\module{m}$ in the face-adjacency relation of the configuration as it is then adjacent to two modules of the head, which itself is two-vertex-connected.
    \end{claimproof}

    To fill a new front cell ${n}\in F_n$, we now take the (up to) four unit tall column of modules below it, down to $z(v_h)$, and perform parallel slides up, see~\cref{fig:push-schedule-a}.
    Note that the topmost module may be unable to perform an upward slide if none of its edge-adjacent cells in the new front are occupied.
    In this case, that module performs a convex transition instead, and then a slide to the desired cell of the new front, as illustrated in~\cref{fig:push-schedule-b}.
    \begin{figure}[htb]
        \captionsetup[subfigure]{justification=centering}%
        \begin{subfigure}[t]{0.4\columnwidth}%
            \centering%
            \includegraphics[page=1]{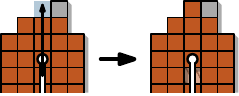}%
            \subcaption{}
            \label{fig:push-schedule-a}%
        \end{subfigure}%
        \begin{subfigure}[t]{0.6\columnwidth}%
            \centering%
            \includegraphics[page=2]{figures/snake/push/push-schedule}%
            \subcaption{}
            \label{fig:push-schedule-b}%
        \end{subfigure}
        \caption{Two-dimensional representations of the schedules used to fill cells of the new front.}
        \label{fig:push-schedule}%
    \end{figure}

    Once this completes, both new and old front are fully occupied.
    We can then apply an analogous argument to~\cref{clm:push-front-free} to argue that the eight outer modules of the old front are free:
    Since the new front is then two-vertex-connected, it cannot cause any module in the old front to be non-free.
    To close the skin of the new head section, those eight modules simply move outward and occupy those cells.
    Using up to four additional modules from the interior of the head section, we can fill the remaining four positions, closing the skin.

    \begin{figure}[htb]
        \captionsetup[subfigure]{justification=centering}%
        \begin{subfigure}[t]{\columnwidth/3 -0.3em}%
            \centering%
            \includegraphics[width=\columnwidth]{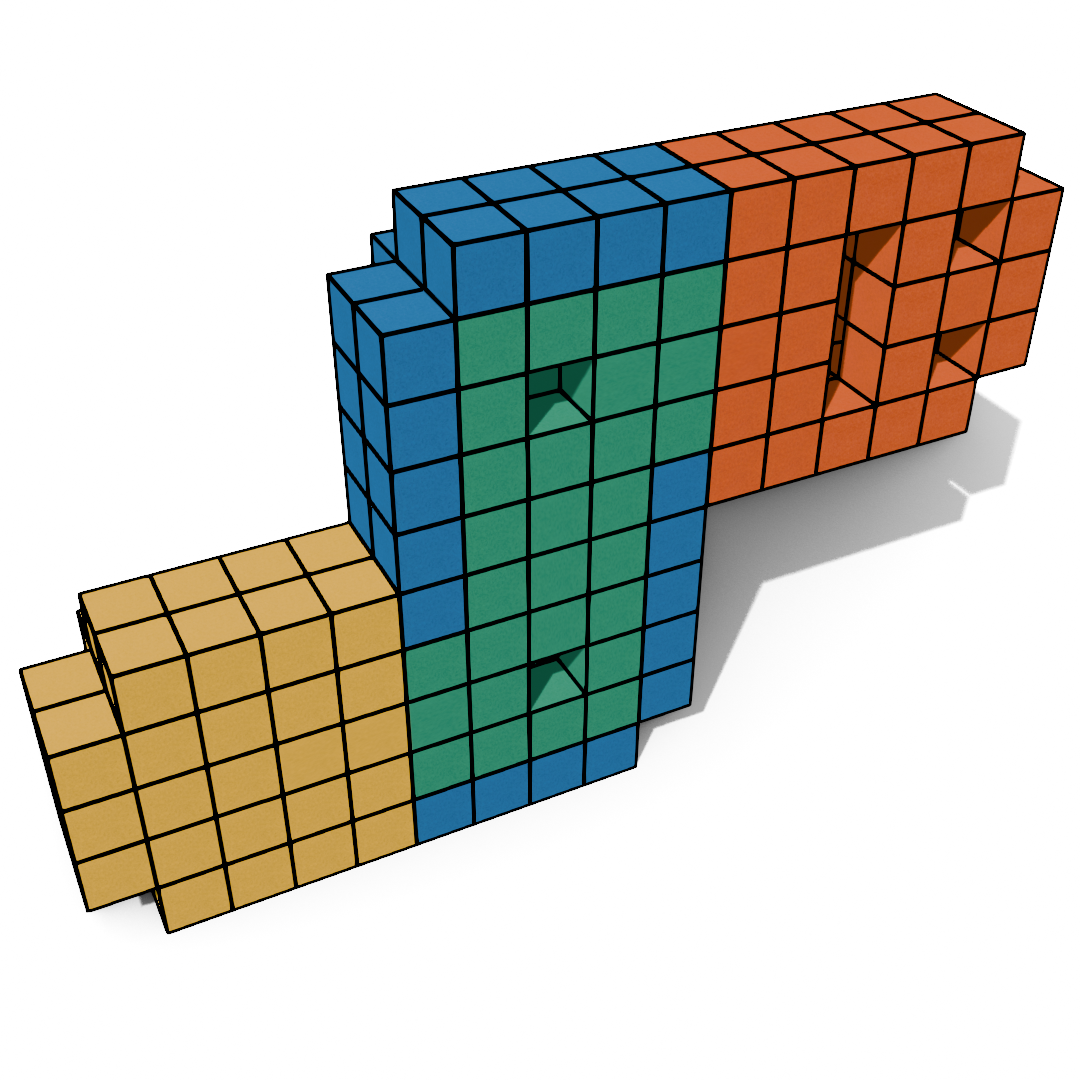}%
            \subcaption{}
        \end{subfigure}%
        \begin{subfigure}[t]{\columnwidth/3 -0.3em}%
            \centering%
            \includegraphics[width=\columnwidth]{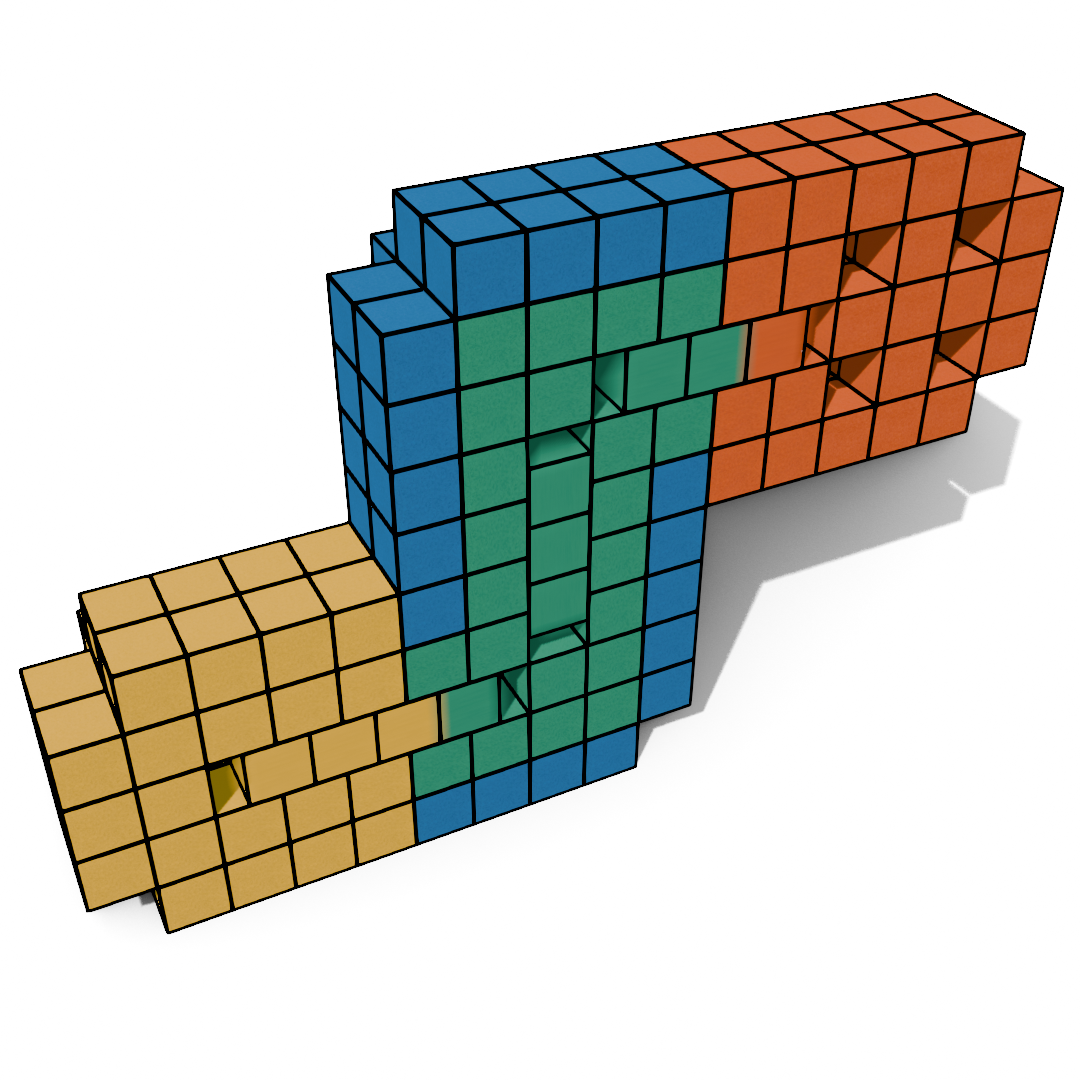}%
            \llap{\includegraphics[width=\columnwidth]{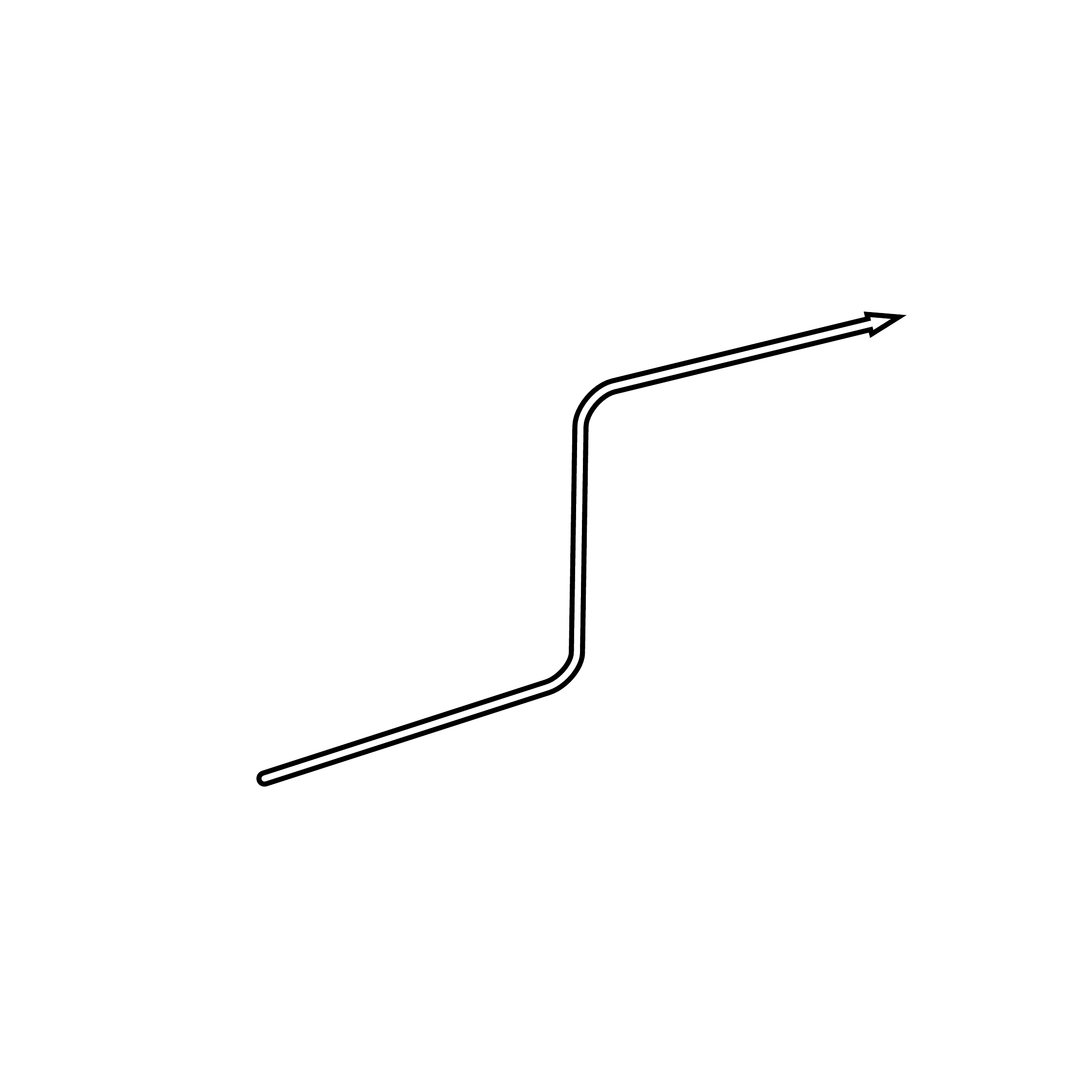}}%
            \subcaption{}
            \label{fig:refilling-head-teleport}
        \end{subfigure}
        \begin{subfigure}[t]{\columnwidth/3 -0.3em}%
            \centering%
            \includegraphics[width=\columnwidth]{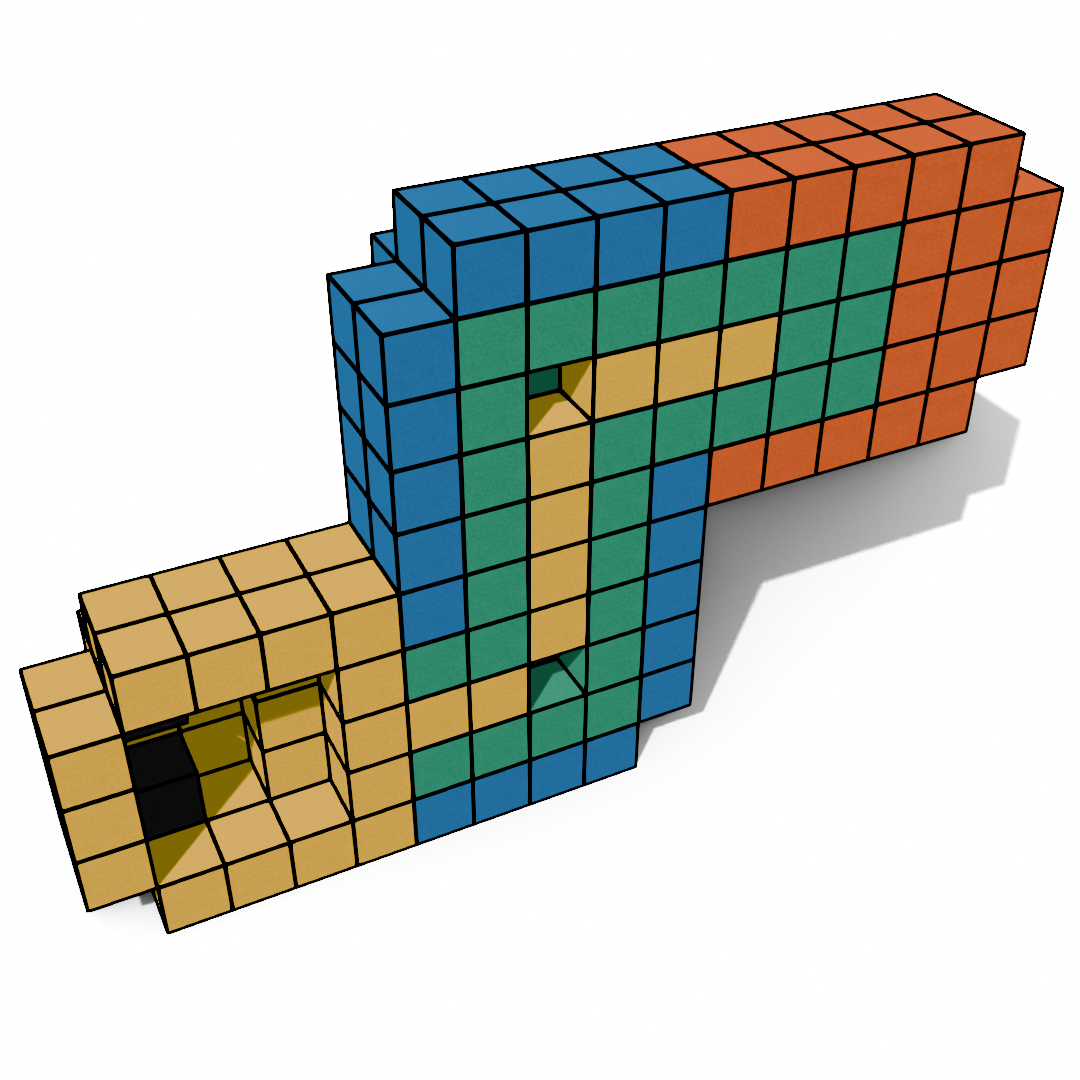}%
            \subcaption{}
        \end{subfigure}
        \caption{To refill empty head cells as shown in (a), we push modules from the tail along the spine path (b) until all interior head cells are occupied, see (c).
        }%
        \label{fig:refilling-head}%
    \end{figure}
    Upon completion of these transformations, there are up to twelve interior modules missing from the new head section.
    To recover a valid snake configuration, we ``teleport'' the appropriate number of modules from the tail section to the new head section using the free space at the major vertices of the spine path to make turns within the configuration.
    Since there are only constantly many modules missing within the head, this takes $\mathcal{O}(1)$ transformations.
\end{proof}

\lemmaSnakeIsReversible*
\begin{proof}
    To achieve this, it suffices to teleport interior modules from the head to the tail section until it is full.
    If the tail section is larger than the head section, this may require joining the head section with the adjacent major vertex's section before teleportation.

    In any case, head and tail section both have constant volume, so teleporting modules along the spine path one at a time as shown in~\cref{fig:refilling-head-teleport} is sufficiently fast.
\end{proof}

\lemmaGrowConsume*

\begin{proof}
    If the cell containing $\module m$ is within the snake's cell space, i.e., $\module m \in S(P)$, we merely need to update its status from held to owned.
    Assume therefore that $\module{m}\notin S(P)$.

    We apply a constant number of fork/join and push/pull operations.
    Let $v_h,\ldots, v, \ldots, v_{t}$ be the spine path $P$ in head to tail order.
    We first fork the snake at $v$, obtaining a snake with spine $P'\coloneqq v,\ldots, v_t$.
    Since $\module m\in N^*_2[v]$, two push operations suffice to transform the forked snake such that it contains $\module{m}$.
    We reverse the orientation of the resulting snake, teleporting interior modules from the modified head to the original tail section at $v_t$.
    Performing constantly many operations toward $v_t$, we can realign the spine of the forked snake with the original snake and perform a join.
\end{proof}

\lemmaDFSTraversal*

\begin{proof}
    The pseudocode in~\cref{alg:snake-dfs} is a heavily simplified representation of how the snake actually traverses the tree.

    In particular, instead of a single spine path $P$, we keep track of a tree $\mathcal{P}$ that represents the union of spine paths that exist due to forking.
    This means that, at any point in time, the spine of the snake we currently consider is defined by a path in $\mathcal{P}$.
    Before the first invocation of the algorithm, we have a valid snake configuration with spine $P$, i.e., $\mathcal{P}=P$.

    Let $v_h\in\mathcal{P}$ now be a chosen head vertex that we seek to extend from the center of one meta-cell to an adjacent meta-cell due to the algorithm.
    Since each push operation requires a specific spine to teleport modules along, we select a maximal path $P$ in $\mathcal{P}$ that starts at $v_h$ and leads to a leaf vertex $v_t$ of the tree.
    If the tail section of the chosen spine is depleted, i.e., it contains insufficiently many modules to fill the head's interior, we can perform a pull operation, shortening that spine path by one unit and removing $v_t$ from the tree $\mathcal{P}$.
    As we chose a maximal path for the spine, this will never result in a disconnection.

    By re-evaluating the tree according to the above description before every push operation, we can always find a spine to push and pull along for lines $4$ and $7$ in the algorithm, and forking simply becomes a matter of choosing the maximal path.

    Since successive modules in the sequence are face-adjacent, it follows by induction that whenever a module $\module m_i$ is removed from the input DFS sequence, its successor $\module{m}_{i+1}$ is in an adjacent meta-cell to some vertex of the spine tree $\mathcal{P}$.
    We will therefore be able to find a head vertex for the next push operation that must be performed to reach that meta-cell.

    Finally, the DFS input sequence has length $\mathcal{O}(\abs{T_\module{s}})$ and we only need constantly many push and pull operations before eliminating the next module from the sequence, the final module will have been consumed in $\mathcal{O}(\ell)=\mathcal{O}(\abs{T_\module{s}})$ transformations.
\end{proof}

%% file: A05_4-scaled-compact.tex
\subsection{Reconfiguration between scaled compact configurations}
\label{subsec:compact-to-compact}

In this section, we show that we can go from any compact configuration of meta-modules to any other in $\mathcal{O}(A+h)$ time, where $A$ and $h$ are projection area and extent of one of the two configurations along some axis.
\lemmaCompactReconfiguration*

To achieve this, we start by creating a cuboid within the bounding box of a given compact configuration.
Our technique leverages insights from the two-dimensional setting to efficiently parallelize each step.

\begin{lemma}
    Any $5$-scaled, compact configuration $C$ can be transformed into a cuboid with roughly identical projection area along a chosen axis $\lambda$ by a schedule of makespan $\mathcal{O}(A+h)$.
    \label{lem:compact-to-cuboid}
\end{lemma}
\begin{proof}
    We assume without loss of generality that $\lambda=z$, all other cases are symmetric.
    For~convenience, let $A=\area(z[C])$, and let $h$ refer to the vertical extent of $C$'s bounding box.

    We start by determining a rectangle $R$ of area~${w\times d\approx A}$ that is anchored at the origin and fully contained within the bounding box of~$z[C]$.
    If $A$ does not have integral divisors that allow for an exact match, we may still assume that $A\geq (w-1)d$ and $A\geq (d-1)w$.

    We now repeatedly modify $C$ to reduce the area of its projection outside $R$ by $25$ units, i.e., one meta-cell, using only constantly many transformations each time.
    Consider the $xy$-maximal occupied two-dimensional meta-cell $c$ outside $R$.
    For every meta-module in the three-dimensional column projected onto $c$, we can find an $L$-shaped path in the $xy$-plane that connects it to an unoccupied three-dimensional meta-cell within $R$; this is essentially a two-dimensional problem and follows from~\cite[Lemma 15]{a.akitaya_et_al:LIPIcs.ESA.2025.28}.
    We can then use parallel tunneling paths for each three-dimensional meta-cell, teleporting it to its destination using only constantly many transformations.

    As the symmetric difference of the projection and $R$ is at most $A$, the projection of $C$ will be fully contained in $R$ after $\mathcal{O}(A)$ transformations.
    \begin{figure}[htb]%
        \captionsetup[subfigure]{justification=centering}%
        \begin{subfigure}[t]{\columnwidth/4}%
            \centering%
            \includegraphics[page=1]{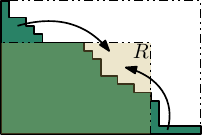}%
            \subcaption{}%
        \end{subfigure}%
        \begin{subfigure}[t]{\columnwidth/4}%
            \centering%
            \includegraphics[page=2]{compact/compact-to-cuboid}%
            \subcaption{}%
        \end{subfigure}%
        \begin{subfigure}[t]{\columnwidth/4}%
            \centering%
            \includegraphics[page=3]{compact/compact-to-cuboid}%
            \subcaption{}%
        \end{subfigure}%
        \begin{subfigure}[t]{\columnwidth/4}%
            \centering%
            \includegraphics[page=4]{compact/compact-to-cuboid}%
            \subcaption{}%
        \end{subfigure}%
        \caption{We obtain a cuboid within the chosen rectangle $R$ (yellow) by (b) first rectifying only the projection and then vertical slices along (c) the $x$- and then the (d) $y$-axis, in parallel.}
        \label{fig:compact-to-cuboid}
    \end{figure}

    To balance the resulting configuration into an actual cuboid, we perform a slice-wise balancing operation on all $y$-slices and then all $x$-slices, transforming each into a rectangle of width $w$ or depth $d$, respectively.
    For each axis, this takes either $w+h$ or $d+h$ transformations.
    The result is a cuboid configuration, i.e., it has a rectangular projection along the $z$-axis and the number of modules in every $z$-column varies by at most one.
\end{proof}

\begin{observation}
    Let $B_1$ and $B_2$ be two axis-aligned bounding volumes anchored at the origin.
    Then either (i) one of the two fully contains the other, or (ii) there exists an axis~$\lambda\in\{x,y,z\}$ such that the extent of $B_2$ along $\lambda$ is greater than that of $B_1$, and the extent of $B_2$ is not greater than that of $B_1$ along either of the other axes.
\end{observation}

\begin{lemma}
    Let $C_1$, $C_2$ be two cuboid configurations of $n$ modules each.
    For some axis~$\lambda$, there exists a schedule of $\mathcal{O}(A+h)$ transformations that transforms $C_1$ into $C_2$, where $A=\max(\area(\lambda[C_1]),\area(\lambda[C_2]))$ is the maximum of the projection areas and $h$ is the maximal extent of the two configurations' bounding boxes along the $\lambda$-axis.
\end{lemma}

\begin{proof}
    Assume without loss of generality, that $\lambda=z$ and that $C_1$ has smaller extent along this axis, but larger extent along both the $x$- and $y$-axes.
    Let $h_1<h_2$ refer to the extent of $C_1$ and $C_2$ along the $z$-axis, respectively, as shown in~\cref{fig:cuboid-to-cuboid}.
    \begin{figure}[htb]%
        \hfil%
        \includegraphics[page=1]{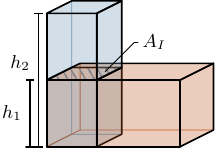}%
        \hfil%
        \includegraphics[page=2]{compact/cuboid-to-cuboid}%
        \caption{We use the 2D algorithm on all layers in parallel to achieve basically this.}
        \label{fig:cuboid-to-cuboid}
    \end{figure}

    Due to~\cref{lem:compact-to-cuboid}, we may assume that all $z$-columns within $C_1$ have identical height (up to a deviation of one unit).
    It suffices to show that in constantly many transformations, we can either (a) move one entire $z$-column of modules from~$C_1$ into $C_2\setminus C_1$, or (b) fill one entire $z$-slice of the volume $C_2\setminus C_1$.

    Let $A_I$ refer to the area of the intersection of the two projections; we distinguish two cases.
    If $A_I>h_1$, we can trivially show that case (a) holds:
    In constant time, we can ``teleport'' all modules in some chosen column of $C_1$ into the bounding volume of $C_2$.
    To achieve this, we pick an L-shaped path in the unique layers of each module that lead into the intersection of the two bounding volumes, and then tunnel straight upwards into the unoccupied volume.
    To avoid collisions, we embed the paths as follows.
    Without loss of generality, let $A_I$ be the rectangle with vertices $(0,y,0), (x,y,0), (x,y,z), (0,y,z)$, which we aim to fill with modules from the column $X\subset C_1$.
    The lowest module in $X$ is assigned an L-shape path to position~$(0,y,0)$, the second-lowest to~$(0,y,1)$, and so on, up to the module assigned to~$(0,y,z)$.
    We then place the next module in position~$(1,y,0)$, and continue filling the row analogously.
    It is easy to see that this set of paths does not cross, and that each path has at most three turns.
    Thus, by pushing along these paths, we can ``instantly'' remove the chosen column and teleport all modules into the desired bounding volume.

    If $A_I<h_1$, are in case (b):
    We pick a subset of the column that has height exactly $A_I$ and only use this portion of the column.
    The same technique for the paths still applies, and we fill a $z$-layer of the symmetric difference completely.
    We can repeat this $A_I / h_1$ many times to fill $\floor{A_I / h_1}$ many layers of the symmetric difference before the column is empty.

    Due to cases (a) and (b), we take $\max(A,h_2)$ time to transfer all modules from the symmetric difference in $C_1$ to the symmetric difference in $C_2$, i.e., reconfigure $C_1$ into $C_2$.
\end{proof}